\documentclass{IEEEtran}
\usepackage{cite}
\usepackage{amsmath,amssymb,amsfonts,amsthm}
\usepackage{algorithmic}
\usepackage{graphicx}
\usepackage{textcomp}
\usepackage{tabularx,colortbl}
\usepackage{subcaption}
\usepackage{ifthen}
\usepackage{textcomp}
\usepackage{hyperref}
\usepackage{subcaption}
\usepackage{siunitx}
\usepackage{url}
\usepackage{graphicx}
\usepackage{color}
\usepackage{placeins}
\usepackage{float}

\newtheorem{proposition}{Proposition}

\def\BibTeX{{\rm B\kern-.05em{\sc i\kern-.025em b}\kern-.08em
    T\kern-.1667em\lower.7ex\hbox{E}\kern-.125emX}}
\begin{document}

\title{Attenuation Compensation in Lossy Media via the Wave Operator Model}

\author{Tianchen Shao, \IEEEmembership{Graduate Student Member, IEEE}, Zekui Jia, Maokun Li, \IEEEmembership{Fellow, IEEE}, \\
Shenheng Xu, \IEEEmembership{Member, IEEE} and Fan Yang, \IEEEmembership{Fellow, IEEE}.

\thanks{This work is supported by Noncommunicable Chronic Diseases-National Science and Technology Major Project under Grant 2024ZD0530000.}
% \thanks{Tianchen Shao, Maokun Li, Shengheng Xu and Fan Yang are with the Department of Electronic Engineering, Beijing National Research Center for Information Science
% and Technology (BNRist), Tsinghua
% University, Beijing 100084, China (e-mail: maokunli@tsinghua.edu.cn).}
\thanks{Tianchen Shao, Maokun Li, Shengheng Xu and Fan Yang are with the Department of Electronic Engineering, Beijing National Research Center for Information Science and Technology (BNRist), 
and State Key laboratory of Space Network and Communications, Tsinghua University, Beijing 100084, China (e-mail: maokunli@tsinghua.edu.cn)}
\thanks{Zekui Jia is with the Elmore Family School of Electrical and Computer Engineering, Purdue University, West Lafayette, IN 47907 USA.}}

\maketitle

\begin{abstract}
The wave operator model provides a framework for modeling wave propagation by encoding material parameter distributions into matrix-form operators. 
This paper extends this framework from lossless to lossy media. 
We present a derivation of the wave operator solution for the electric field in dissipative environments, which can be decomposed into a closed-form propagation term and a non-closed-form dissipation term. 
Based on an analysis of the dominant exponential decay within the propagation term, an attenuation compensation strategy is proposed to restore the attenuated data to an approximate lossless state. 
The performance of this compensation strategy is analyzed and validated through numerical experiments, establishing the theoretical foundation for reduced order model (ROM)-based techniques in lossy media.

\end{abstract}


\section{Introduction}
\label{sec:introduction}

\IEEEPARstart{E}{lectromagnetic} inverse scattering problems (ISPs) have wide applications across various domains, such as geophysical exploration \cite{Bucci2001subsurface}, \cite{guo2020joint}, remote sensing \cite{Moghaddam2020Soil}, \cite{tan2021precise}, nondestructive evaluation \cite{pastorino2002global}, \cite{wu2019time}, and biomedical imaging \cite{song2020study, abubakar2002imaging, bolomey1990microwave, henriksson2010quantitative, mojabi2009microwave}, etc., to reconstruct distributions of target parameters from measured scattering data \cite{chen2018computational}.
However, ISPs are challenging to solve due to their inherent high nonlinearity and ill-posedness, arising from the complex interactions between waves and inhomogeneous media, and insufficient physical information in limited measurements \cite{chew1999waves}, \cite{mueller2012linear}.
In lossy media, the reconstruction of buried targets also suffers from the attenuation of internal waveforms \cite{Catapano2008buried, cui2000novel}.

A variety of methods have been developed to address ISPs, which can be broadly categorized into qualitative and quantitative methods. 
Qualitative methods, such as the linear sampling method (LSM) \cite{colton1996simple, crocco2012linear, bevacqua2025improved}, back-projection (BP) algorithm \cite{haynes2023t}, time-reversal technique \cite{yavuz2008sensitivity, wu2022tailoring}, and the MUSIC algorithm \cite{kirsch2002music}, are noted for their computational efficiency but are limited to providing structural information rather than precise reconstructions. 
In contrast, quantitative methods aim for high-fidelity inversion of material parameters. 
Analytical approaches, such as the Gelfand-Levitan-Marchenko (GLM) method \cite{gel1951determination, slob2016green} and the method of characteristics \cite{habashy1987some, parker1990forward}, transform the governing partial differential equations (PDEs) into more manageable linear equations. 
These methods compute exact solutions for one-dimensional ISPs but struggle with extensions to higher dimensions. 
Various quantitative methods based on numerical optimization have been developed, including the Born iterative method (BIM) \cite{wang1989iterative}, distorted Born iterative method (DBIM) \cite{chew1990reconstruction, wu2025inverse}, contrast source inversion (CSI) \cite{van2001contrast, xu2020deep}, subspace-based optimization method (SOM) \cite{Chen2010subspace, wang2024cross}, iterative multiscaling approach (IMSA) \cite{zhong2020multiresolution, salucci2023recent}, and the contraction integral equation (CIE) \cite{zhong2016new, bevacqua2021effective, bevacqua2024support}. 
These approaches formulate the ISP as a nonlinear optimization problem that is solved with iterative schemes. 
Although these quantitative methods have demonstrated the potential for accurate reconstruction, their practical application is often hindered by heavy computational burden and limited convergence performance.

Recently, the wave operator model has emerged as a framework for addressing ISPs \cite{borcea2024data, borcea2023waveform, borcea2023waveforminternal}.
In this model, parameter distributions are encoded into a wave operator that preserves the propagation characteristics of the governing wave equation \cite{borcea2020reduced}.
In lossless media, leveraging the properties of the self-adjoint wave operator, the electric field admits a time-domain closed-form solution, which thereby enables the construction of reduced order models (ROMs) from data directly \cite{borcea2020reduced, shao2025robust}.
The ROM can be employed to extract single scattering (Born) data from measurements, mitigating artifacts in the inversion result \cite{borcea2018untangling, jia2021enhanced, jia2022linearization, shao2024robust}.
This technique has also been applied to reconstruct layered media with single-input single-output (SISO) data \cite{jia2024estimation}, and has been integrated with the Lippmann-Schwinger-Lanczos (LSL) system in the frequency domain to facilitate target inversion \cite{druskin2021lippmann, druskin2024reduced, druskin2024rom, baker2025regularized}.
However, in lossy media, the wave is attenuated during propagation and the time-reversal symmetry of the wave equation is broken \cite{borcea2021reduced}.
Consequently, these developed ROM-based inversion techniques cannot be applied directly, which highlights the need for extending the wave operator model to lossy environments.

Addressing the challenge of wave attenuation is a crucial step toward extending the ROM to dissipative environments.
In lossy media, wave attenuation causes amplitude decay and phase distortion, leading to a loss of information essential for high-accuracy inversion \cite{yavuz2006full}.
To mitigate these effects and preserve the precision of parameter inversion, various attenuation compensation techniques have been developed, which can be broadly classified into two categories: data-domain and path-based methods \cite{wang2024enhanced}.
Data-domain methods are applied directly to the attenuated data for waveform enhancement, such as time-varying deconvolution, inverse Q-filtering \cite{hargreaves1991inverse, wang2002stable}, time-variant spectral whitening \cite{yilmaz2001seismic}, the short-time Fourier transform (STFT) method \cite{yavuz2005frequency, abduljabbar2016frequency}, and the wavelet transform method \cite{abduljabbar2017continuous, xanthos2021resolution}, which are computationally efficient.
However, these techniques may lack a rigorous physical basis and suffer from numerical instability. 
In contrast, path-based compensation methods are constrained by the governing wave equation and conducted to the wave propagation process, which are more physically interpretable. 
This category, which includes numerical time-reversal algorithm \cite{kosmas2005time},  Q-compensated one-way wave equation migration \cite{mittet1995prestack, zhang2013compensation}, Q-compensated reverse time migration (Q-RTM) \cite{zhang2010compensating, wang2018adaptive} and least-squares reverse time migration \cite{dutta2014attenuation, yang2019viscoacoustic}, can achieve more accurate energy recovery and enhanced stability, but these advantages come at the cost of high algorithmic complexity and reliance on accurate input models.

In this paper, we leverage the wave operator model to develop an attenuation compensation strategy.
Based on our previous work that introduced a ROM-based inversion method for lossless media \cite{shao2025robust}, this paper extends the wave operator model to lossy media.
We begin by deriving the non-closed-form wave operator solution for the electric field, which can be decomposed into a closed-form propagation term and a non-closed-form dissipation term.
Based on the wave operator solution, we propose an efficient attenuation compensation strategy by analyzing the exponential decay within the propagation term.
This strategy amplifies the measured data to its approximate lossless state, and its performance is analyzed and validated through numerical experiments.
This work establishes the necessary theoretical foundation for facilitating the ROM-based parameter inversion procedure in lossy environments, which will be the focus of our future work.

This paper is organized as follows. Section II details the formulation of the wave operator model in lossy media and introduces the attenuation compensation strategy.
Section III presents several numerical examples to evaluate the performance of attenuation compensation.
Conclusions are presented in Section IV.

\section{Theory}

\subsection{Wave Operator Solution}

In 2D lossy medium, the propagation of the TM mode electric field $E_0(\boldsymbol{x}, t)$, generated by a time-even-symmetric source $f_0(t)$ with compact support $(-t_0, t_0)$ located at $\boldsymbol{x_s}$, is governed by the second-order wave equation \cite{chew1999waves}
\begin{equation}
  (\partial_t^2 + \frac{\sigma(\boldsymbol{x})}{\epsilon(\boldsymbol{x})}\partial_t - \frac{1}{\epsilon(\boldsymbol{x})}\nabla\cdot\frac{1}{\mu(\boldsymbol{x})}\nabla) E_0(\boldsymbol{x}, t) = \partial_t f_0(t)\delta_{\boldsymbol{x_s}}(\boldsymbol{x}).
  \label{Wave Equation 1}
\end{equation}
Here $\sigma(\boldsymbol{x})$, $\epsilon(\boldsymbol{x})$, and $\mu(\boldsymbol{x})$ represent the conductivity, permittivity and permeability distributions, respectively.
The Dirac function $\delta(\boldsymbol{x}-\boldsymbol{x_s})$ is abbreviated as $\delta_{\boldsymbol{x_s}}(\boldsymbol{x})$.

To facilitate the analysis, the Liouville transformation \cite{borcea2019robust} and initial condition is introduced.
By converting the field as $E=\sqrt{\epsilon(\boldsymbol{x})}E_0$ and the source as $f = \sqrt{\epsilon(\boldsymbol{x})}f_0$, the wave equation \eqref{Wave Equation 1} is transformed into the following form:
\begin{equation}
  \begin{aligned}
    (\partial_t^2 + \frac{\sigma(\boldsymbol{x})}{\epsilon(\boldsymbol{x})}\partial_t + A) E(\boldsymbol{x}, t) &= \partial_t f(\boldsymbol{x}, t) \delta_{\boldsymbol{x_s}}(\boldsymbol{x}), &t\in\mathbb{R}, \\
    E(\boldsymbol{x}, t) &= 0, \quad &t<0,
  \end{aligned}
  \label{Wave Equation 2}
\end{equation}
% \begin{equation}
%   (\partial_t^2 + \frac{\sigma(\boldsymbol{x})}{\epsilon(\boldsymbol{x})}\partial_t + A) E(\boldsymbol{x}, t) = \partial_t f(\boldsymbol{x}, t) \delta_{\boldsymbol{x_s}}(\boldsymbol{x}), t\in\mathbb{R} 
%   \label{Wave Equation 2}
% \end{equation}
where the self-adjoint wave operator $A$ is defined as
\begin{equation}
  A = -\frac{1}{\sqrt{\epsilon(\boldsymbol{x})}}\nabla\cdot\frac{1}{\mu(\boldsymbol{x})}\nabla\frac{1}{\sqrt{\epsilon(\boldsymbol{x})}}.
  \label{Wave Operator}
\end{equation}
This wave operator preserves the propagation characteristics of the field and exhibits a quadratic relationship with $\eta=1 / \sqrt{\epsilon(\boldsymbol{x})}$.
% maintaining a quadratic relationship with $\eta=1 / \sqrt{\epsilon(\boldsymbol{x})}$.

The solution to \eqref{Wave Equation 2} can be formulated using Green's function $G(\boldsymbol{x}-\boldsymbol{x_s}, t)$, which is the impulse response of the system satisfying:
\begin{equation}
  \begin{aligned}
    (\partial_t^2 + p(\boldsymbol{x})\partial_t + A) G(\boldsymbol{x}-\boldsymbol{x_s}, t) &= \partial_t \delta(t) \delta_{\boldsymbol{x_s}}(\boldsymbol{x}), &t\in\mathbb{R}, \\
    G(\boldsymbol{x}-\boldsymbol{x_s}, t) &= 0, \quad &t<0.
  \end{aligned}
  \label{Green Function 1}
\end{equation}
Here, we introduce the dissipation parameter $p(\boldsymbol{x})=\frac{\sigma(\boldsymbol{x})}{\epsilon(\boldsymbol{x})}$ to represent the distribution of media loss.
The solution for $G(\boldsymbol{x}-\boldsymbol{x_s}, t)$ depends on the spatial distribution of $p(\boldsymbol{x})$, as stated in the following proposition.

\begin{proposition}
The solution to the Green's function in \eqref{Wave Equation 1} is determined by the spatial distribution of the dissipation parameter $p(\boldsymbol{x})$.

\textbf{1) Uniform Dissipation:} For a medium with a spatially uniform dissipation parameter $p(\boldsymbol{x})=p$, where $p$ is a constant (i.e., $\sigma(\boldsymbol{x})\propto\epsilon(\boldsymbol{x})$), the Green's function has an exact closed-form solution:
\begin{equation}
 G(\boldsymbol{x}, t) = e^{-\frac{p}{2}t}\cos(t\sqrt{A'})H(t)\delta_{\boldsymbol{x_s}}(\boldsymbol{x}),
 \label{Green Function Solution 1}
\end{equation}
where $A'=A-p^2/4$ and $H(t)$ is the Heaviside step function.

\textbf{2) Non-Uniform Dissipation:} In a medium with non-uniform dissipation, $p(\boldsymbol{x})=p_0+\Delta p(\boldsymbol{x})$, the uniform solution \eqref{Green Function Solution 1} can be used as an approximation, denoted $G_{app}$, by setting $p=p_0$. The resulting approximation error is quantified by the residual from applying the solution \eqref{Green Function Solution 1} in the wave system \eqref{Wave Equation 1}:
\begin{equation}
 \left[\partial_t^2 + p(\boldsymbol{x})\partial_t + A\right]G_{app}(t) \approx \frac{t}{2}\left[\Delta p(\boldsymbol{x}), A'_0\right]\cos(t\sqrt{A'_0})H(t),
 \label{Green Function Solution 2}
\end{equation}
where $A'_0=A-p_0^2/4$. The term on the right is governed by the commutator $\left[\Delta p(\boldsymbol{x}), A'_0\right]=\Delta p(\boldsymbol{x})A'_0-A'_0\Delta p(\boldsymbol{x})$, which measures the degree of non-commutativity between the spatial variation of the loss and the wave operator.
\end{proposition}

% \begin{proposition}
% For a medium with a spatially uniform dissipation parameter $p(\boldsymbol{x})=p$, where $p$ is a constant, i.e. $\sigma(\boldsymbol{x})\propto\epsilon(\boldsymbol{x})$, the Green's function for \eqref{Green Function 1} has an exact closed-form solution
% \begin{equation}
%   G(\boldsymbol{x}, t) = e^{-\frac{p}{2}t}\cos(t\sqrt{A'})H(t)\delta_{\boldsymbol{x_s}}(\boldsymbol{x}),
%   \label{Green Function Solution 1}
% \end{equation}
% where $A'=A-p^2/4$ and $H(t)$ is the Heaviside step function.

% For a medium with a non-uniform dissipation parameter $p(\boldsymbol{x})=p_0+\Delta p(\boldsymbol{x})$, where $p_0$ is a constant background and $\Delta p(\boldsymbol{x})$ is a spatial variation term, applying the solution \eqref{Green Function Solution 1} with $p=p_0$ to model system \eqref{Green Function 1} results in an approximation error:
% \begin{equation}
%   \left[\partial_t^2 + p(\boldsymbol{x})\partial_t + A\right]G(t) \approx \frac{t}{2}\left[\Delta p(\boldsymbol{x}), A'\right]\cos(t\sqrt{A'})H(t),
%   \label{Green Function Solution 2}
% \end{equation}
% where $\left[\Delta p(\boldsymbol{x}), A'\right]=\Delta p(\boldsymbol{x})A'-A'\Delta p(\boldsymbol{x})$ is the commutator, measuring the degree of non-commutativity between operators.
% \end{proposition}

The proof is detailed in Appendix \ref{Appendix Green's Function}.
This proposition indicates that the closed-form solution \eqref{Green Function Solution 1} is a precise representation for media with homogeneous loss distribution.
However, for media with inhomogeneous loss distribution, its accuracy decreases as the variation $\Delta p(\boldsymbol{x})$ and propagation time $t$ increases.
% This proposition demonstrates that for the medium with inhomogeneous $p(\boldsymbol{x})$ distribution, the Green's function \eqref{Green Function Solution 1} is accurate. 
% For inhomogeneous $p(\boldsymbol{x})$ distribution, the Green's function \eqref{Green Function Solution 1} yields accurate characterization of the lossy medium as the variation $\Delta p(\boldsymbol{x})$ decreases.
% Besides, the approximation error accumulates as the wave propagation time $t$ increases.

% It can be written as
% \begin{equation}
%   G(\boldsymbol{x}-\boldsymbol{x_s}, t) = e^{-\frac{p(\boldsymbol{x})}{2}t}\cos(t\sqrt{A'})H(t)\delta_{\boldsymbol{x_s}}(\boldsymbol{x}),
%   \label{Green Function 2}
% \end{equation}
% where $H(t)$ is the Heaviside distribution, $p(\boldsymbol{x})={\sigma(\boldsymbol{x})}/{\epsilon(\boldsymbol{x})}$ is the dissipation term, and $A'=A-\frac{p^2}{4}$.

Using the Green's function from \eqref{Green Function Solution 1}, we calculate the electric field $E(\boldsymbol{x}, t)$ as
\begin{equation}
  \begin{aligned}
    E(\boldsymbol{x}, t) =& \int_{-\infty}^{\infty} G(\boldsymbol{x}-\boldsymbol{x_s}, t-t')f(t')dt' \\
    =& G(\boldsymbol{x}-\boldsymbol{x_s}, t) \ast_t f(t).
  \end{aligned}
  \label{E Field Time}
\end{equation}
Assume $\tilde{G}(t)=e^{\frac{p}{2}t}G(t)=\cos(t\sqrt{A'})H(t)$, then in Fourier domain
\begin{equation}
  \hat{G}(\omega) = \hat{\tilde{G}}(\omega-j\frac{p}{2}).
\end{equation}
We apply the Fourier transform on \eqref{E Field Time} and obtain
\begin{equation}
  \hat{E}(\omega) = \hat{f}(\omega)\hat{\tilde{G}}(\omega-j\frac{p}{2}).
\end{equation}

To calculate $\widehat{E}(\omega)$ in frequency domain, we compute $\hat{\tilde{G}}(\omega)$ by Fourier transformation
\begin{equation}
    \begin{aligned}
        \widehat{\tilde{G}}(\omega) =& \int_{-\infty}^{\infty}\cos(t\sqrt{A'})H(t)e^{-j\omega t}dt \\
        =& \frac{1}{2}\int_{0}^{\infty}(e^{-j\sqrt{A'}t} + e^{j\sqrt{A'}t})e^{-j\omega t}dt \\
        =& -\frac{j}{2(\omega+\sqrt{A'})} - \frac{j}{2(\omega-\sqrt{A'})}  \\
        &+ \left.\frac{je^{-j(\omega+\sqrt{A'})t}}{2(\omega+\sqrt{A'})}\right|_{t=\infty} + \left.\frac{je^{-j(\omega-\sqrt{A'})t}}{2(\omega-\sqrt{A'})}\right|_{t=\infty} \\
        =& \frac{j\omega}{A'-\omega^2}+\frac{\pi}{2}\left[\delta(\omega+\sqrt{A'})+\delta(\omega-\sqrt{A'})\right].
    \end{aligned}
    \label{Fourier Green 1}
\end{equation}
The formulation of Dirac functions $\delta(\omega\pm\sqrt{A'})$ in the last equation is referred to \cite[Appendix B]{shao2024robust}.
Then we get
\begin{equation}
    \begin{aligned}
        \widehat{G}(\omega) =& \widehat{\tilde{G}}(\omega - j\frac{p}{2}) \\
        =& \frac{\pi}{2}\left[\delta(\omega+\sqrt{A'}-\frac{jp}{2})+\delta(\omega-\sqrt{A'}-\frac{jp}{2})\right] \\
        &+ \frac{j\omega+\frac{p}{2}}{A-\omega^2+j\omega p}.
    \end{aligned}
    \label{Fourier Green 2}
\end{equation}
Note that $A'-(\omega-j\frac{p}{2})^2 = A-\omega^2+j\omega p$.

Then we obtain the electric field in Fourier domain:
\begin{equation}
    \begin{aligned}
        \widehat{E}(\omega) =& \hat{f}(\omega)\widehat{G}(\omega) \\
        =& \frac{\pi}{2}\left[\delta(\omega+\sqrt{A'}-\frac{jp}{2})+\delta(\omega-\sqrt{A'}-\frac{jp}{2})\right]\hat{f}(\omega) \\
        &+ \frac{j\omega+\frac{p}{2}}{A-\omega^2+j\omega p}\hat{f}(\omega).
    \end{aligned}
    \label{Fourier E 1}
\end{equation}
The electric field in time domain is derived via inverse Fourier transform
\begin{equation}
    \begin{aligned}
        E(t) =& \frac{1}{2\pi}\int_{-\infty}^{\infty}\widehat{E}(\omega)e^{j\omega t}d\omega \\
        % =& \frac{1}{4}\int_{-\infty}^{\infty}\hat{f}(\omega)\left[\delta(\omega+\sqrt{A'}+\frac{jp}{2})+\delta(\omega-\sqrt{A'}+\frac{jp}{2})\right]e^{j\omega t}d\omega \\
        =& \frac{1}{2\pi}\int_{-\infty}^{\infty}\frac{j\omega+\frac{p}{2}}{A-\omega^2+j\omega p}\hat{f}(\omega)e^{j\omega t}d\omega \\
        &+ \frac{1}{4}\int_{-\infty}^{\infty}\hat{f}(\omega)\delta(\omega+\sqrt{A'}-\frac{jp}{2})e^{j\omega t}d\omega \\
        &+ \frac{1}{4}\int_{-\infty}^{\infty}\hat{f}(\omega)\delta(\omega-\sqrt{A'}-\frac{jp}{2})e^{j\omega t}d\omega \\
        =& \frac{1}{2\pi}\int_{-\infty}^{\infty}\frac{j\omega+\frac{p}{2}}{A-\omega^2+j\omega p}\hat{f}(\omega)e^{j\omega t}d\omega \\
        &+ \frac{1}{4}\hat{f}(\sqrt{A'}-\frac{jp}{2})e^{-j(\sqrt{A'}-j\frac{p}{2})t} \\
        &+ \frac{1}{4}\hat{f}(\sqrt{A'}+\frac{jp}{2})e^{j(\sqrt{A'}+j\frac{p}{2})t}. 
    \end{aligned}
    \label{Time E 1}
\end{equation}

From \eqref{Time E 1}, we can observe that the electric field distribution $E(\boldsymbol{x}, t)$ is decomposed into two components:
\begin{equation}
  \begin{aligned}
    E(\boldsymbol{x}, t) =& E(t)\delta_{\boldsymbol{x_s}}(\boldsymbol{x}) \\
    =& E^{Pro}(\boldsymbol{x}, t) + E^{Dis}(\boldsymbol{x}, t).
  \end{aligned}
    \label{Time E 2}
\end{equation}
Here, we denote the closed-form term $E^{Pro}(\boldsymbol{x}, t)$ as the propagation term, and the non-closed-form term $E^{Dis}(\boldsymbol{x}, t)$ as the dissipation term.
These components are expressed as follows:
\begin{equation}
  \begin{aligned}
    E^{Pro}(\boldsymbol{x}, t) =& \frac{1}{4}e^{-\frac{p}{2}t} \hat{f}\left(\sqrt{A'} + j\frac{p}{2}\right)e^{j\sqrt{A'}t}\delta_{\boldsymbol{x_s}}(\boldsymbol{x}) \\
    &+ \frac{1}{4}e^{-\frac{p}{2}t} \hat{f}\left(\sqrt{A'} - j\frac{p}{2}\right)e^{-j\sqrt{A'}t}\delta_{\boldsymbol{x_s}}(\boldsymbol{x}),
  \end{aligned}
  \label{Propagation Solution}
\end{equation}
\begin{equation}
  E^{Dis}(\boldsymbol{x}, t) = \frac{1}{2\pi}\int_{-\infty}^{\infty}\frac{j\omega+\frac{p}{2}}{A-\omega^2+j\omega p}\hat{f}(\omega)e^{j\omega t}\delta_{\boldsymbol{x_s}}(\boldsymbol{x})d\omega.
  \label{Dissipative Solution}
\end{equation}
% This decomposition reveals the propagation and energy dissipation characteristics of traveling waves in lossy media.
This decomposition separates the electric field distribution into a propagating component, which has a structure analogous to the lossless case but is modulated by an exponential decay, and a dissipative component that arises purely from the lossy nature of the medium.

% The electric field $E$ is then calculated as
% \begin{equation}
%   \begin{aligned}
%     E(\boldsymbol{x}, t) =& G(\boldsymbol{x}-\boldsymbol{x_s}, t) \ast_t f(t) \\
%     =& E^{Pro}(\boldsymbol{x}, t) + E^{Dis}(\boldsymbol{x}, t),
%   \end{aligned}
% \end{equation}
% where $E(\boldsymbol{x}, t)$ is decomposed into two components: the closed-form propagation term $E^{Pro}(\boldsymbol{x}, t)$ and the non-closed-form dissipation term $E^{Dis}(\boldsymbol{x}, t)$.
% These components are expressed as follows:
% \begin{equation}
%   E^{Pro}(t) = \frac{1}{4}e^{-\frac{p}{2}t} \sum_{n=0,1} \hat{f}\left[\sqrt{A'} + (-1)^n j\frac{p}{2}\right]e^{(-1)^n j\sqrt{A'}t},
%   \label{Propagation Solution}
% \end{equation}

In pratice, we work with the even-extended electric field:
\begin{equation}
  \begin{aligned}
    \tilde{E}(\boldsymbol{x}, t) =& E(\boldsymbol{x}, t) + E(\boldsymbol{x}, -t) \\
    =& \tilde{E}^{Pro}(\boldsymbol{x}, t) + \tilde{E}^{Dis}(\boldsymbol{x}, t).
  \end{aligned}
  \label{Even Extension}
\end{equation}
The corresponding even-extended propagation term and dissipation term are
% \begin{equation}
%   \tilde{E}^{Pro}(t) = \frac{1}{2}\sum_{n=0,1}\hat{f}\left[\sqrt{A'}+(-1)^n j\frac{p}{2}\right] \cos\left[(\sqrt{A'}+(-1)^n j\frac{p}{2})t\right]
% \end{equation}
\begin{equation}
  \begin{aligned}
    \tilde{E}^{Pro}(\boldsymbol{x}, t) =& \frac{1}{2}\hat{f}(\sqrt{A'}+j\frac{p}{2}) \cos\left[(\sqrt{A'}+j\frac{p}{2})t\right]\delta_{\boldsymbol{x_s}}(\boldsymbol{x}) \\
    &+ \frac{1}{2}\hat{f}(\sqrt{A'}-j\frac{p}{2}) \cos\left[(\sqrt{A'}-j\frac{p}{2})t\right]\delta_{\boldsymbol{x_s}}(\boldsymbol{x}),
    \label{Even Propagation Solution}
  \end{aligned}
\end{equation}
\begin{equation}
  \tilde{E}^{Dis}(\boldsymbol{x}, t) = \frac{p}{2\pi}\int_{-\infty}^{\infty}\frac{(A+\omega^2)\hat{f}(\omega)}{(A-\omega^2)^2+\omega^2 p^2}\cos(\omega t)\delta_{\boldsymbol{x_s}}(\boldsymbol{x})d\omega.
  \label{Even Dissipative Solution}
\end{equation}

As the dissipation term $p(\boldsymbol{x})=p$ approaches to zero, the even-extended electric field \eqref{Even Extension} will converge to the lossless wave operator solution:
\begin{equation}
  \lim_{p\rightarrow0}\tilde{E}(\boldsymbol{x}, t) = \hat{f}(\sqrt{A})\cos(t\sqrt{A})\delta_{\boldsymbol{x_s}}(\boldsymbol{x}),
\end{equation}
which is consistent with previous work \cite{shao2025robust}.

For the measurement, we assume the receiver is located at $\boldsymbol{x_r}$, and the measured data is given by
\begin{equation}
  \begin{aligned}
    D(t) =& \int_{\Omega} \delta^T_{\boldsymbol{x}_r}(\boldsymbol{x}) \tilde{E}(\boldsymbol{x}, t) d\boldsymbol{x} \\
    =& \int_{\Omega} \delta^T_{\boldsymbol{x}_r}(\boldsymbol{x})\left[\tilde{E}^{Pro}(\boldsymbol{x}, t)+\tilde{E}^{Dis}(\boldsymbol{x}, t)\right] d\boldsymbol{x} \\
    =& D^{Pro}(t) + D^{Dis}(t),
  \end{aligned}
  \label{Data Function}
\end{equation}
where $D^{Pro}(t)$ and $D^{Dis}(t)$ are the data components corresponding to the propagation term and dissipation term, respectively.

\subsection{Attenuation Compensation}
According to \cite{shao2025robust}, the ROM of wave operators in the lossless medium can be derived from measured data directly. 
However, in lossy media, the presence of the non-closed-form dissipation term $\tilde{E}^{Dis}(\boldsymbol{x}, t)$ and the complex arguments in the propagation term $\tilde{E}^{Pro}(\boldsymbol{x}, t)$ violate the ROM construction approach. 
To overcome this limitation and enable the ROM-based techniques, we propose an attenuation compensation strategy based on the wave operator solution \eqref{Even Extension}.

Observing the propagation term in (\ref{Propagation Solution}), the primary effect of wave attenuation is governed by the exponential decay factor $e^{-\frac{p}{2}t}$. 
% Here we assume $p$ to be a constant, implying that the conductivity distribution is proportional to the permittivity distribution. 
Since the objective of attenuation compensation is to restore the attenuated data to its approximated lossless state, the influence of the decay can be reversed by applying a corresponding exponential gain. 
We define the compensated data $\tilde{D}(t)$ as:
\begin{equation}
  \tilde{D}(t) = D(t) \cdot e^{\frac{p_{app}t}{2}},
  \label{Data Compensation}
\end{equation}
where $p_{app}$ is an estimated compensation parameter.
The goal for $\tilde{D}(t)$ is to approximate the data that would be measured in a lossless medium with the same permittivity and permeability distributions.
% The compensated data $\tilde{D}(t)$ converges to the lossless data under the same permittivity and permeability distribution, as $p_{app}$ characterizes the true conductivity distribution.

We propose an estimation approach to determine a suitable value for compensation parameter $p_{app}$, which is based on the dissipation fraction defined in the following proposition.
\begin{proposition}
The dissipation fraction is defined as
\begin{equation}
  k(p_1, p_2) = \frac{D(t=0|p=p_1) - D(t=0|p=0)}{D(t=0|p=p_2) - D(t=0|p=0)},
  \label{Dissipation Fraction}
\end{equation}
where $D(t=0|p)$ is synthetic data at $t=0$ for a given dissipation parameter $p$.
Assume the $f(t)$ is a Gaussian-type source.
For small dissipation parameters, where $0 \le p_1 \le p_2 \ll \|A\|$, this dissipation fraction is approximately linear:
\begin{equation}
  k(p_1, p_2) = \frac{p_1}{p_2} + O\left(\frac{p_2^2-p_1^2}{\|A\|}\right).
\label{Dissipation Approximation}
\end{equation}
\end{proposition}

The proof is detailed in Appendix \ref{Appendix Dissipation Approximation}.
This proposition suggests that the dissipation parameter of an unknown parameter can be estimated by comparing its measured data at $t=0$ to that of a reference medium.
We set the reference dissipation parameter to $p_2=1$, and the estimated compensation parameter $p_{app}$ is defined as:
% The dissipation fraction is defined as
\begin{equation}
  \begin{aligned}
      p_{app} =& k(p_{true}, 1) \\
      =& \frac{D(t=0|p=p_{true}) - D(t=0|p=0)}{D(t=0|p=1) - D(t=0|p=0)}, \\
  \end{aligned}
  \label{Loss Estimation}
\end{equation}
where $D(t=0|p=p_{true})$ is the measured data at $t=0$, while $D(t=0|p=0)$ and $D(t=0|p=1)$ are computed by simulation using a reference medium.
As validated by numerical experiments in Section \ref{Numerical Experiments}, this provides an accurate estimate for $p_{app}$ and enables effective compensation via \eqref{Data Compensation}.
% Therefore, the estimated compensation parameter is set as $p_{app} = k(p_{true}, 1)$.

% \subsection{Active media}
% To be filled.

\section{Numerical Experiments}
\label{Numerical Experiments}

In this section, we present three 2-D numerical experiments with various dissipation distributions to validate our proposed attenuation compensation strategy.
The forward modeling is simulated by the finite-difference time-domain (FDTD) method.

\begin{figure}[!htb]
  \centering
  \includegraphics[width=0.28\textwidth]{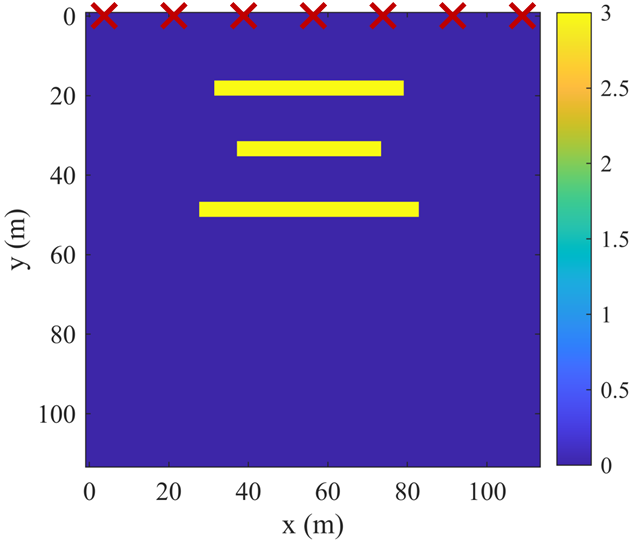}
  \caption{Contrast distribution of the three-stripe target example.
  Yellow stripes represent the target on a homogeneous background, with sensors marked by red crosses.
  The colorbar indicates the contrast value.} 
  \label{Model}
\end{figure}

\begin{figure}[!htb]
  \centering
  \includegraphics[width=0.28\textwidth]{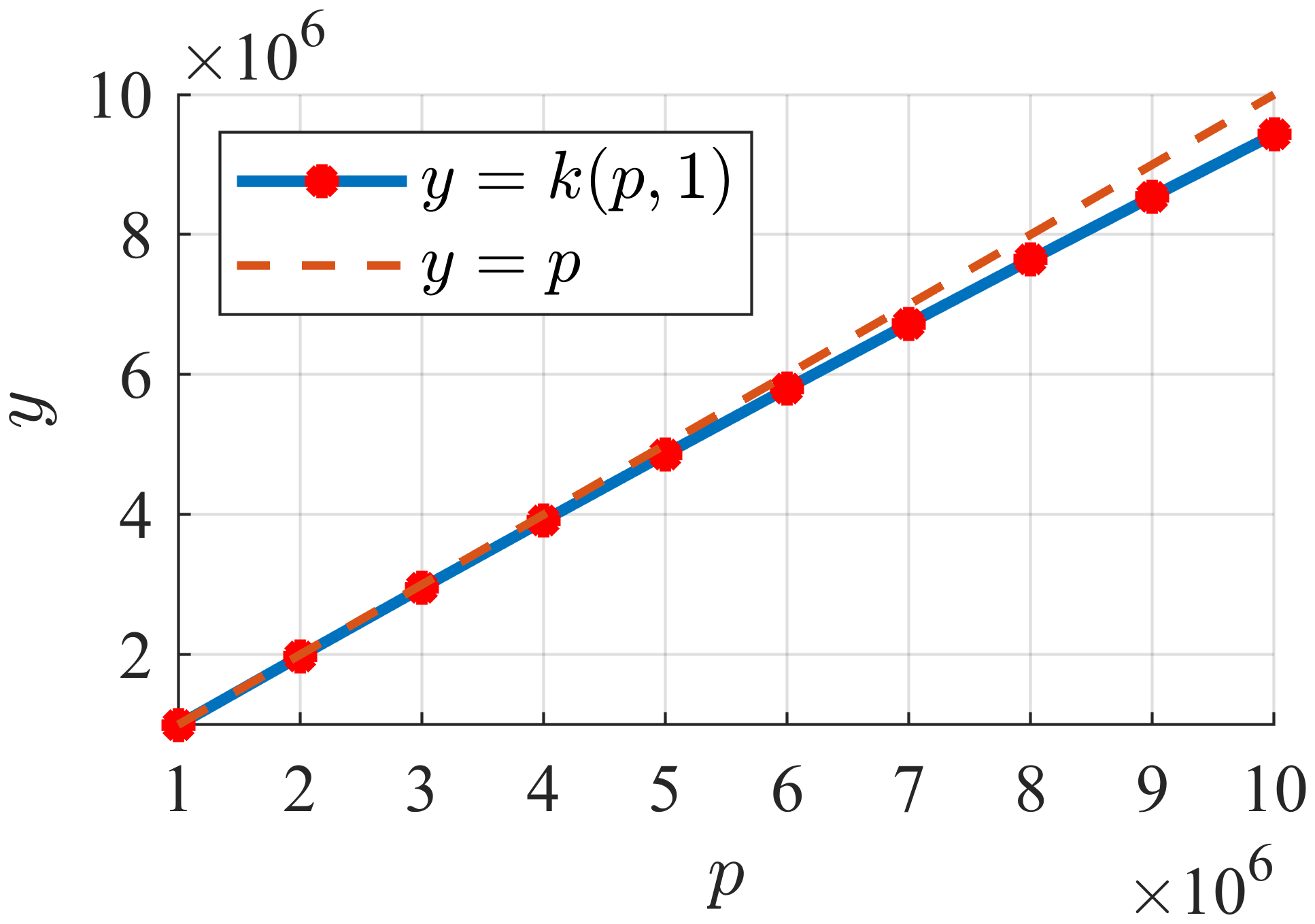}
  \caption{Comparison between $y=k(p, 1)$ and $y=p$.} 
  \label{Trend}
\end{figure}

\subsection{Case 1: Uniform Dissipation Distribution}

The first example investigates a scenario with a uniform loss distribution. 
As depicted in Fig.~\ref{Model}, the model consists of three equidistant horizontal stripes with a relative permittivity of 3, embedded in a background medium characterized by a uniform dissipation parameter of $p(\boldsymbol{x})=3\times10^6$ \si{\per\second}. 
The Domain of Interest (DoI) spans $3\lambda\times3\lambda$, where $\lambda$ is the wavelength corresponding to the central frequency of \SI{8}{M\hertz}, and is discretized into a $60\times60$ grid.
The top boundary is subject to the Neumann boundary condition, while the other boundaries are set as absorbing boundary conditions. 
A linear array of twenty sensors, serving as both transmitters and receivers, is uniformly positioned along the upper boundary of the DoI. 
These sensors emit modulated Gaussian pulses with a central frequency of \SI{8}{M\hertz}, and the scattered data are collected in a multiple-input, multiple-output (MIMO) configuration. 
The time-domain signals are recorded with a sampling interval of \SI{4.7}{n\second} over a total observation window of \SI{0.625}{m\second}.

The validation of our compensation approach begins with a numerical verification of the dissipation fraction \eqref{Dissipation Approximation}. 
By keeping the permittivity fixed and varying the dissipation parameter $p$, the linear relationship $k(p, 1)\approx p$ is confirmed, as shown in Fig.~\ref{Trend}. 
The deviation from linearity at large $p$ values is attributed to the increased attenuation effects, which aligns with theoretical expectations. 
Based on this relationship, the compensation parameter for the model in Fig.~\ref{Model} is estimated using (\ref{Loss Estimation}), yielding $p_{app}=2.94\times10^6$ \si{\per\second}.

\begin{figure}[!htb]
	\centering
	\begin{subfigure}[b]{43mm}
		\includegraphics[width=\textwidth]{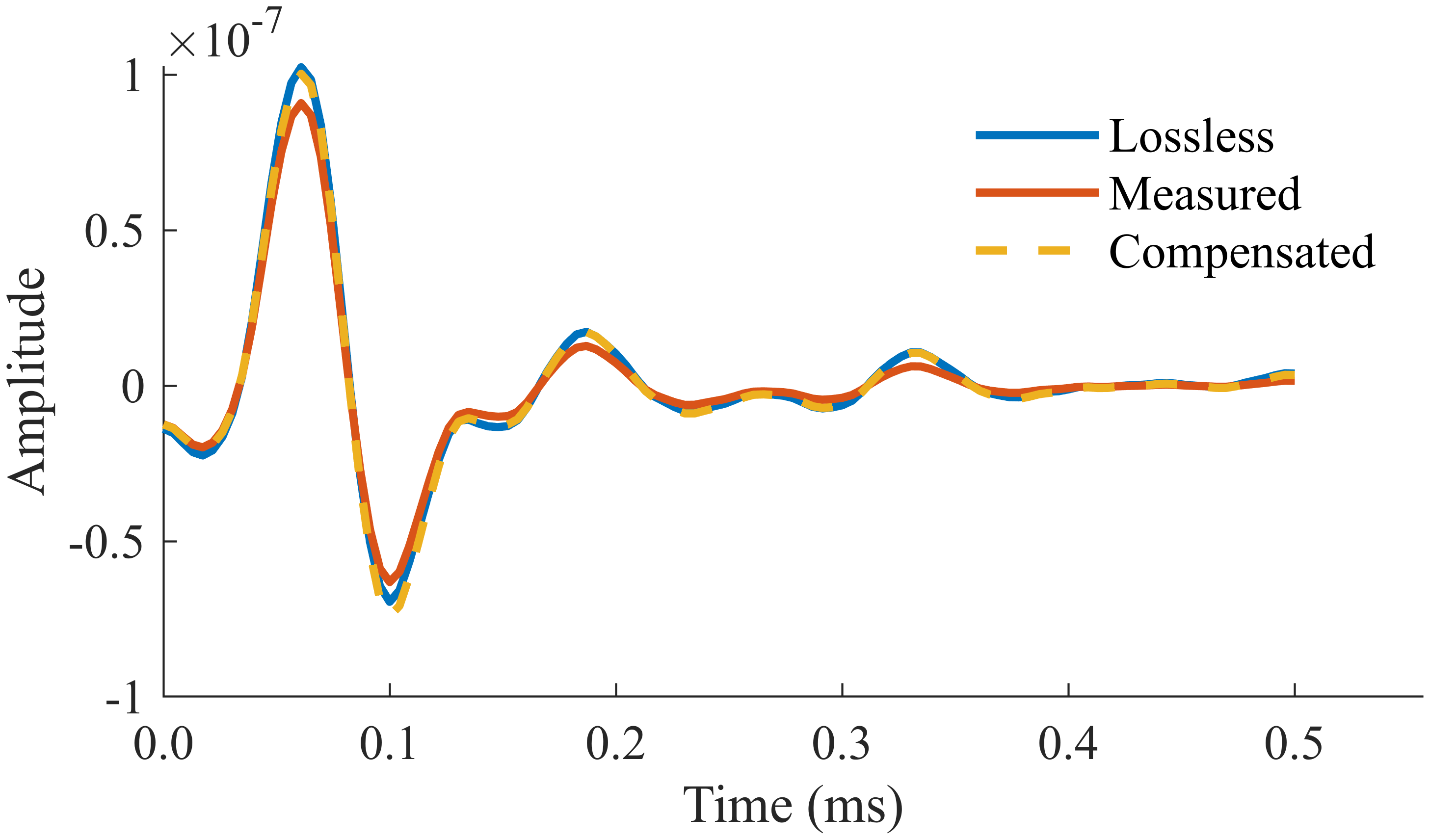}
		\caption{} 
    % \label{Mono Data}
	\end{subfigure}
    \hfill
    \begin{subfigure}[b]{43mm}
		\includegraphics[width=\textwidth]{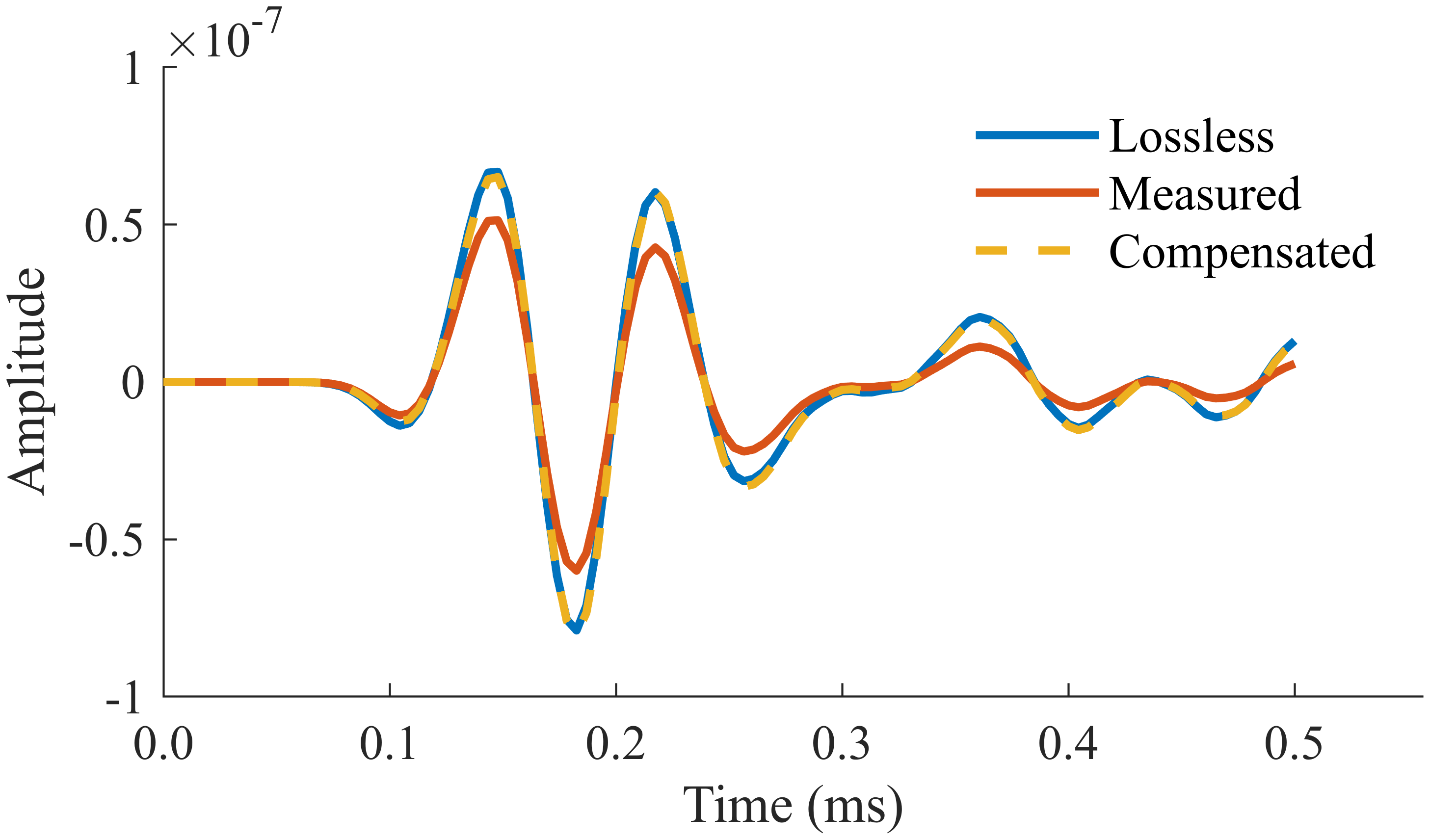}
		\caption{} 
    % \label{Bi Data}
	\end{subfigure}
  \centering
	\begin{subfigure}[b]{43mm}
		\includegraphics[width=\textwidth]{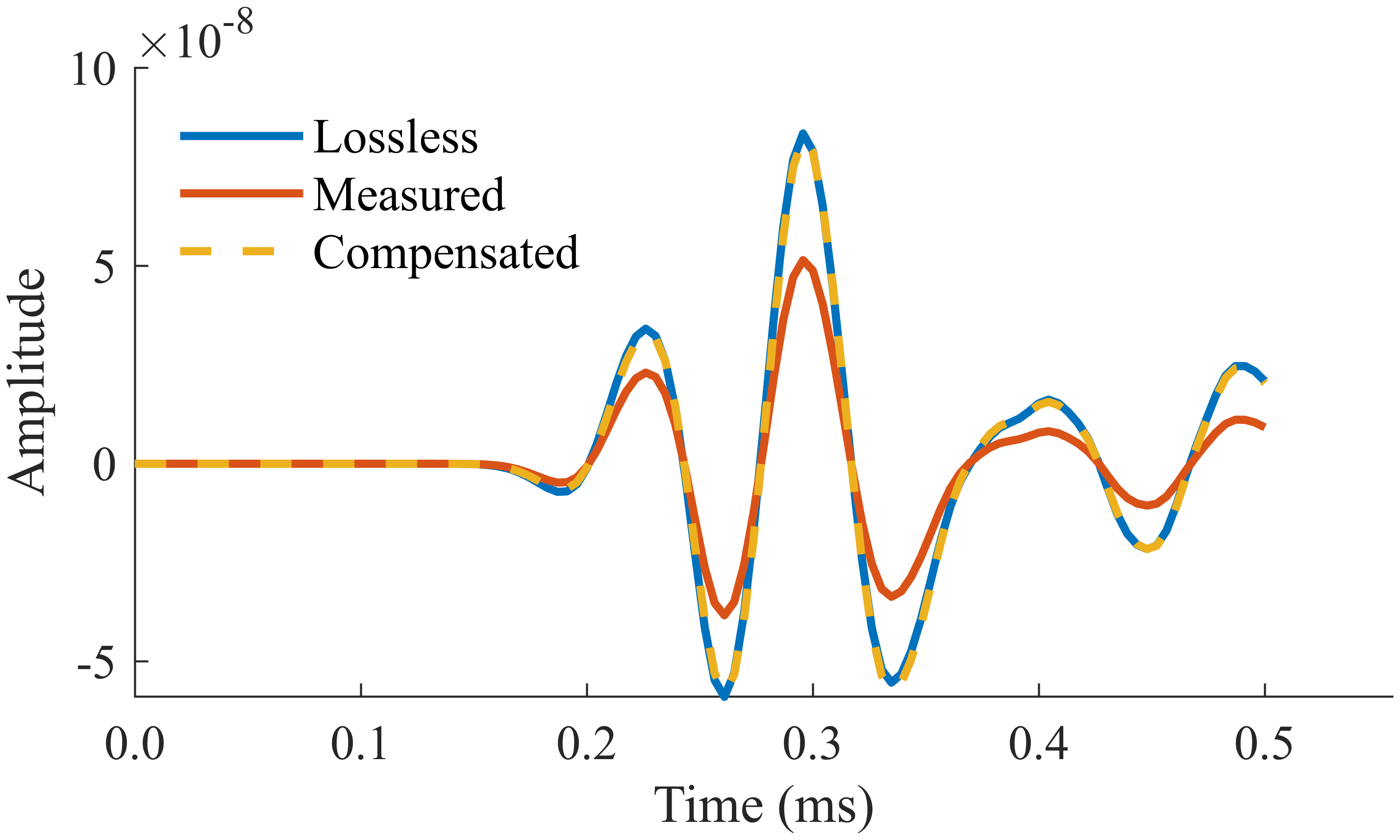}
		\caption{} 
    % \label{Mono Data}
	\end{subfigure}
    \hfill
    \begin{subfigure}[b]{43mm}
		\includegraphics[width=\textwidth]{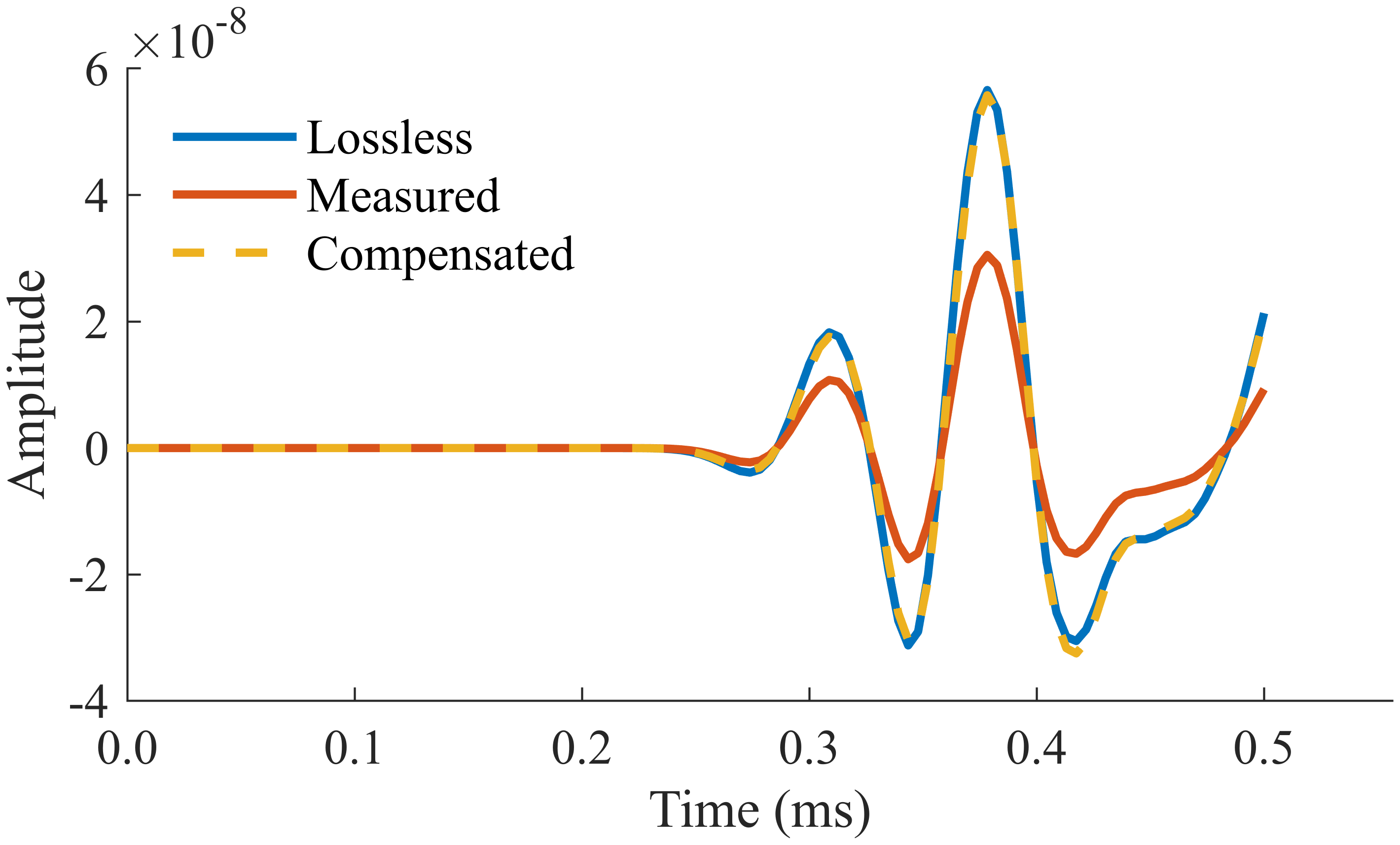}
		\caption{} 
    % \label{Bi Data}
	\end{subfigure}
	\caption{Comparison of lossless data, measured data and compensated data with uniform dissipation distribution. 
  The transmitter is at Sensor 1 and receivers are at (a) Sensor 5, (b) Sensor 10, (c) Sensor 15, (d) Sensor 20.}   
  \label{Data Comparison for Constant Distribution}
\end{figure}

\begin{figure}[!htb]
  \centering
	\begin{subfigure}[b]{43mm}
		\includegraphics[width=\textwidth]{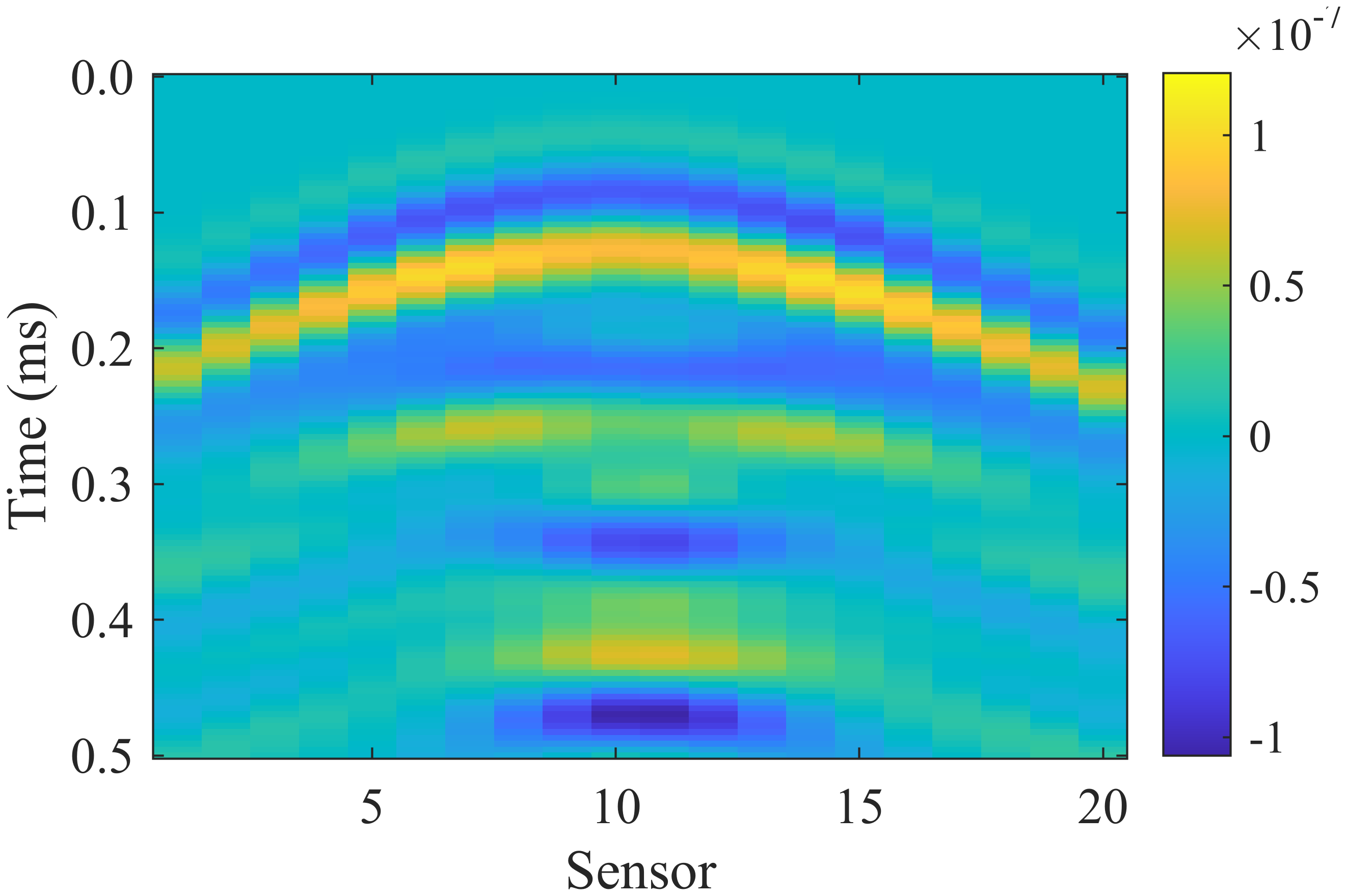}
		\caption{} 
    \label{CSG for Constant Distribution: Lossless} 
    % \label{CSG Raw Data}
	\end{subfigure}
  \hfill
  \begin{subfigure}[b]{43mm}
		\includegraphics[width=\textwidth]{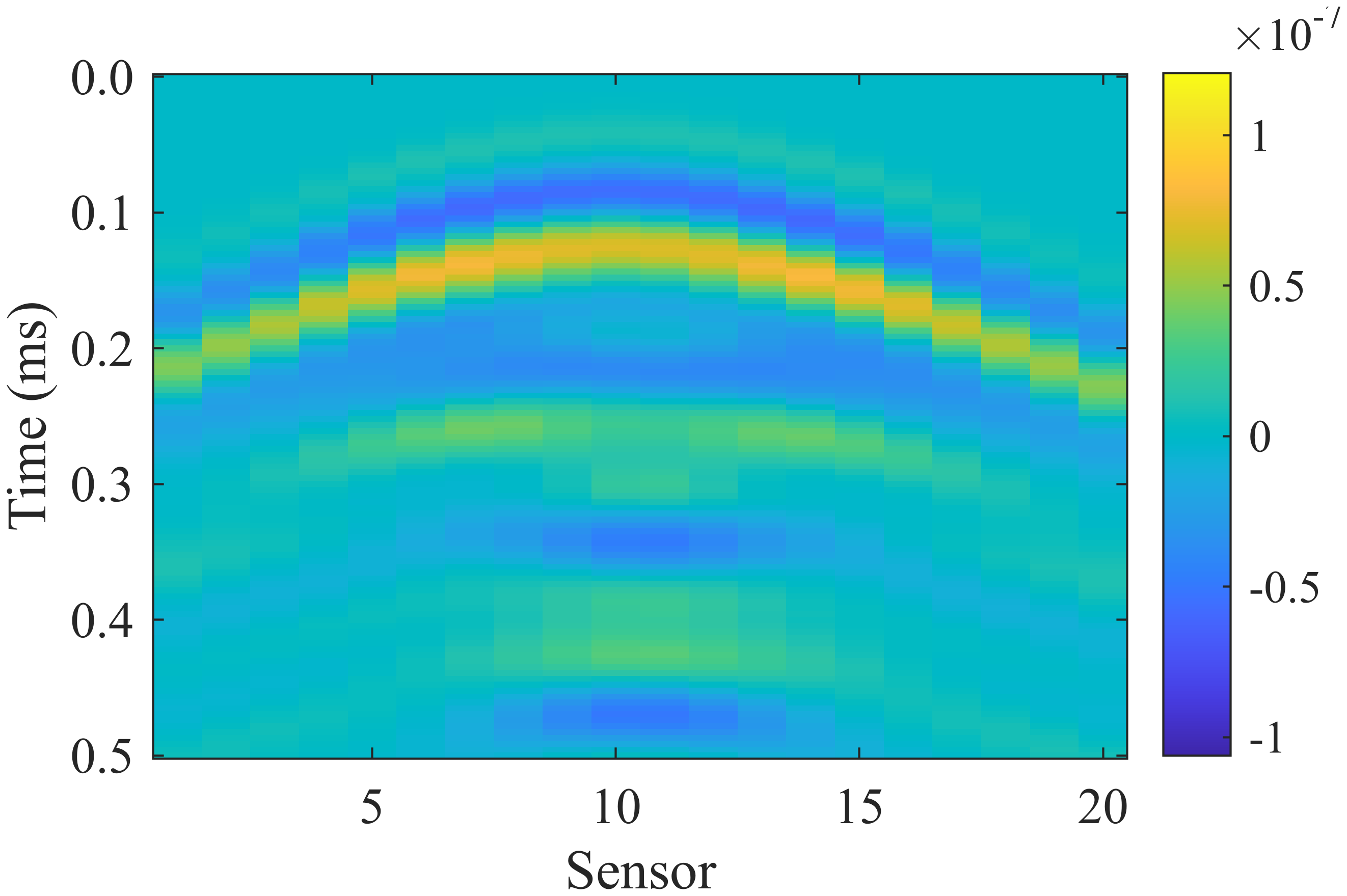}
    \caption{} \label{CSG for Constant Distribution: Lossy}
	\end{subfigure}
  \centering
	\begin{subfigure}[b]{43mm}
		\includegraphics[width=\textwidth]{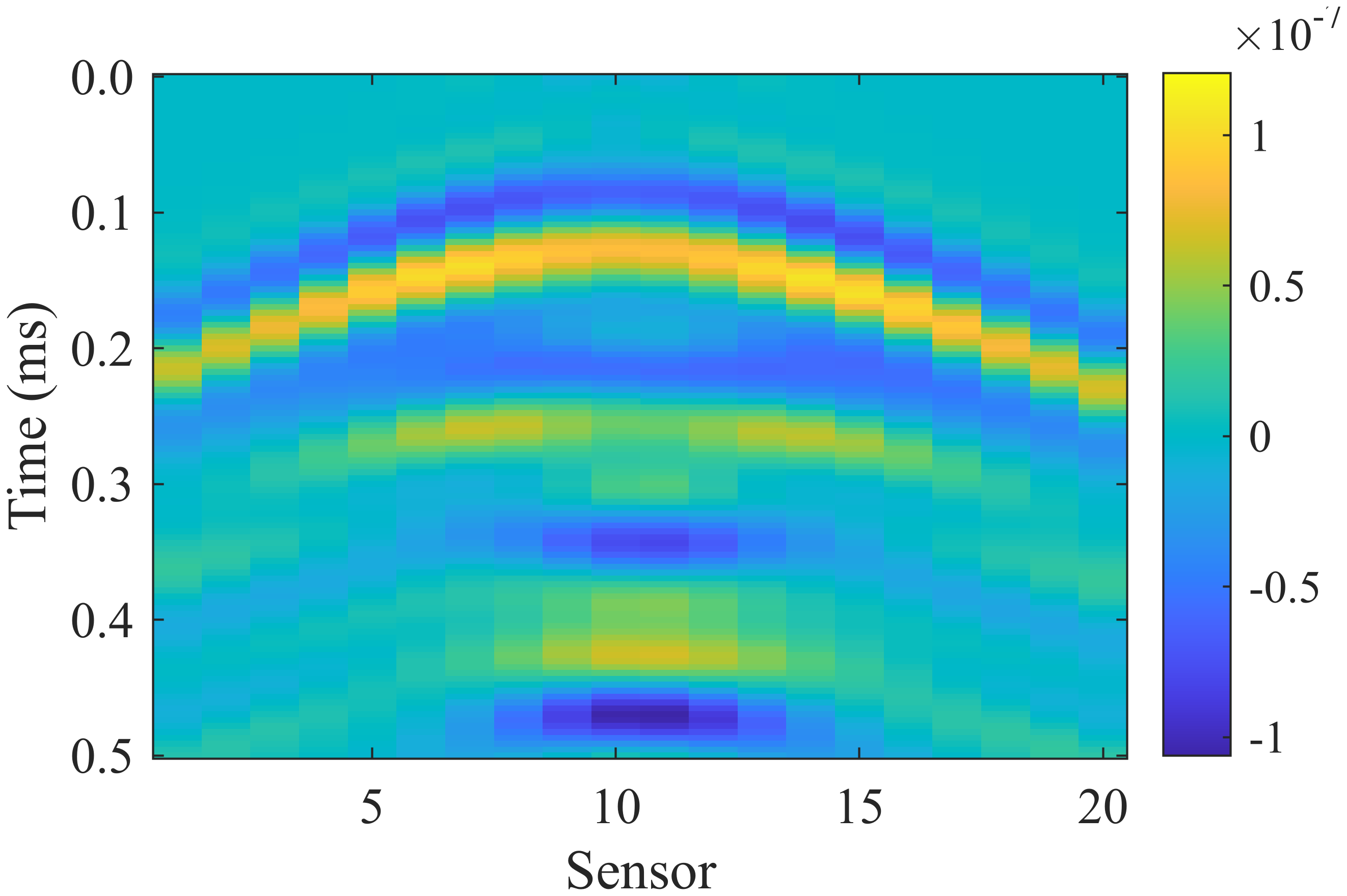}
		\caption{} \label{CSG for Constant Distribution: Comp} 
    % \label{CSG Raw Data}
	\end{subfigure}
	\caption{Common shot gathers at Sensor 10 with uniform dissipation distribution. 
  The colorbar indicates the magnitude of the collected data. 
  (a) Synthetic data in lossless environment. (b) Measured data in lossy environment. (c) Compensated data.} 
  \label{CSG for Constant Distribution}
\end{figure}

The effectiveness of this compensation method is demonstrated in Fig.~\ref{CSG for Constant Distribution}, which compares the lossless, measured, and compensated bistatic data traces for a fixed transmitter at Sensor 1 and receivers at Sensors 5, 10, 15, and 20. 
The compensation method accurately restores the attenuated waveforms to the level of their lossless counterparts, yielding low relative data misfits in the $L2$-norm of only 4.86\%, 3.69\%, 3.40\%, and 4.38\% for these four traces, respectively. 
This is consistent across all collected data, leading to a low overall relative data misfit of only 3.66\% between the entire sets of compensated and ideal lossless data.

\begin{figure}[!htb]
  \centering
  \begin{subfigure}[b]{43mm}
		\includegraphics[width=\textwidth]{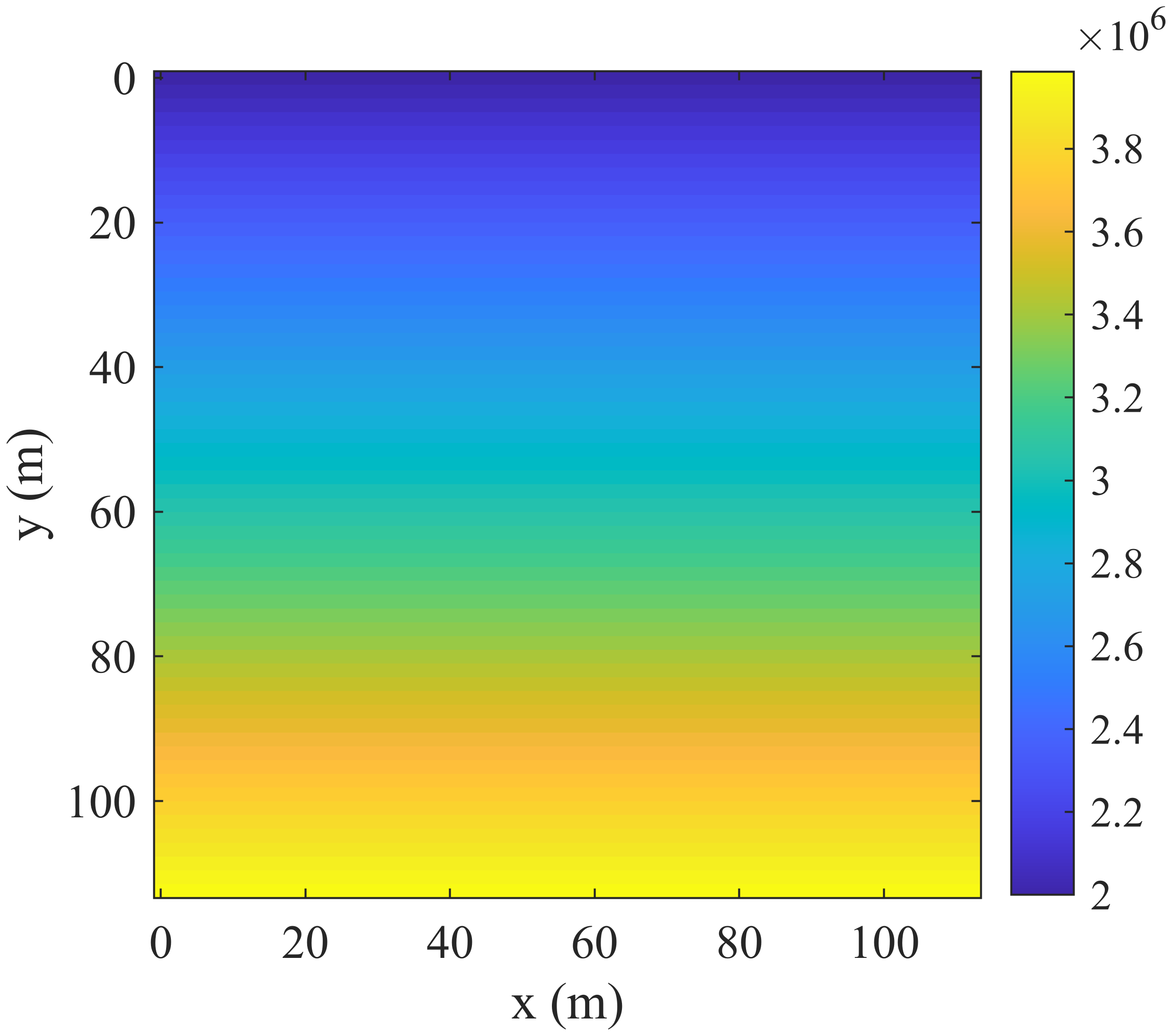}
		\caption{} 
    \label{R for Linear Distribution}
	\end{subfigure}
  \hfil
  \begin{subfigure}[b]{43mm}
    \includegraphics[width=\textwidth]{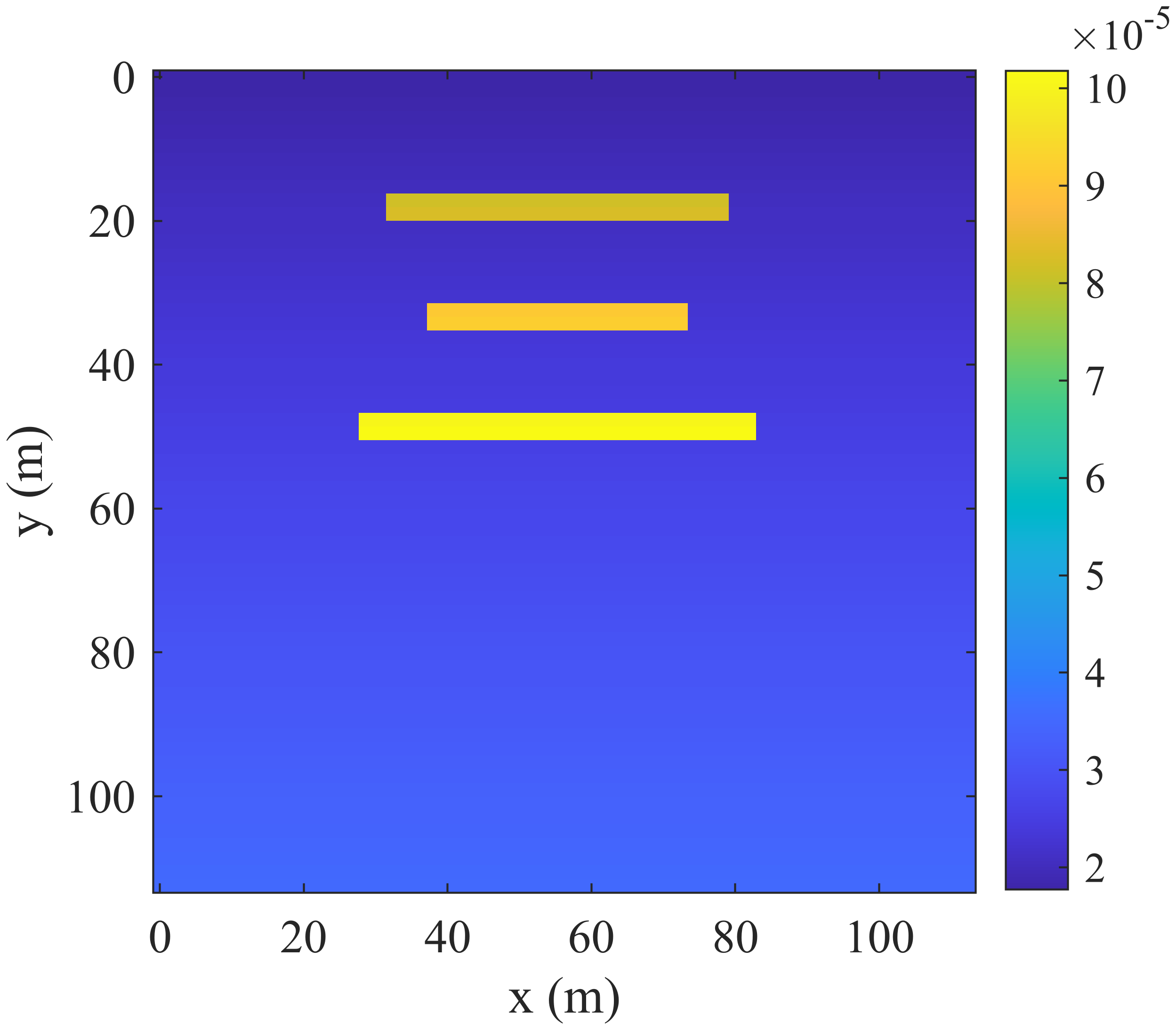}
    \caption{} 
    \label{Sigma for Linear Distribution}
  \end{subfigure}
  \centering
	\begin{subfigure}[b]{43mm}
		\includegraphics[width=\textwidth]{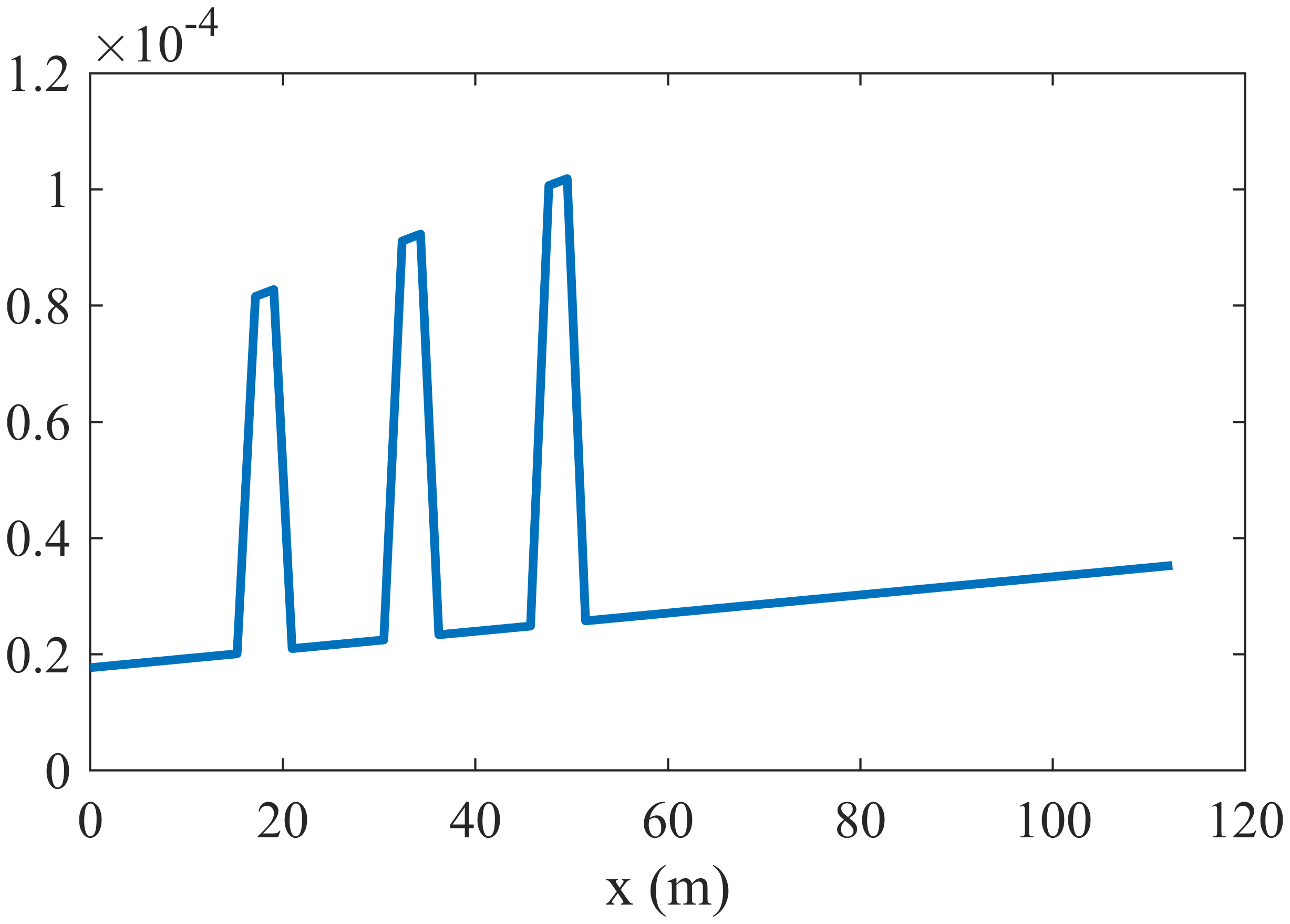}
		\caption{} 
    \label{Sigma Line for Linear Distribution}
	\end{subfigure}
  \caption{(a) Distribution of dissipation parameter $p(\boldsymbol{x})$. (b) Distribution of conductivity $\sigma(\boldsymbol{x})$. (c) Conductivity profile along the vertical line at $\boldsymbol{x}=60$ m.}
\end{figure}
Another assessment of the compensation performance is conducted by analyzing the Common Shot Gathers (CSGs) \cite{mahdad2011separation} corresponding to the transmitter at Sensor 10 shown in Fig.~\ref{CSG for Constant Distribution}. 
The ideal lossless CSG provides a clear depiction of the reflection phenomena, with the hyperbolic reflection curves fully capturing the physical information of wave propagation in this region. 
However, in the measured CSG depicted in Fig.~\ref{CSG for Constant Distribution: Lossy}, these crucial features are obscured by the rapid decay of reflected energy due to medium loss. 
This results in a loss of physical information about the wave propagation. 
After applying the compensation strategy, the CSG, depicted in Fig.~\ref{CSG for Constant Distribution: Comp}, successfully restores the data amplitudes across the entire receiver array. 
The relative error between the compensated and lossless CSG is only 5.74\%, which confirms the robustness and effectiveness of the method across the entire measurement aperture.

\subsection{Case 2: Linear Inhomogeneous Dissipation}

Next, we consider a more challenging scenario in which the dissipation distribution is spatially inhomogeneous. 
The permittivity distribution is the same as in Fig.~\ref{Model}, while the dissipation parameter $p(\boldsymbol{x})$ now varies linearly along the vertical axis. 
As depicted in Fig.~\ref{R for Linear Distribution}, its value increases from $2\times10^6$ \si{\per\second} at the top to $4\times10^6$ \si{\per\second} at the bottom boundary. 
This setup, where the dissipation parameter $p(\boldsymbol{x})$ is the sum of a uniform dissipation $p_0=3\times10^6$ \si{\per\second} and a linearly varying term $\Delta p(\boldsymbol{\boldsymbol{x}})$ as described in Proposition 1, results in the conductivity distribution, $\sigma(\boldsymbol{x})=p(\boldsymbol{x})\cdot\epsilon(\boldsymbol{x})$ shown in Fig.~\ref{Sigma for Linear Distribution}.
The profile of $\sigma(\boldsymbol{x})$ along the vertical line at $\boldsymbol{x}=60$ m is shown in Fig.~\ref{Sigma Line for Linear Distribution}, which illustrates the stepwise increments at the locations of the three stripes.
This experiment is thereby designed to validate the robustness of our compensation strategy against the approximation errors when the assumption of uniform dissipation is violated.

\begin{figure}[!htb]
	\centering
	\begin{subfigure}[b]{43mm}
		\includegraphics[width=\textwidth]{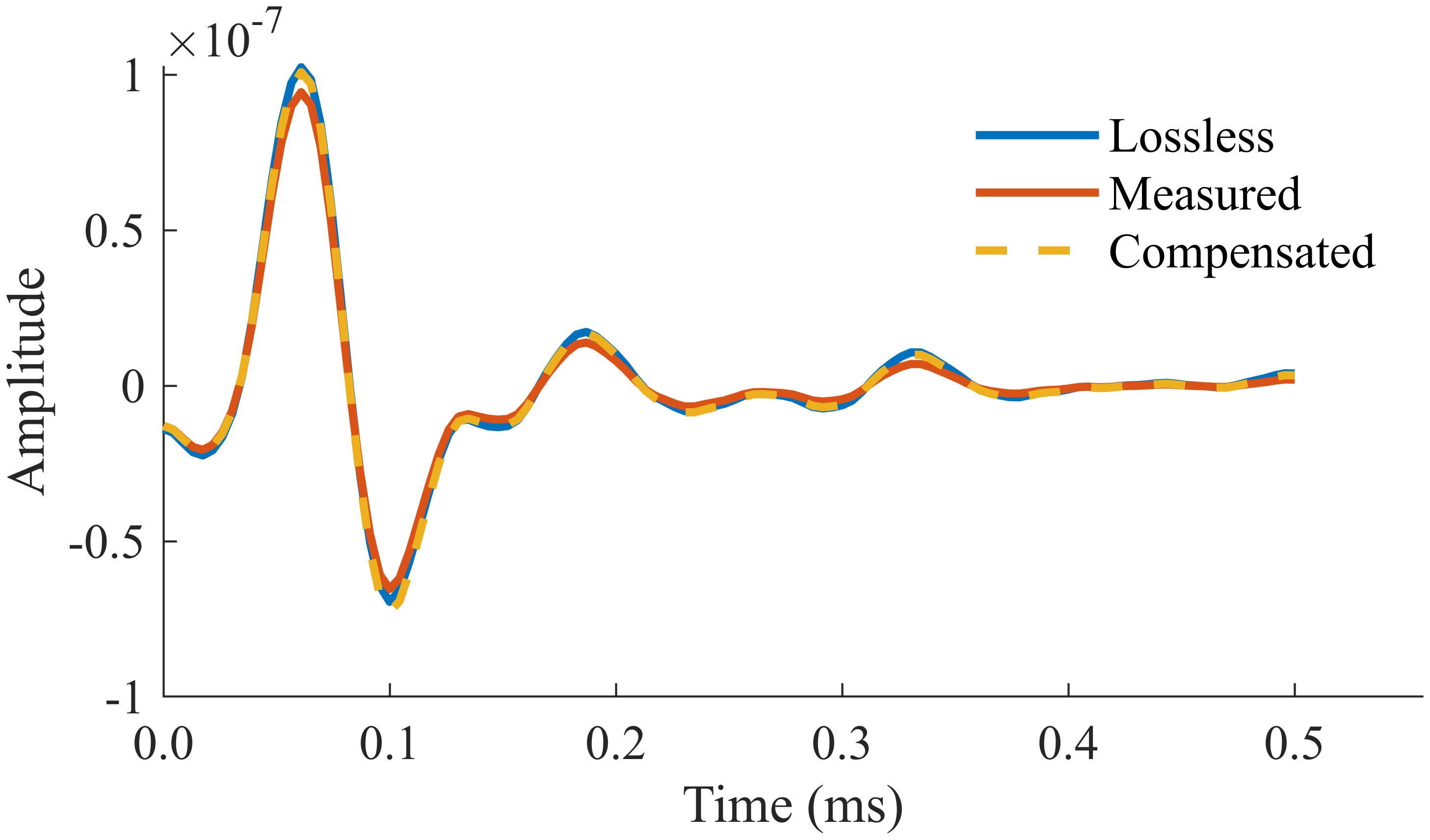}
		\caption{} 
    % \label{Mono Data}
	\end{subfigure}
    \hfill
    \begin{subfigure}[b]{43mm}
		\includegraphics[width=\textwidth]{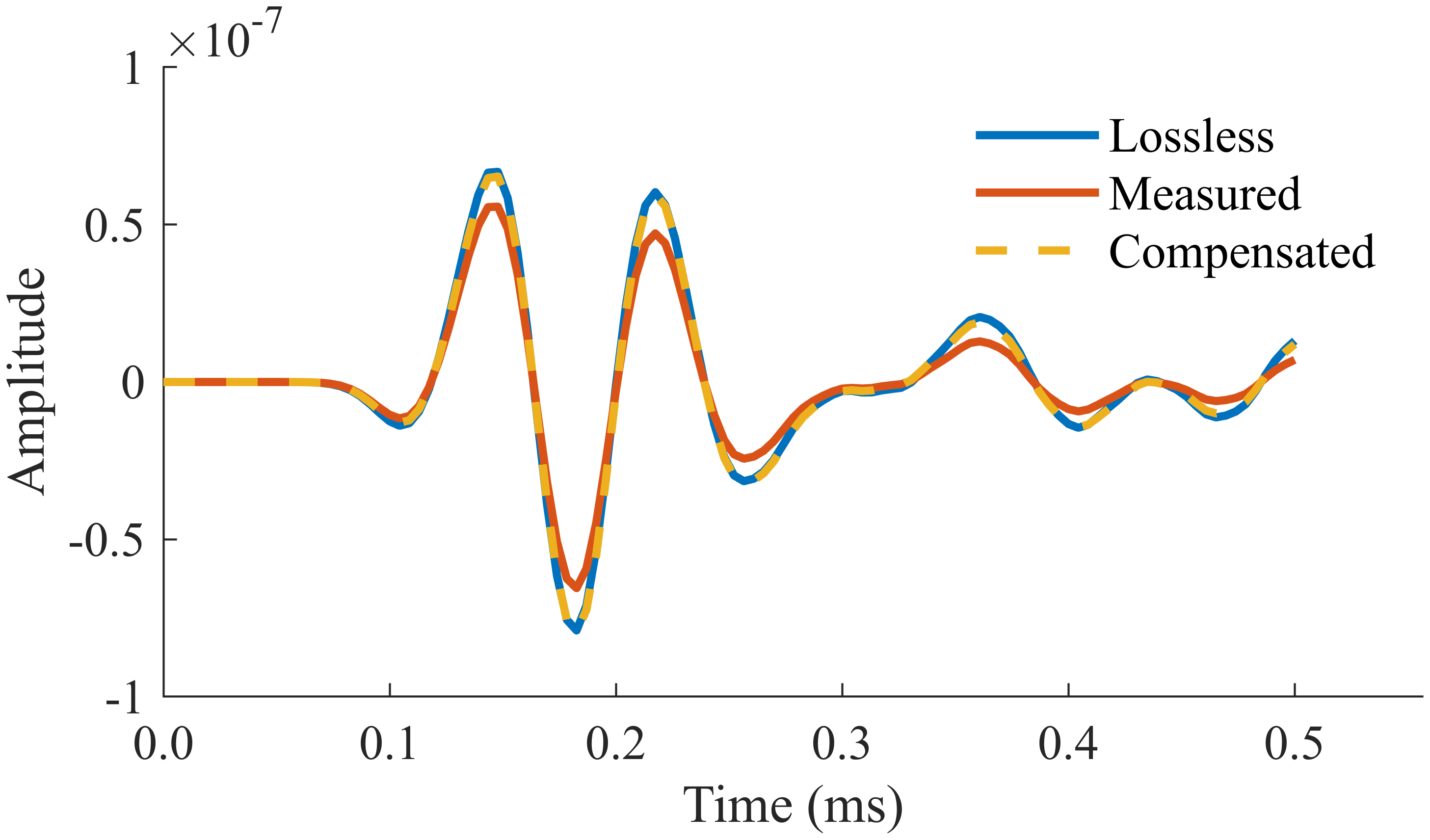}
		\caption{} 
    % \label{Bi Data}
	\end{subfigure}
  \centering
	\begin{subfigure}[b]{43mm}
		\includegraphics[width=\textwidth]{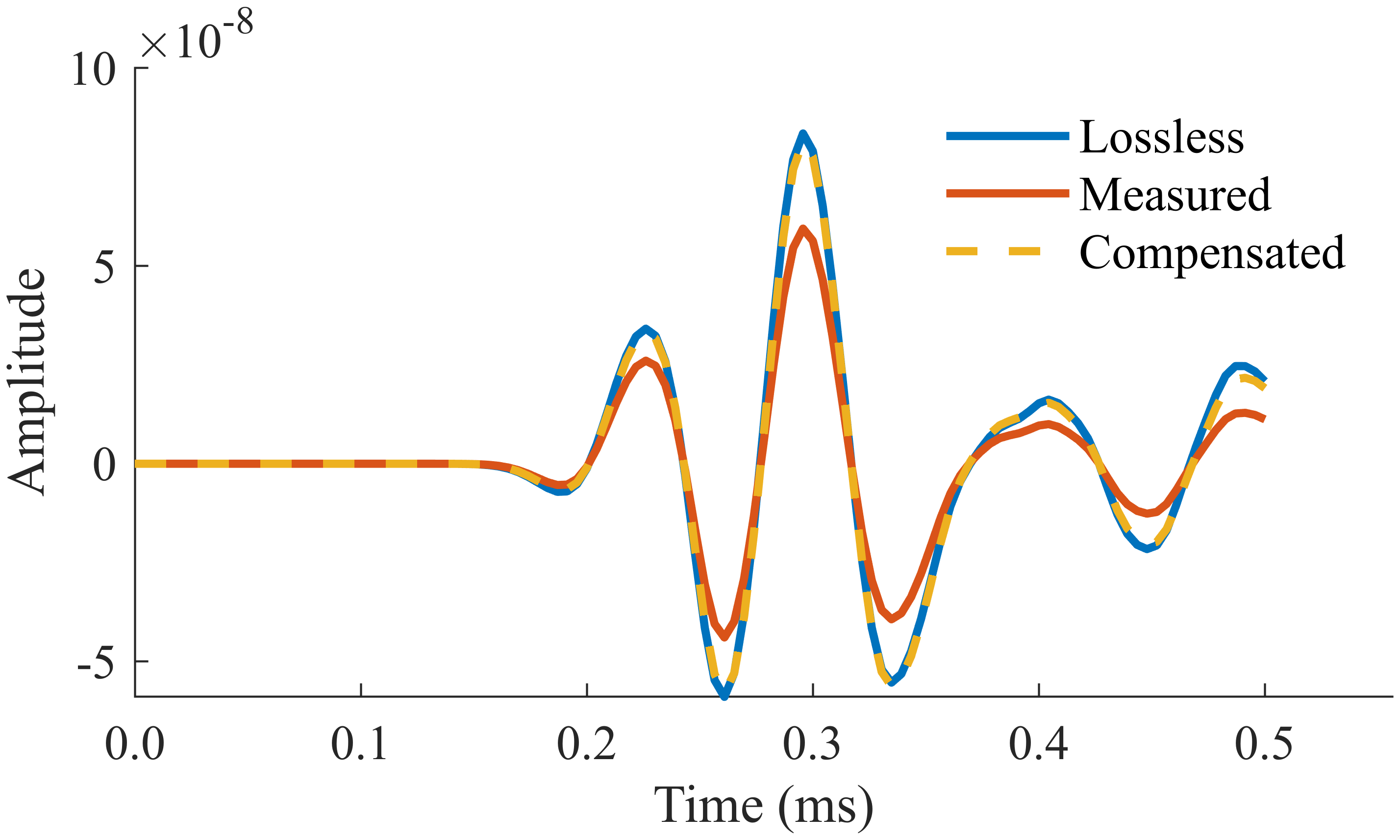}
		\caption{} 
    % \label{Mono Data}
	\end{subfigure}
    \hfill
    \begin{subfigure}[b]{43mm}
		\includegraphics[width=\textwidth]{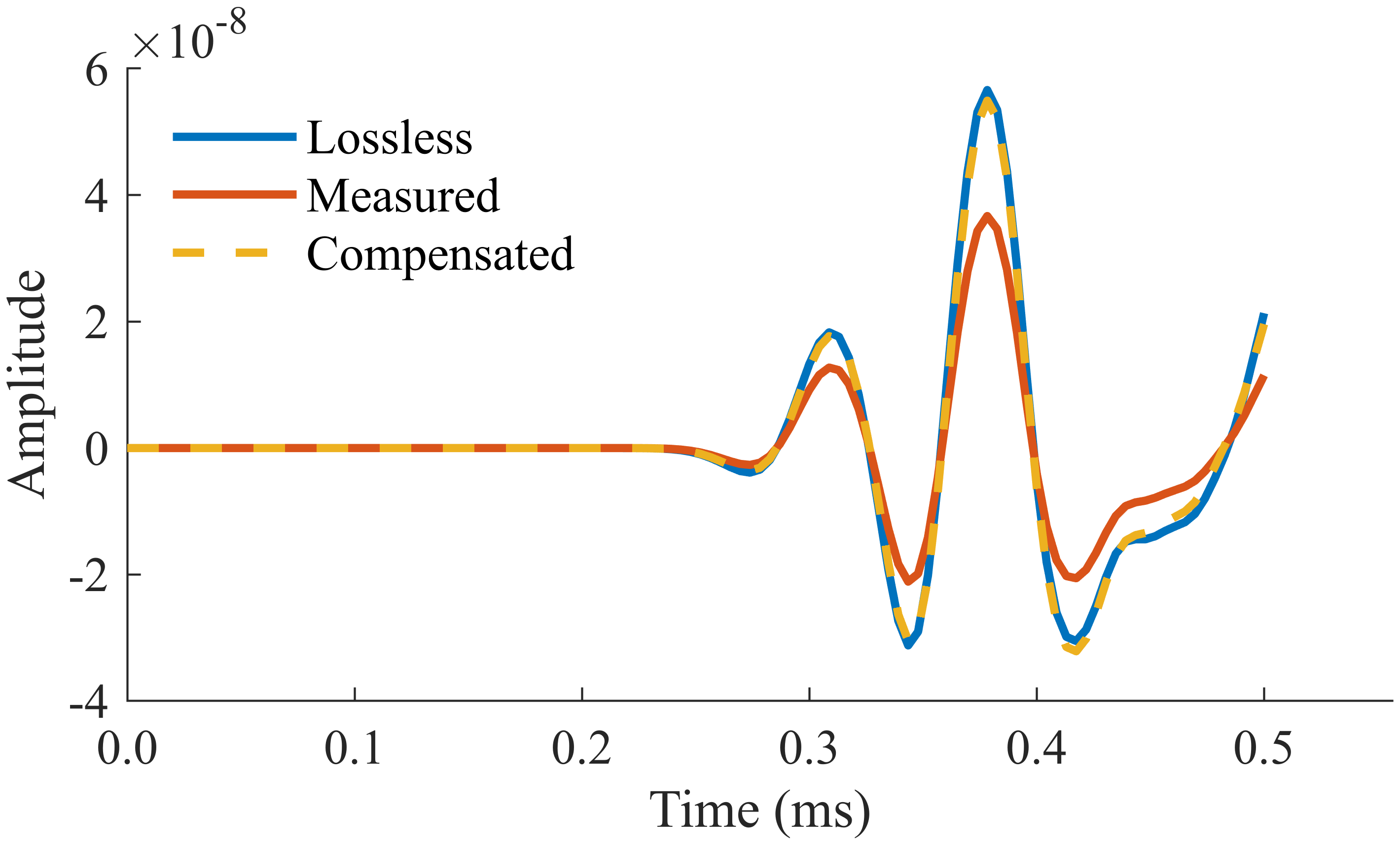}
		\caption{} 
    % \label{Bi Data}
	\end{subfigure}
	\caption{Comparison of lossless data, measured data and compensated data with linear inhomogeneous distribution. 
  The transmitter is at Sensor 1 and receivers are at (a) Sensor 5, (b) Sensor 10, (c) Sensor 15, (d) Sensor 20.} 
  \label{Data Comparison for Linear Distribution}
\end{figure}

\begin{figure}[!htb]
	\centering
	\begin{subfigure}[b]{43mm}
		\includegraphics[width=\textwidth]{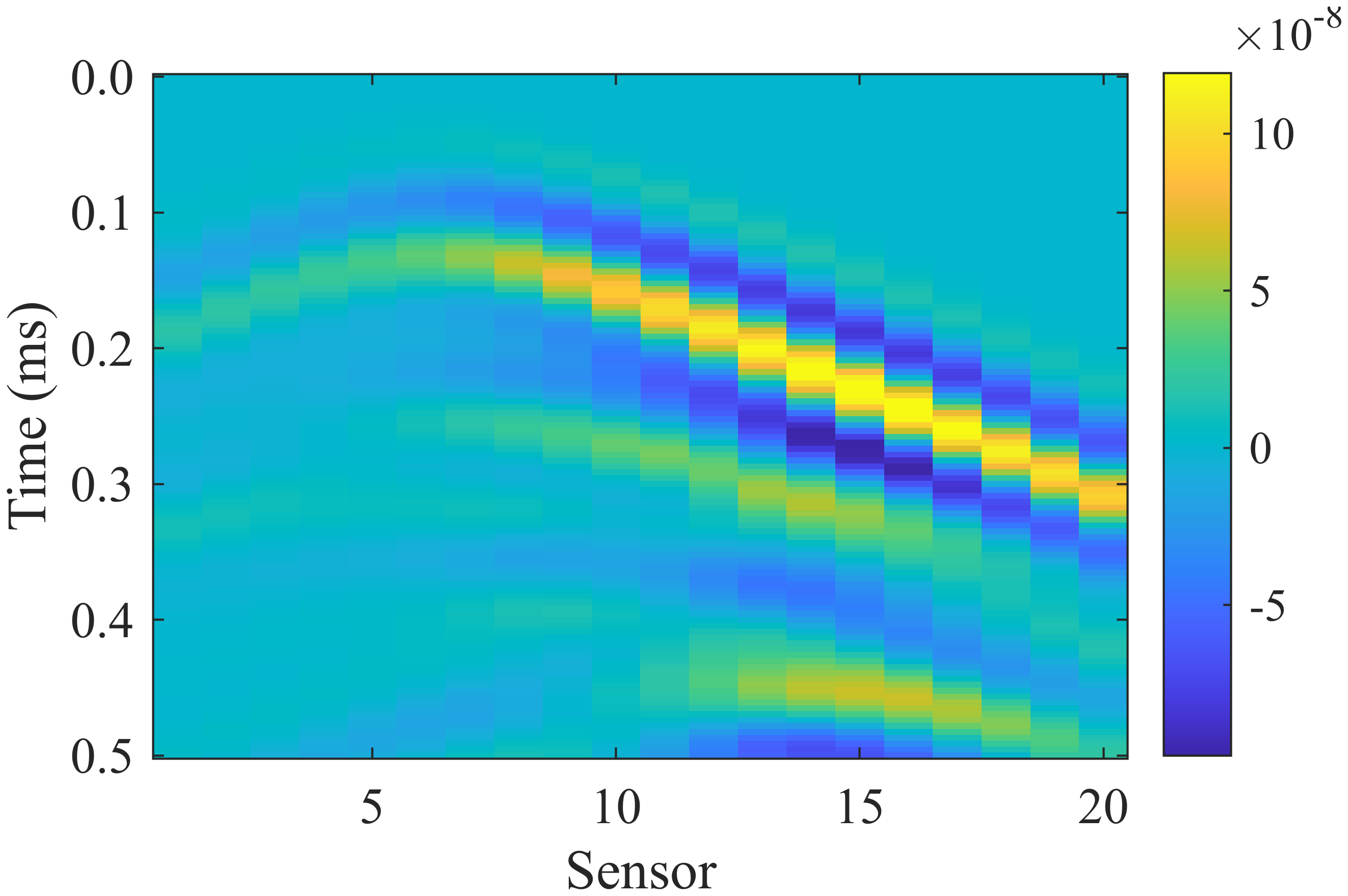}
		\caption{} \label{CSG for Linear Distribution: Lossless}
    % \label{CSG Raw Data}
	\end{subfigure}
    \hfill
    \begin{subfigure}[b]{43mm}
		\includegraphics[width=\textwidth]{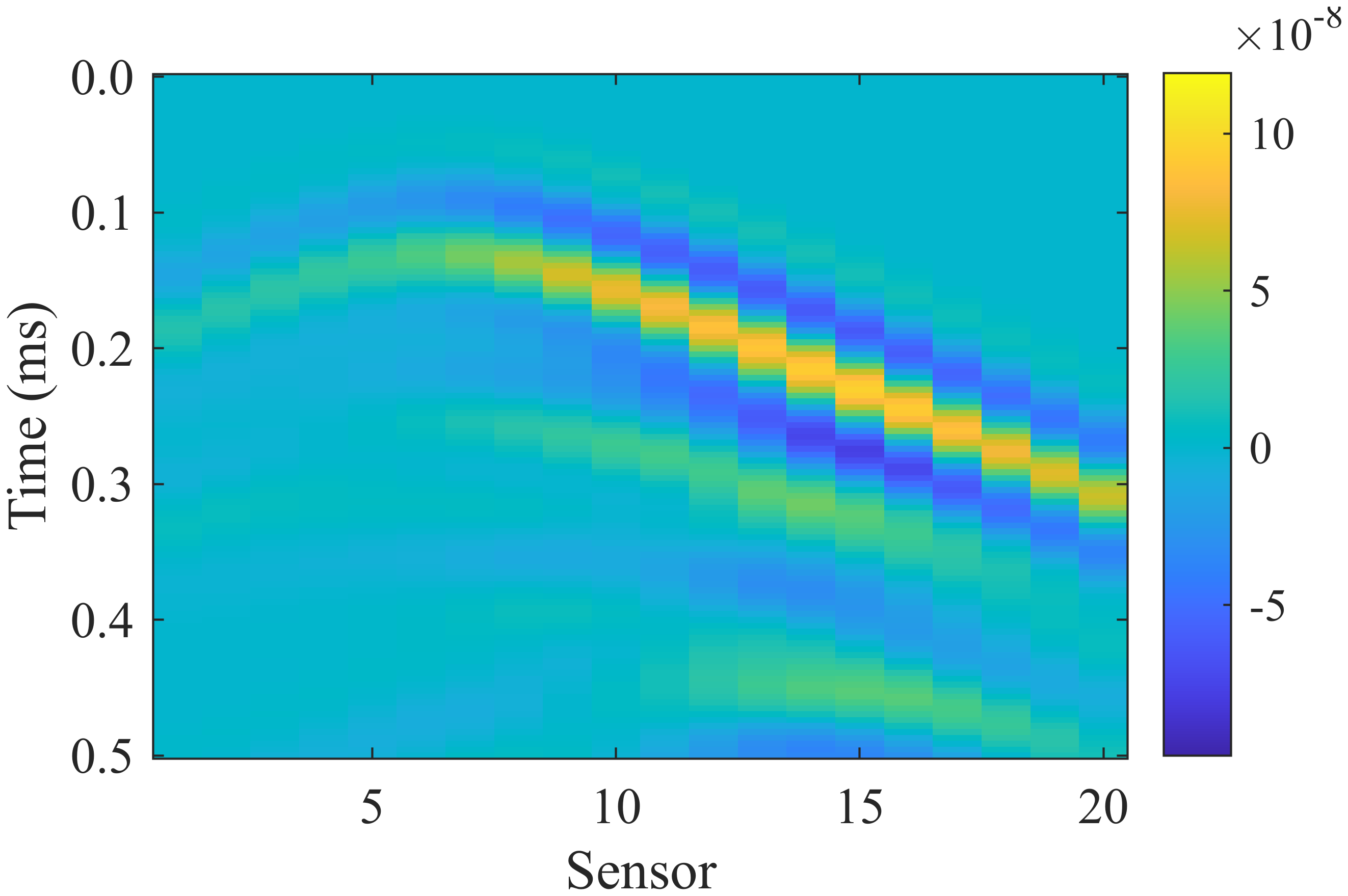}
		\caption{} \label{CSG for Linear Distribution: Lossy}
    % \label{CSG DtB Data}
	\end{subfigure}
  	\centering
	\begin{subfigure}[b]{43mm}
		\includegraphics[width=\textwidth]{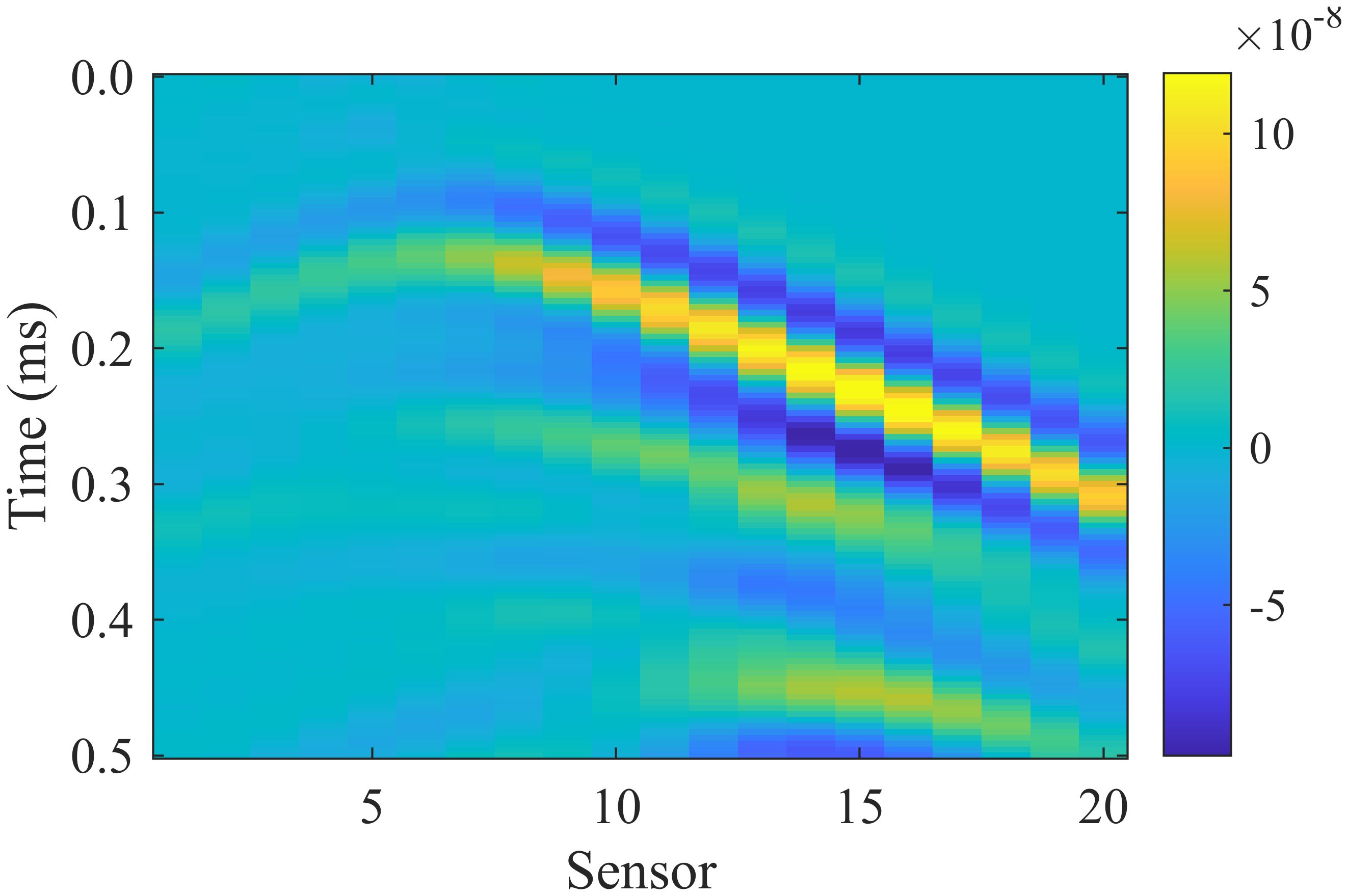}
		\caption{}  \label{CSG for Linear Distribution: Comp}
    % \label{CSG Raw Data}
	\end{subfigure}
	\caption{Common shot gathers at Sensor 5 with linear inhomogeneous distribution. 
  The colorbar indicates the magnitude of the collected data. 
  (a) Synthetic data in lossless environment. (b) Measured data in lossy environment. (c) Compensated data.} 
  \label{CSG for Linear Distribution}
\end{figure}

The result is shown in Fig.~\ref{Data Comparison for Linear Distribution}, which compares the lossless, measured, and compensated bistatic data traces for transmitter at Sensor 1 and receivers at Sensors 5, 10, 15, and 20. 
Despite the spatially varying dissipation, the compensation method accurately amplifies the attenuated waveforms to the level of their lossless counterparts, yielding relative data misfits in the $L2$-norm of 3.56\%, 3.21\%, 4.25\%, and 5.16\%, respectively. 
This is consistent across the entire dataset, leading to a overall relative data misfit of 2.87\%.

The global attenuation compensation performance is also evaluated using the CSGs at Sensor 5, as presented in Fig.~\ref{CSG for Linear Distribution}.
The lossless CSG depicts the hyperbolic reflection events. 
In contrast, the measured CSG in Fig.~\ref{CSG for Linear Distribution: Lossy} suffers from spatially distributed attenuation, which obscures the reflection events at deep region.
After applying the compensation strategy, the resulting CSG in Fig.~\ref{CSG for Linear Distribution: Comp} shows an improvement.
\begin{figure}[!htb]
  \centering
  % \includegraphics[width=0.28\textwidth]{Figure/Linear/r_dist.png}
  % \caption{Contrast distribution of the 2-D scenario.} 
  \begin{subfigure}[b]{43mm}
		\includegraphics[width=\textwidth]{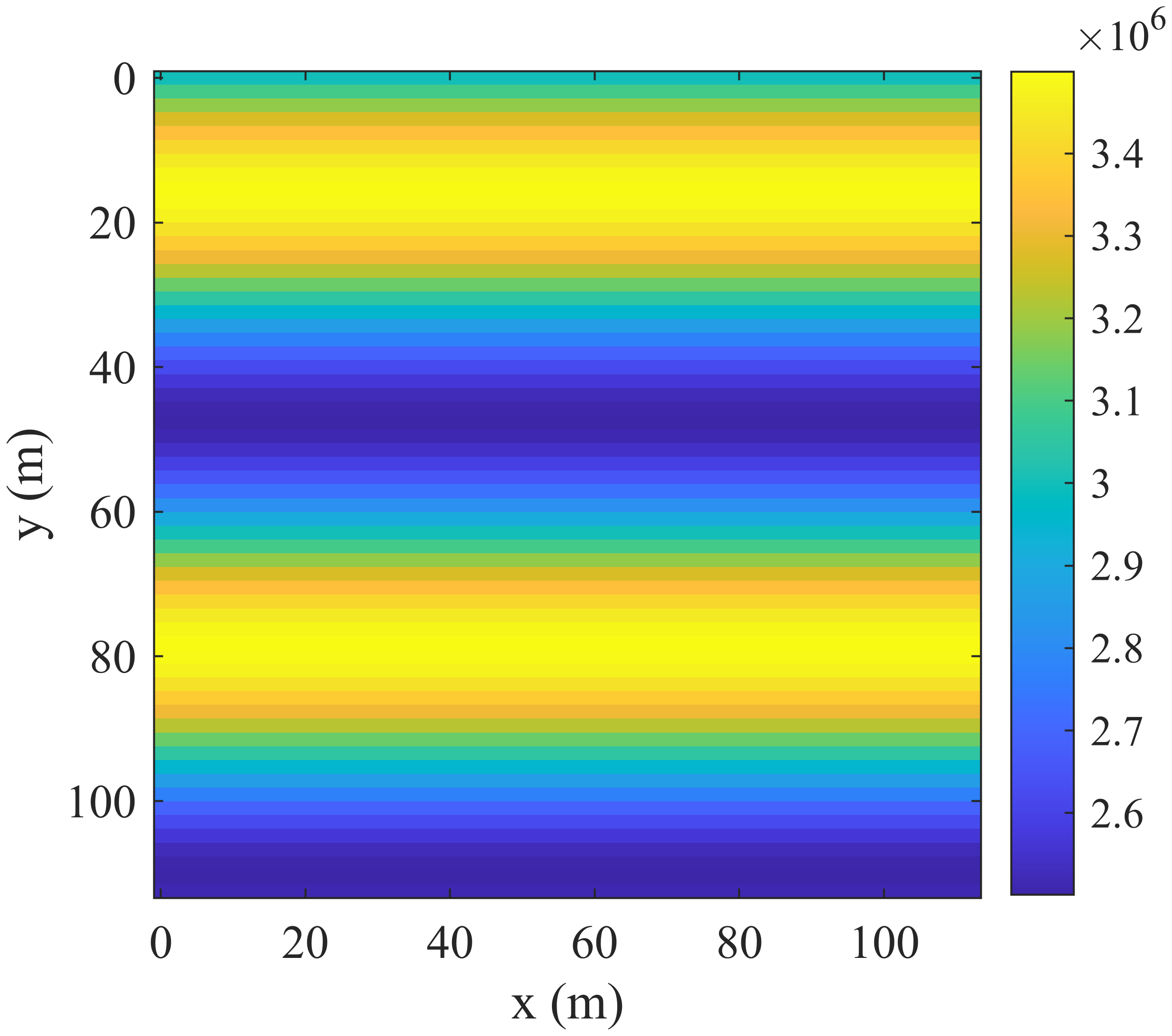}
		\caption{} \label{R for Sine Distribution}
	\end{subfigure}
  \hfil
  \begin{subfigure}[b]{43mm}
    \includegraphics[width=\textwidth]{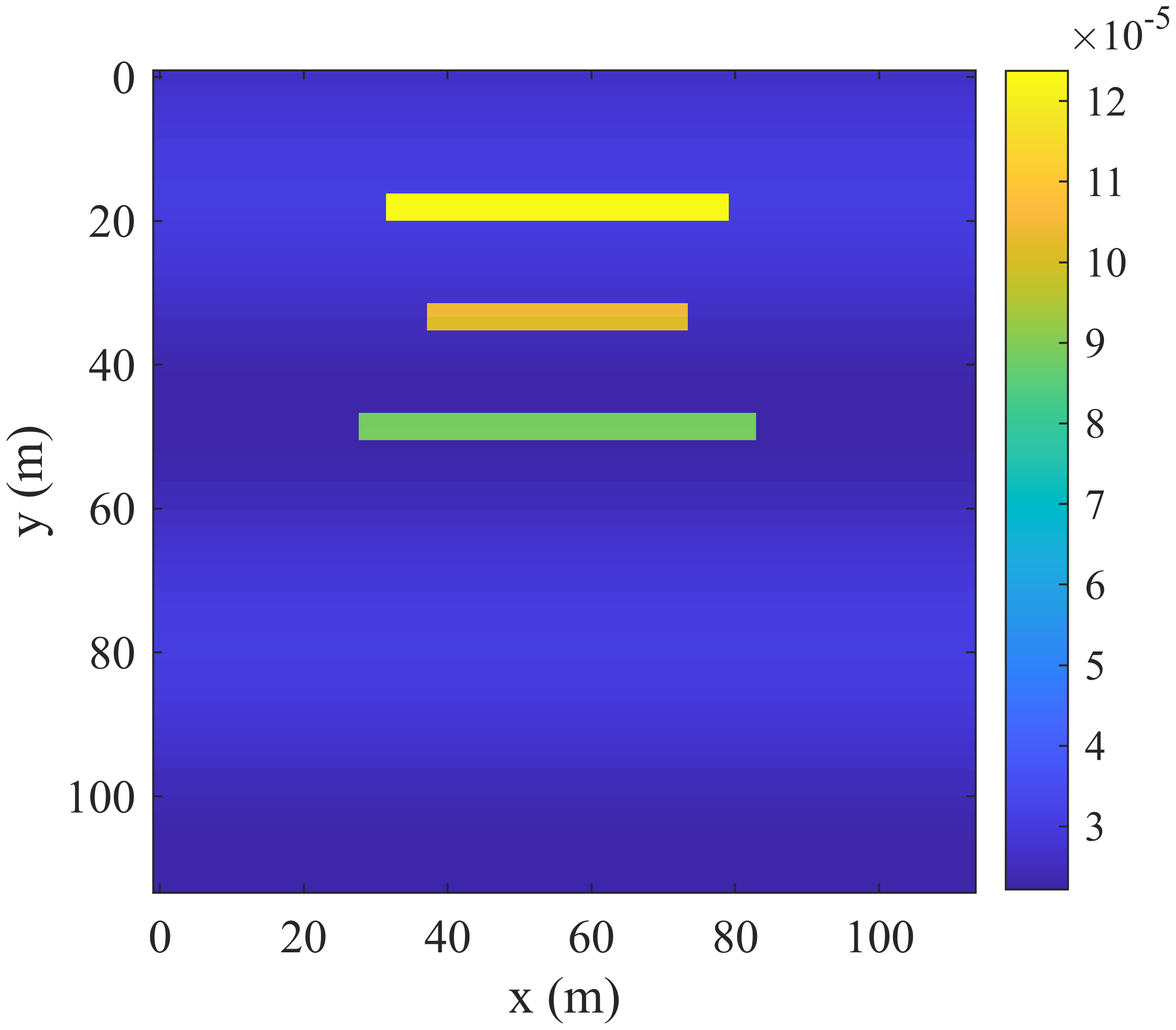}
    \caption{} \label{Sigma for Sine Distribution}
  \end{subfigure}
  \centering
	\begin{subfigure}[b]{43mm}
		\includegraphics[width=\textwidth]{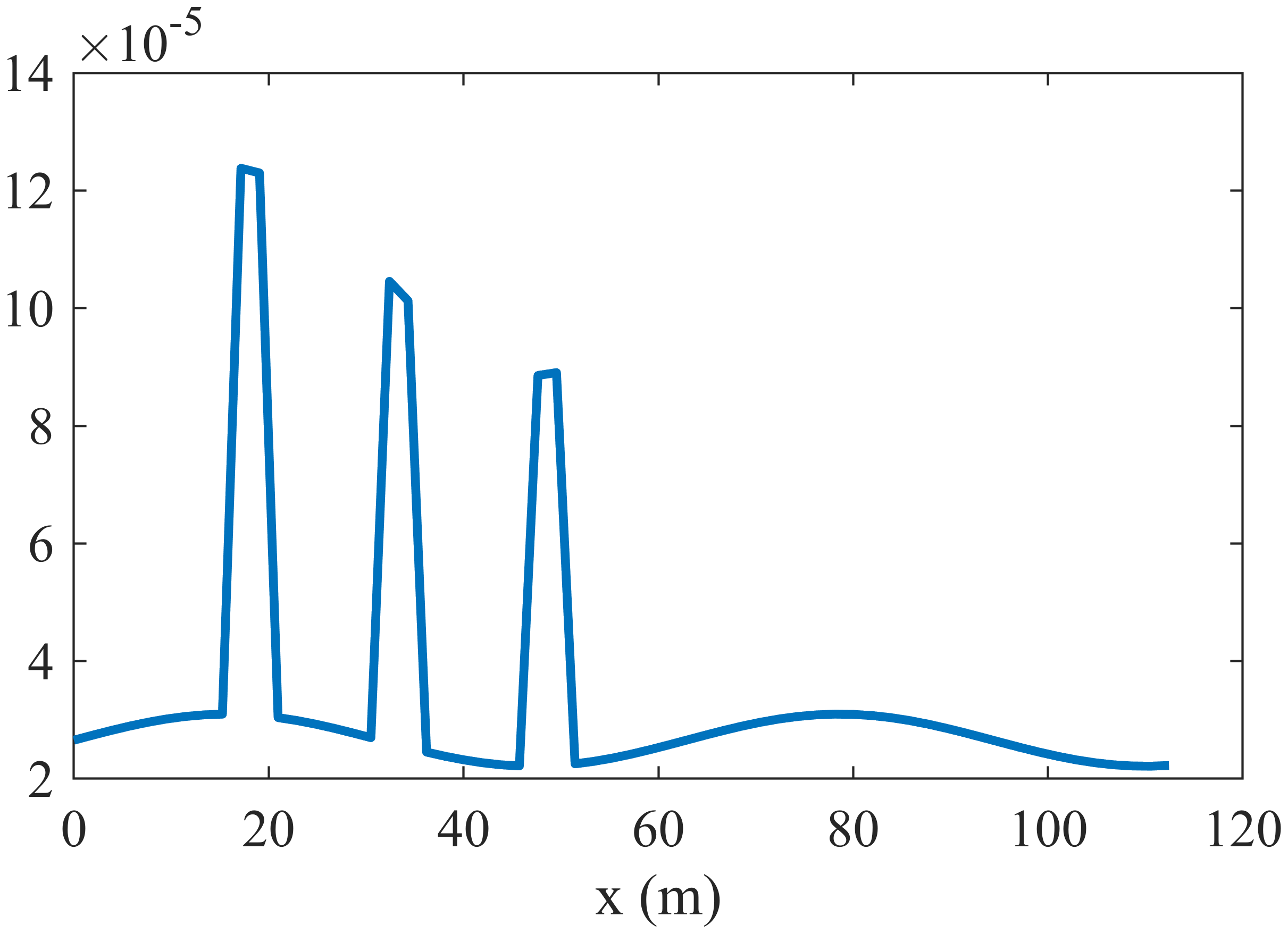}
		\caption{} \label{Sigma Line for Sine Distribution} 
    % \label{Mono Data}
	\end{subfigure}
  \caption{(a) Distribution of dissipation parameter $p(\boldsymbol{x})$. (b) Distribution of conductivity $\sigma(\boldsymbol{x})$. (c) Conductivity profile along the vertical line at $\boldsymbol{x}=60$ m.}
\end{figure}
The hyperbolic reflection events are effectively restored across the full array, and the relative error compared to the lossless CSG is 4.78\%.
This confirms that the proposed strategy performs reliably for non-uniform dissipation.

\subsection{Case 3: Biased Sinusoidal Inhomogeneous Dissipation}

Finally, we investigate a more complex case with a biased sinusoidal dissipation distribution to further validate the proposed compensation strategy. 
The permittivity distortion remains identical to that in Fig.\ref{Model}.
The dissipation parameter $p(\boldsymbol{x})$ is modeled by the expression $p(\boldsymbol{x}) = p_0 \cdot [3+0.5 \cos(k\boldsymbol{x})]$, with the scaling factor for this case set to $p_0=1\times10^6$ \si{\per\second}. 
The spatial frequency is $k=4\pi/3\lambda$, where $\lambda$ is the wavelength in vacuum corresponding to the central frequency of \SI{8}{M\hertz}. This distribution is illustrated in Fig.~\ref{R for Sine Distribution}.
The corresponding conductivity distribution $\sigma(\boldsymbol{x})$ is illustrated in Fig.~\ref{Sigma for Sine Distribution}, and its profile along the vertical line at $\boldsymbol{x}=60$ m is shown in Fig.~\ref{Sigma Line for Sine Distribution}.
This experiment aims to test the robustness of the compensation method in a non-monotonically varying lossy environment.

The results are presented in Fig.~\ref{Data Comparison for Sine Distribution}. 
The figure displays four data traces corresponding to a transmitter at Sensor 10 and receivers at Sensors 5, 10, 15, and 20. 
The proposed method restores the attenuated signals with relative data misfits in the $L2$-norm are 3.20\%, 4.12\%, 3.30\%, and 4.43\% for these selected traces. 
This accuracy is maintained throughout the measurement set, resulting in an overall data misfit of 3.95\%.

\begin{figure}[!htb]
	\centering
	\begin{subfigure}[b]{43mm}
		\includegraphics[width=\textwidth]{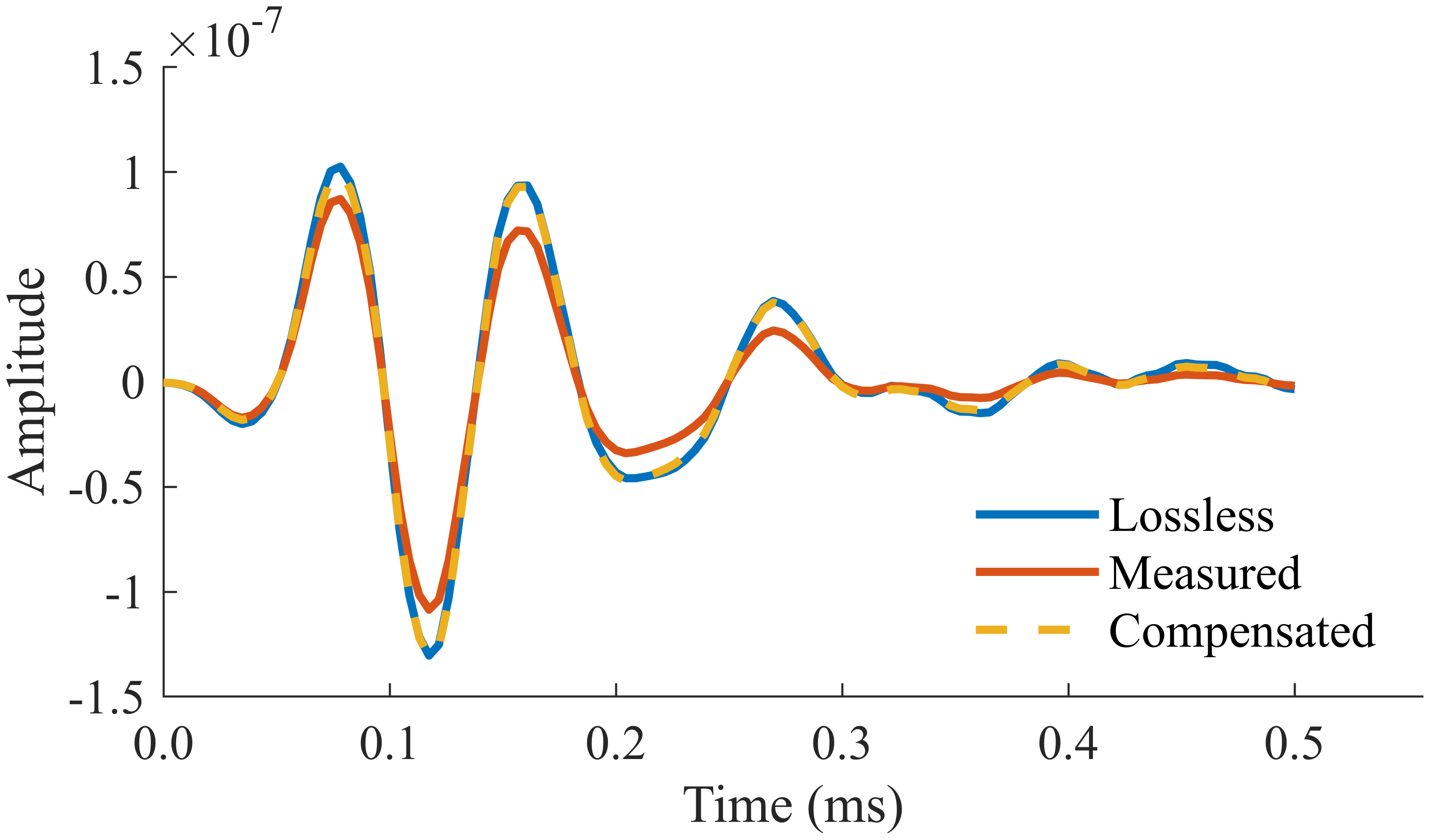}
		\caption{} 
    % \label{Mono Data}
	\end{subfigure}
    \hfill
    \begin{subfigure}[b]{43mm}
		\includegraphics[width=\textwidth]{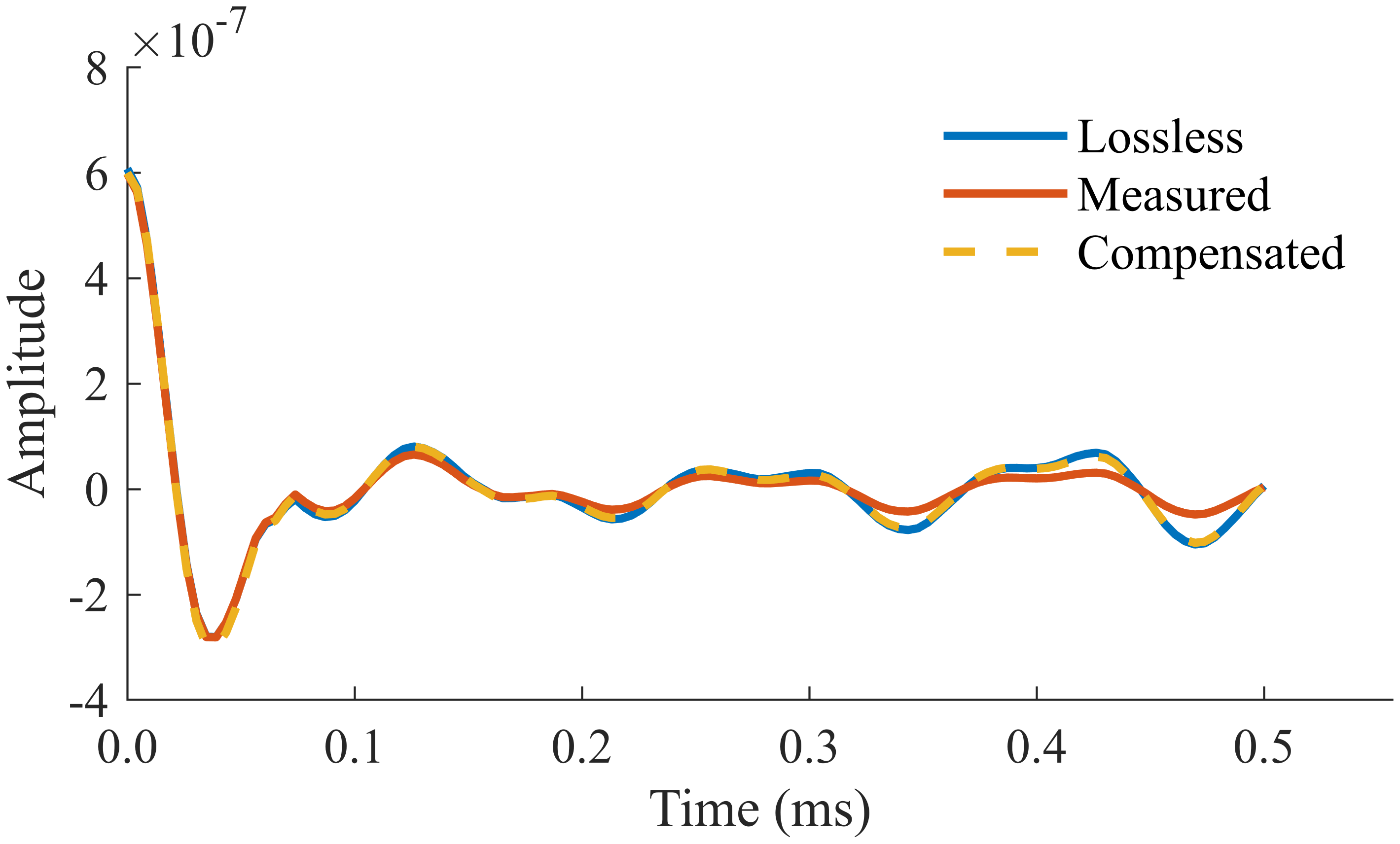}
		\caption{} 
    % \label{Bi Data}
	\end{subfigure}
  \centering
	\begin{subfigure}[b]{43mm}
		\includegraphics[width=\textwidth]{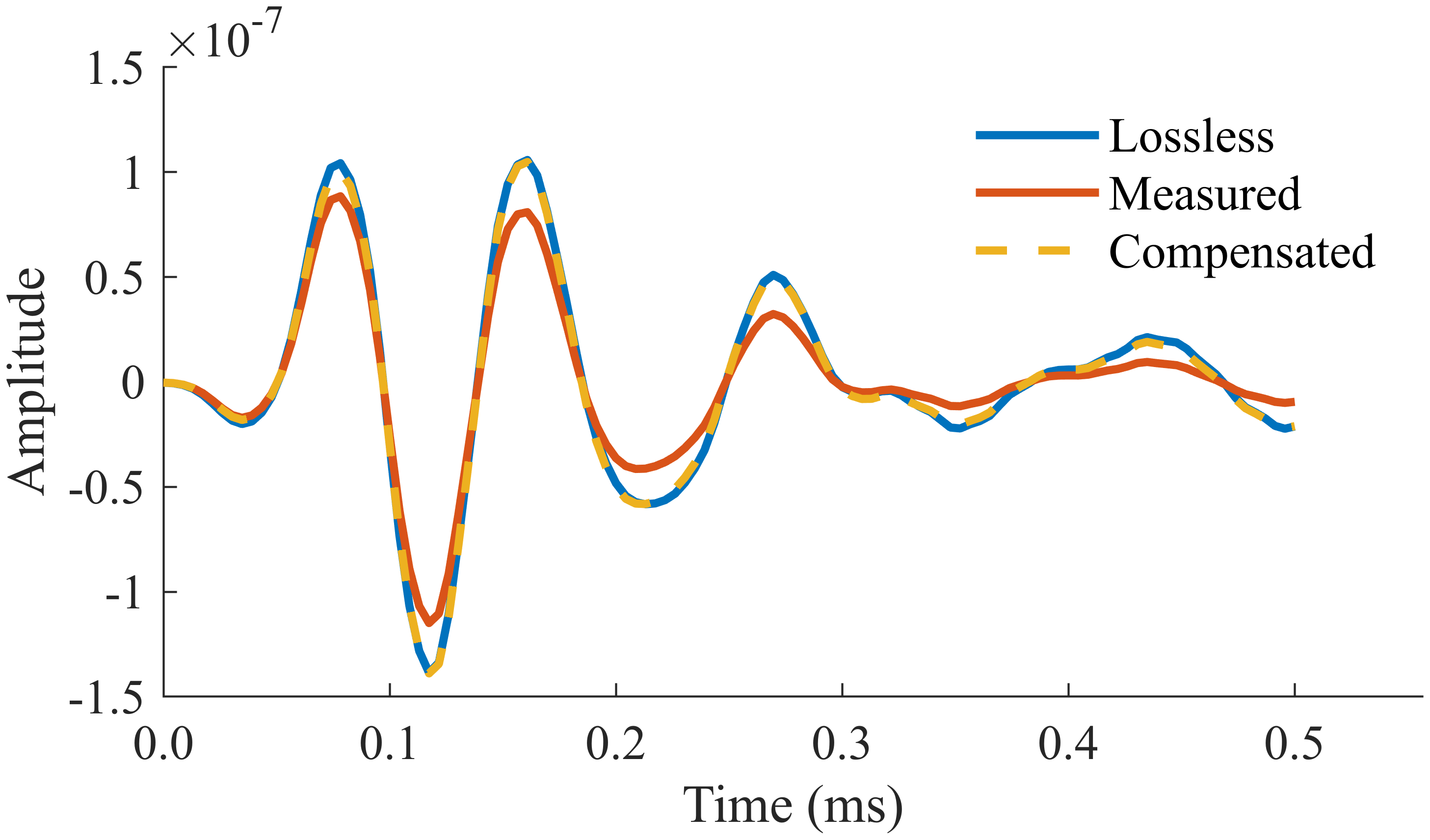}
		\caption{} 
    % \label{Mono Data}
	\end{subfigure}
    \hfill
    \begin{subfigure}[b]{43mm}
		\includegraphics[width=\textwidth]{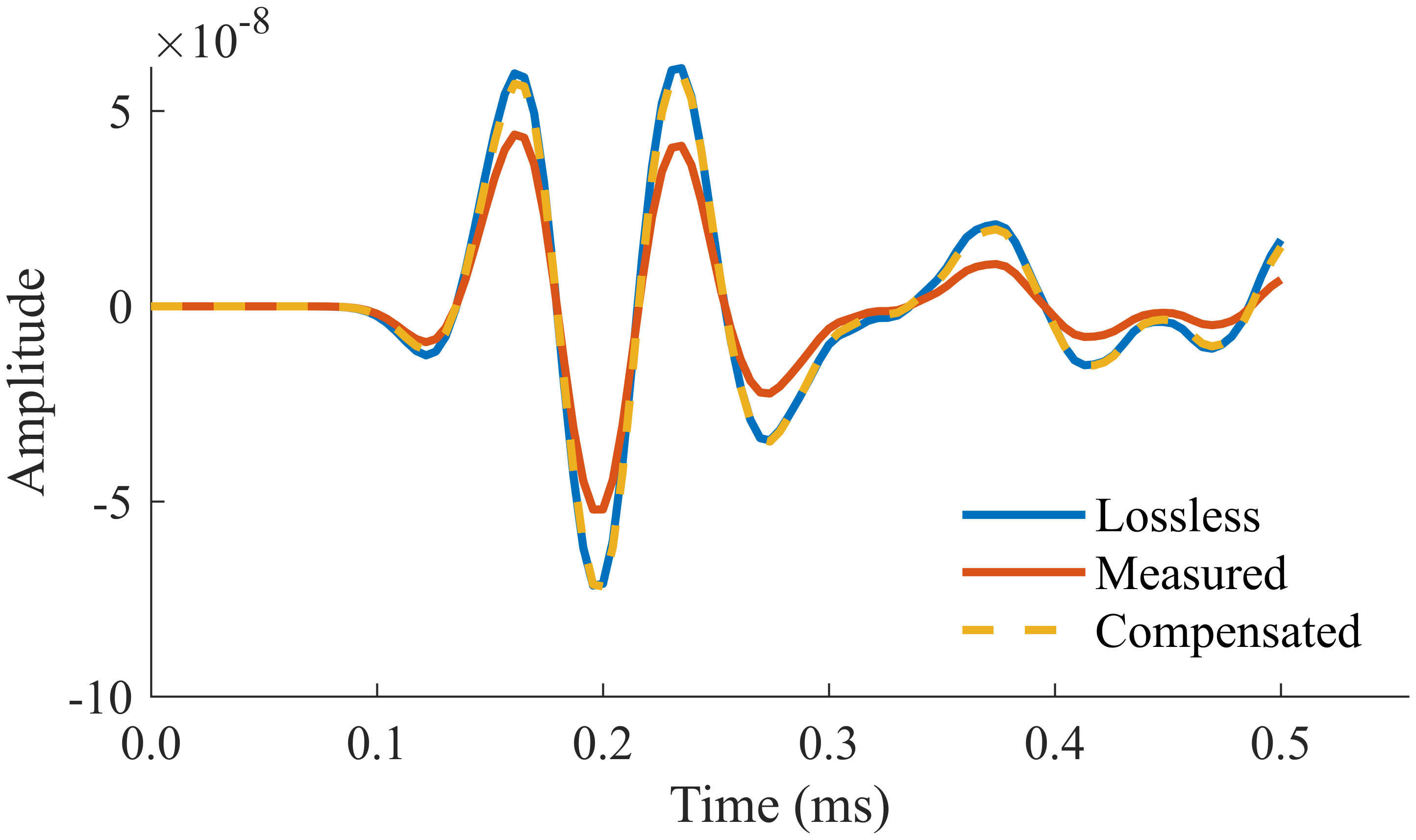}
		\caption{} 
    % \label{Bi Data}
	\end{subfigure}
	\caption{Comparison of lossless data, measured data and compensated data with biased sinusoidal inhomogeneous distribution. 
  The transmitter is at Sensor 10 and receivers are at (a) Sensor 5, (b) Sensor 10, (c) Sensor 15, (d) Sensor 20.} 
  \label{Data Comparison for Sine Distribution}
\end{figure}

\begin{figure}[!htb]
	\centering
	\begin{subfigure}[b]{43mm}
		\includegraphics[width=\textwidth]{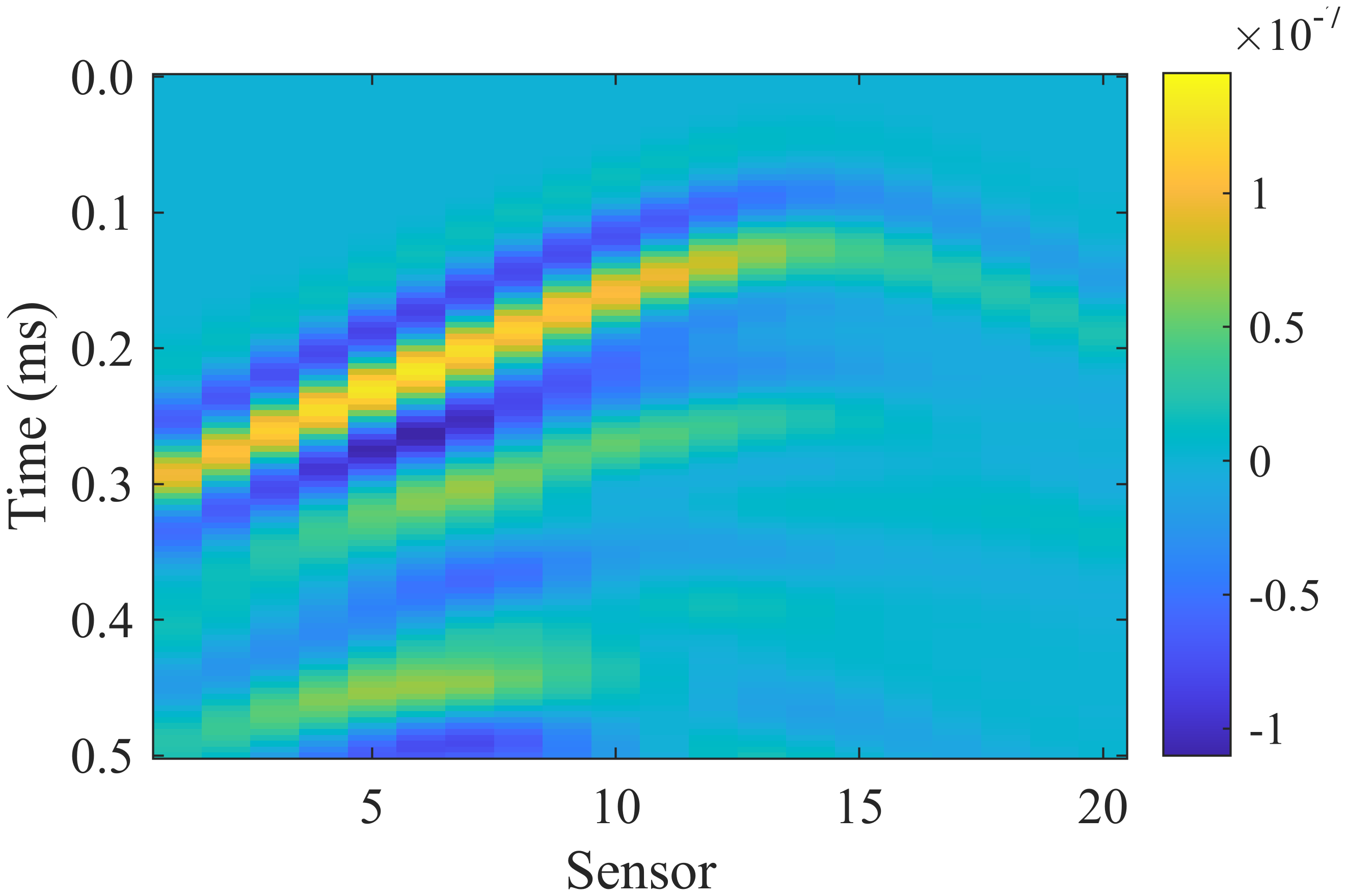}
		\caption{} \label{CSG for Sin Distribution: Lossless}
	\end{subfigure}
    \hfill
    \begin{subfigure}[b]{43mm}
		\includegraphics[width=\textwidth]{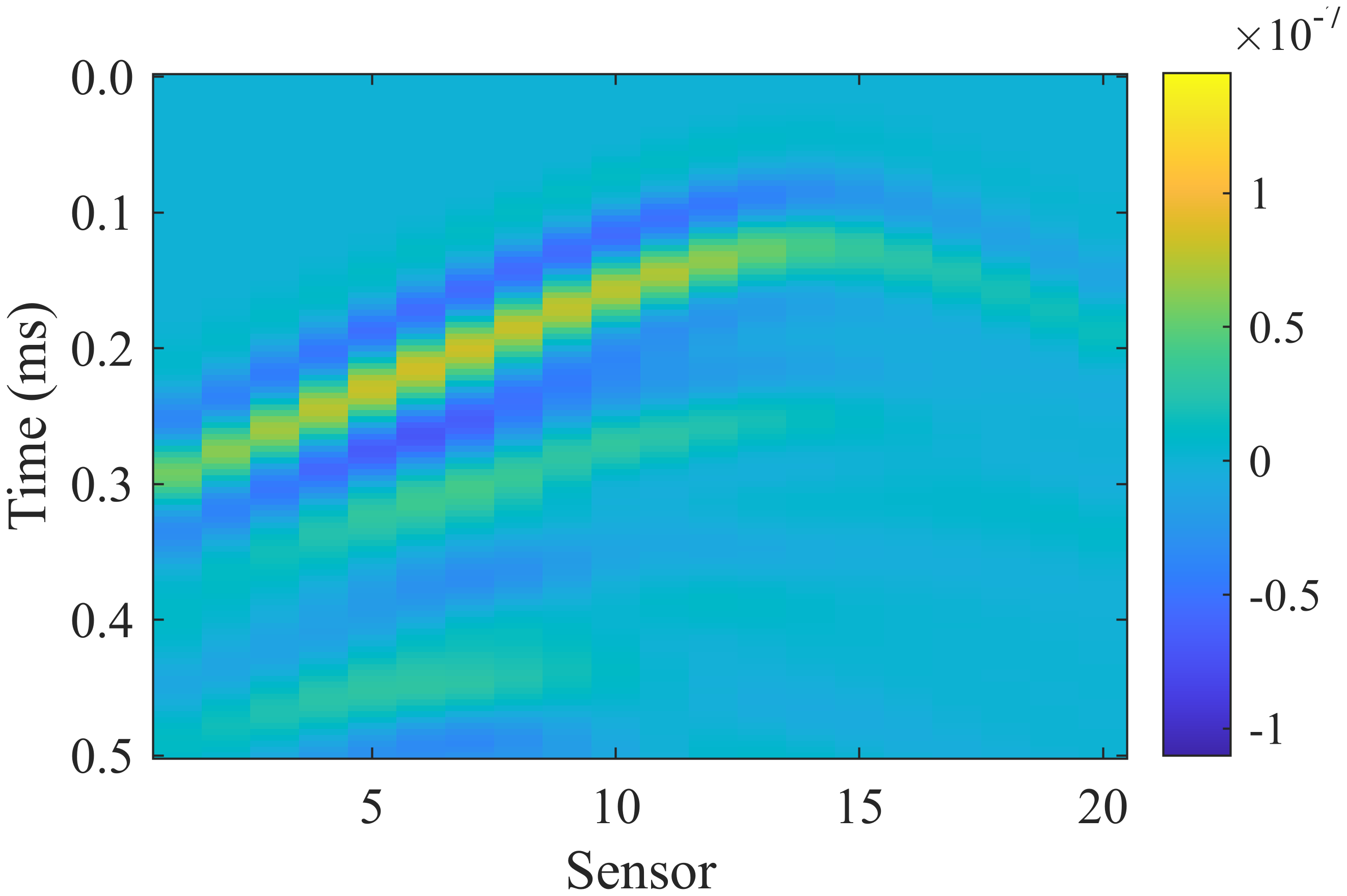}
		\caption{} \label{CSG for Sin Distribution: Lossy}
	\end{subfigure}
  	\centering
	\begin{subfigure}[b]{43mm}
		\includegraphics[width=\textwidth]{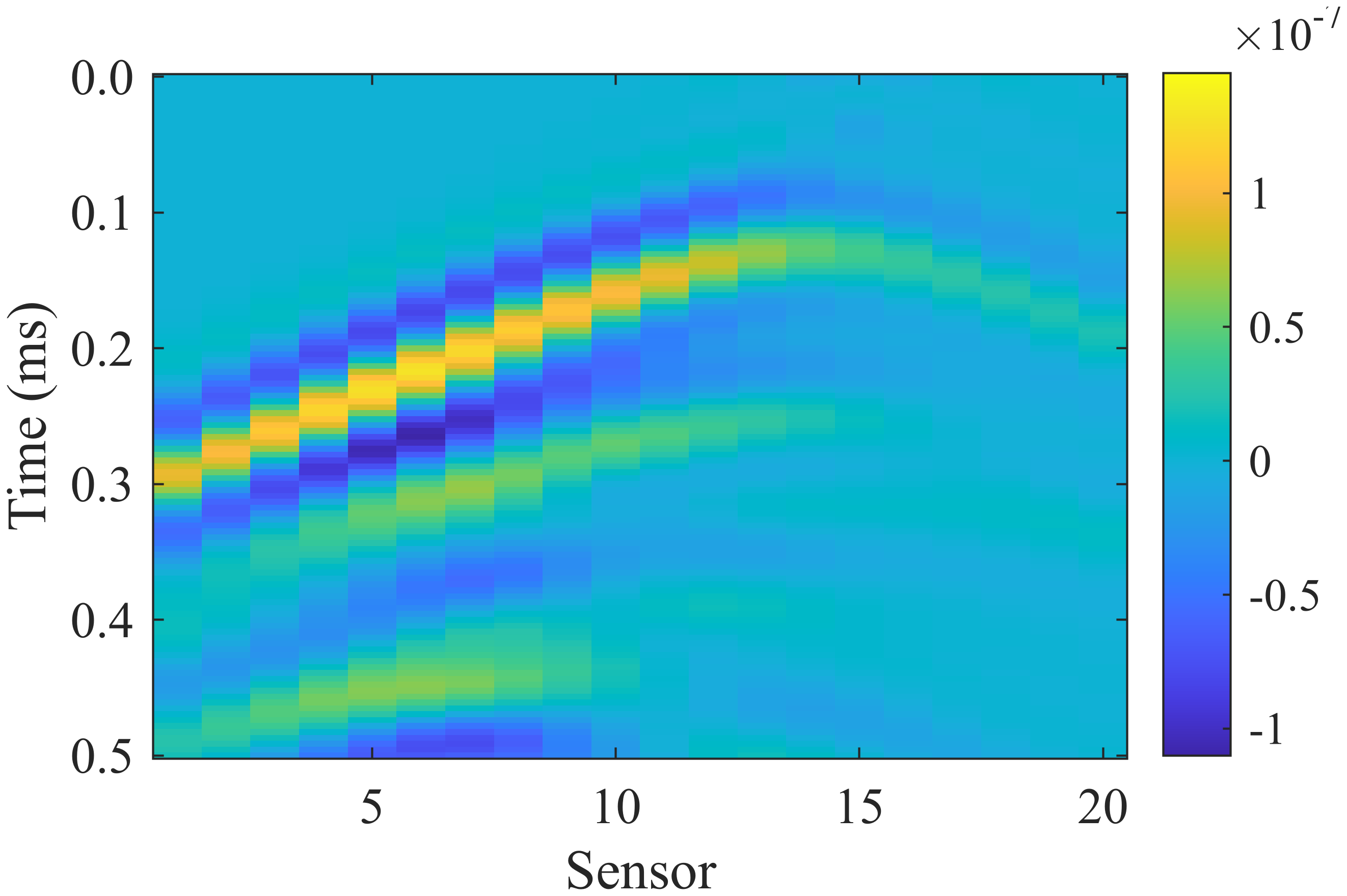}
		\caption{} \label{CSG for Sin Distribution: Comp}
	\end{subfigure}
	\caption{Common shot gathers at Sensor 15. 
  The colorbar indicates the magnitude of the collected data with biased sinusoidal inhomogeneous distribution. 
  (a) Synthetic data in lossless environment. (b) Measured data in lossy environment. (c) Compensated data.} 
  \label{CSG for Sin Distribution}
\end{figure}

The CSGs for the transmitter at Sensor 15 is shown in Fig.~\ref{CSG for Sin Distribution}. 
While the measured data in Fig.~\ref{CSG for Sin Distribution: Lossy} is distorted by complex attenuation patterns, the compensated CSG in Fig.~\ref{CSG for Sin Distribution: Comp} effectively restores the hyperbolic reflection events. 
The result closely matches the ideal lossless CSG in Fig.~\ref{CSG for Sin Distribution: Lossless}, with a relative error of 6.62\%. 
This test shows the robustness of the proposed compensation strategy, confirming its applicability to inhomogeneous dissipation parameter $p(\boldsymbol{x})$.

\subsection{Boundary of Compensation Method}

To investigate the performance boundary of the proposed attenuation compensation method, we test it on a series of lossy media.
The dissipation distribution is identical to the Case 3, described by $p(\boldsymbol{x})=p_0 \cdot [3+0.5 \cos(k\boldsymbol{x})]$ where $k=4\pi/3\lambda$.
The scaling factor $p_0$ is varied from $1\times10^6$ \si{\per\second} to $10\times10^6$ \si{\per\second}. 
% Note that Case 3 corresponds to $p_0=1\times10^6$.

\begin{figure}[!htb]
	\centering
	\begin{subfigure}[b]{43mm}
		\includegraphics[width=\textwidth]{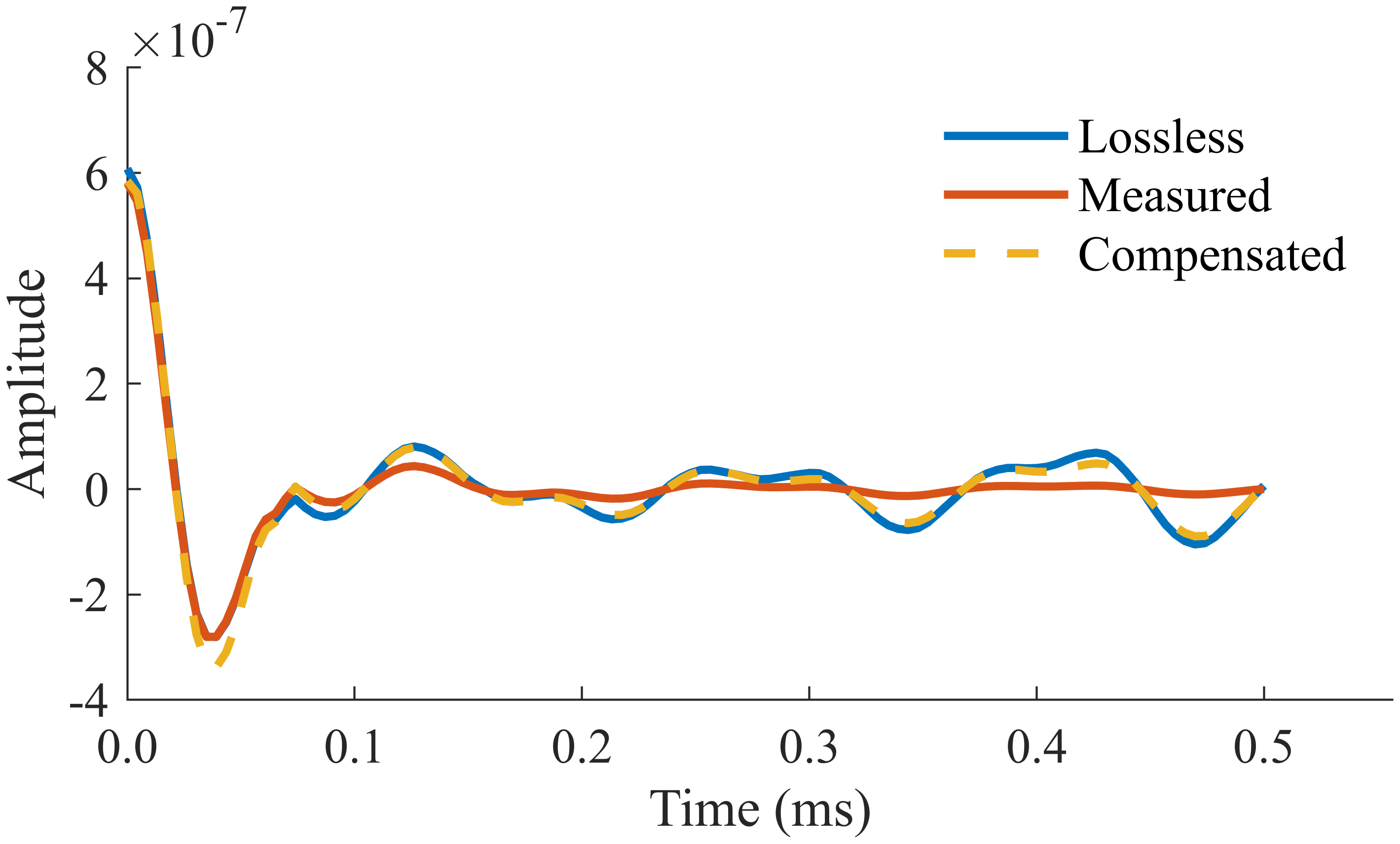}
		\caption{} 
    % \label{Mono Data}
	\end{subfigure}
    \hfill
    \begin{subfigure}[b]{43mm}
		\includegraphics[width=\textwidth]{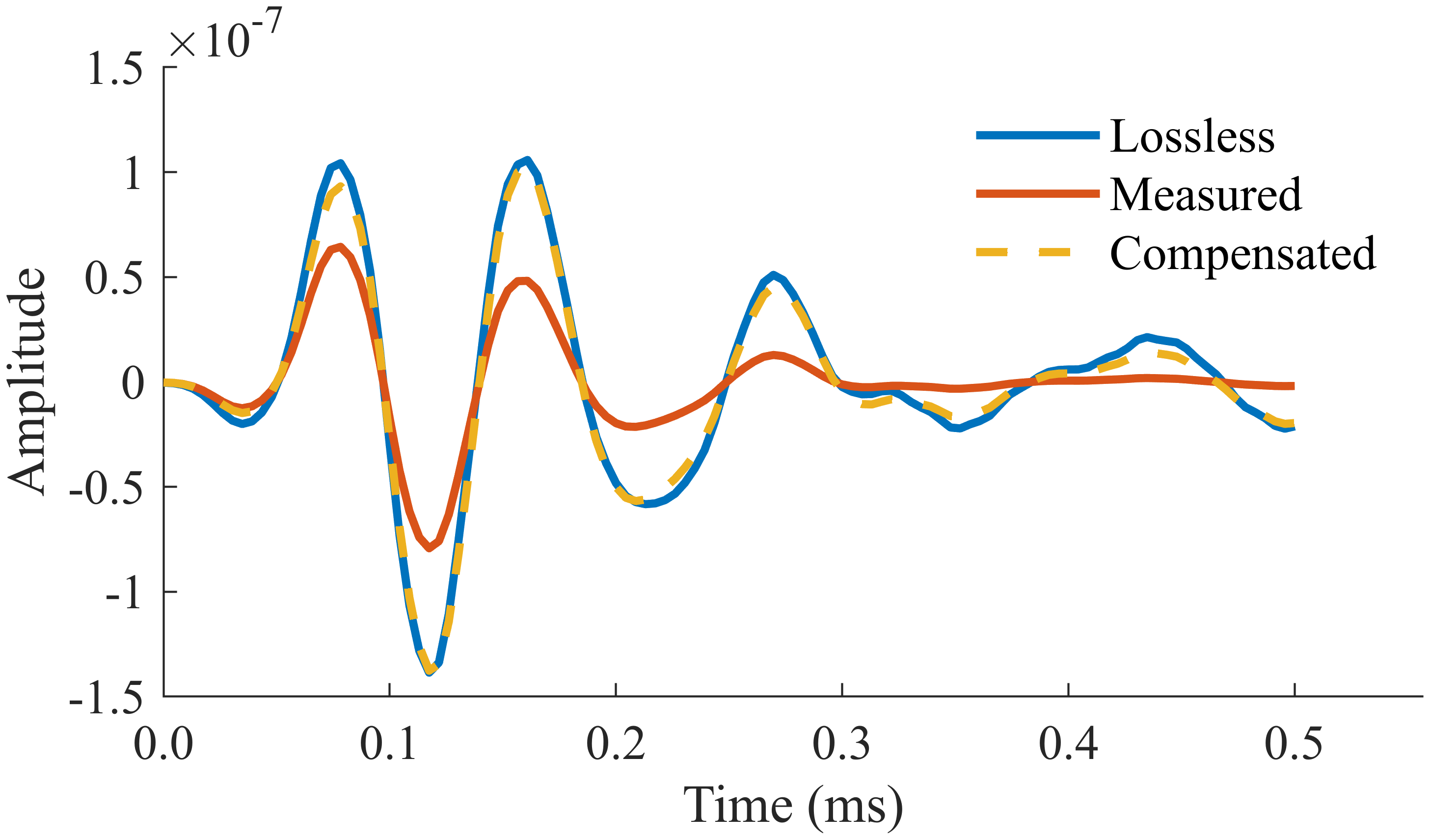}
		\caption{} 
	\end{subfigure}
  \centering
	\begin{subfigure}[b]{43mm}
		\includegraphics[width=\textwidth]{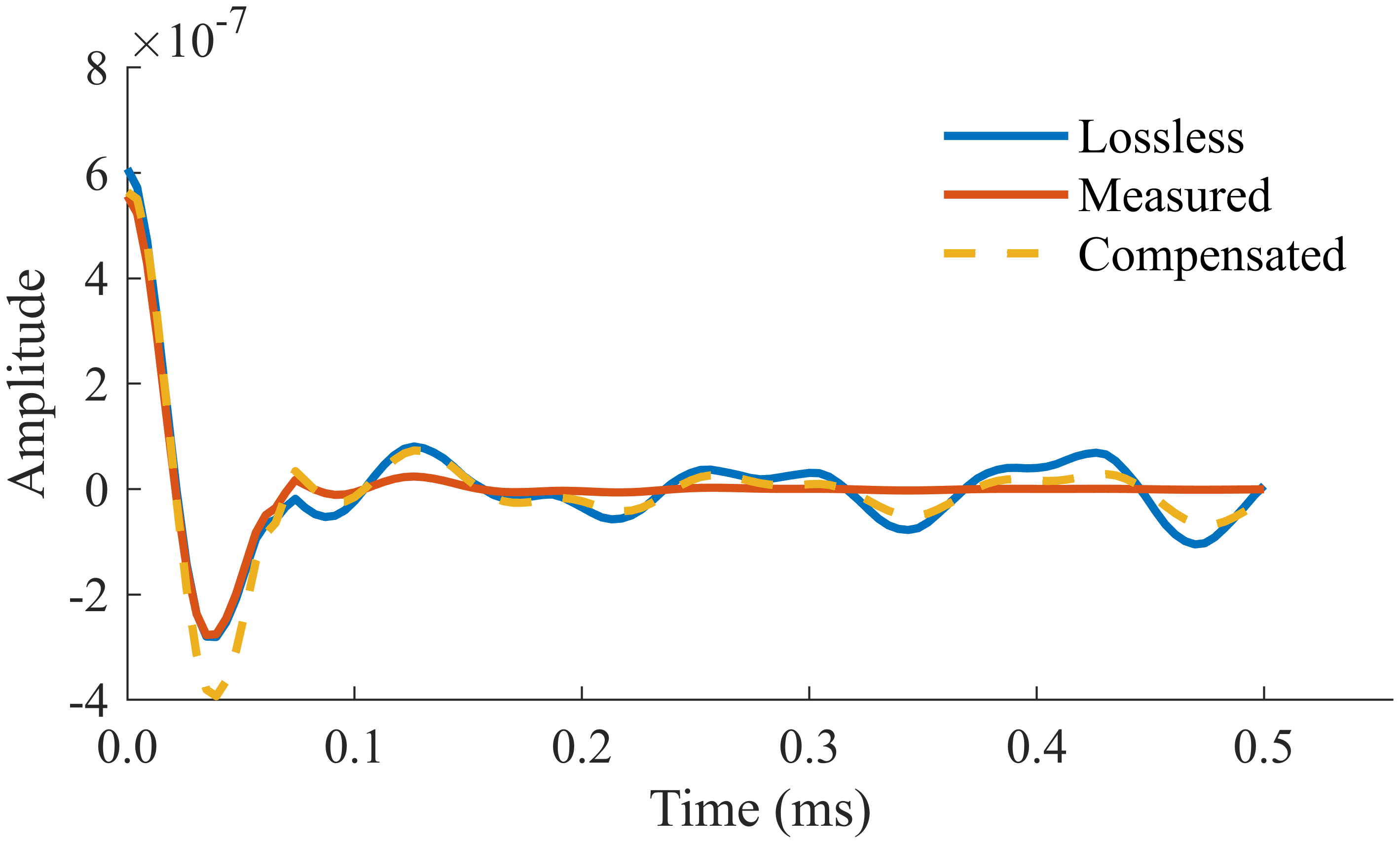}
		\caption{} 
	\end{subfigure}
    \hfill
    \begin{subfigure}[b]{43mm}
		\includegraphics[width=\textwidth]{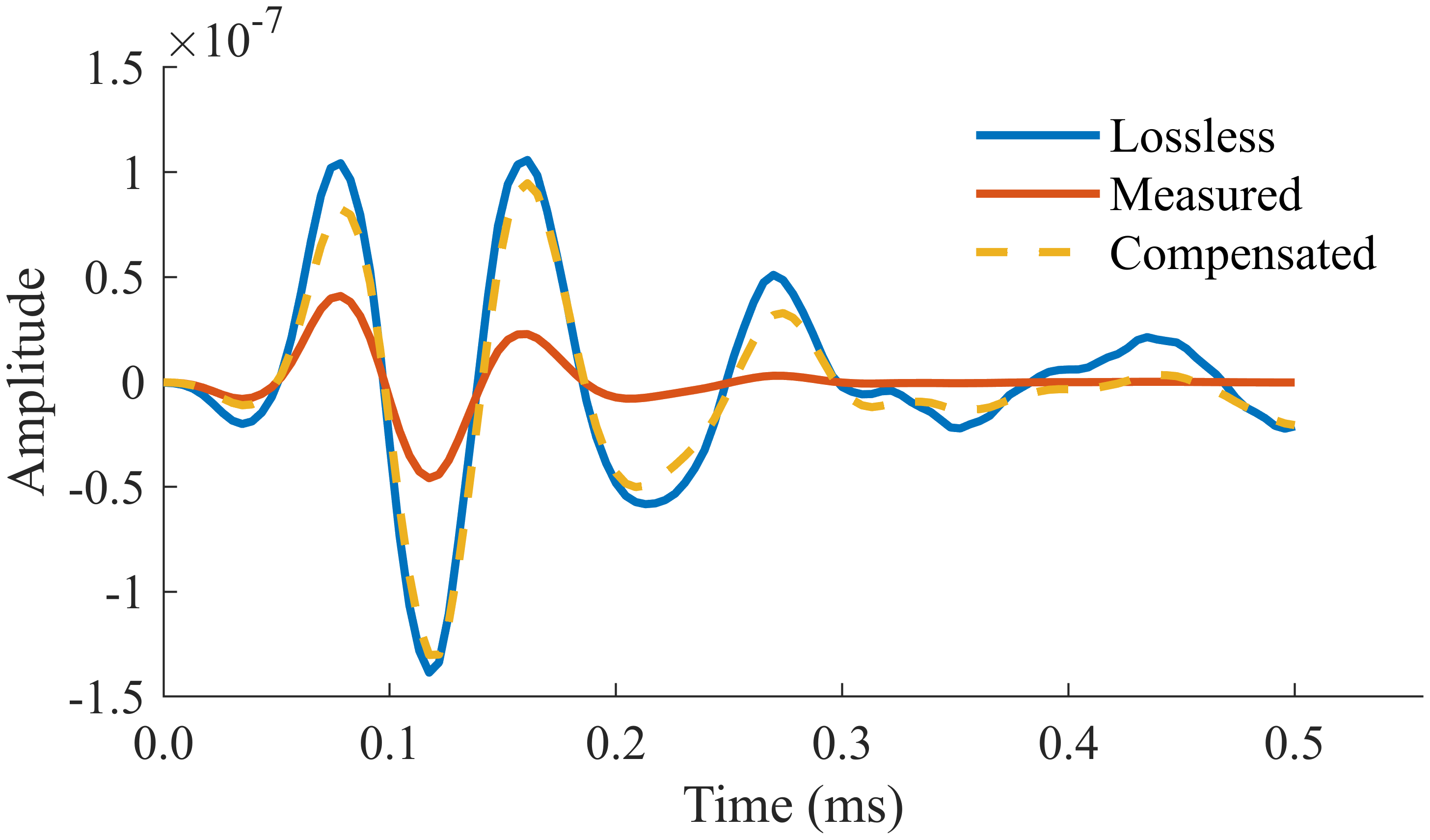}
		\caption{} 
	\end{subfigure}
    \centering
	\begin{subfigure}[b]{43mm}
		\includegraphics[width=\textwidth]{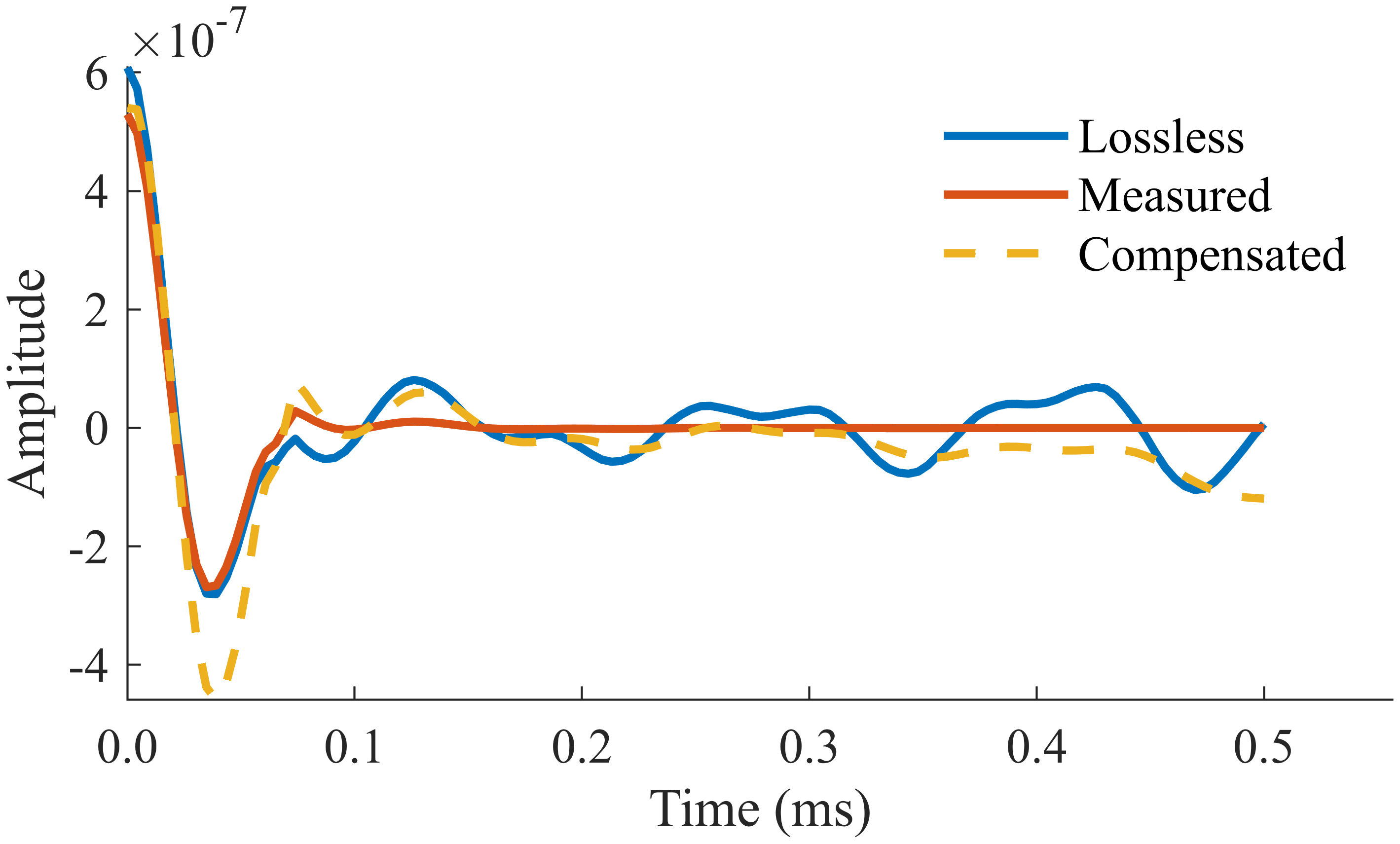}
		\caption{} 
	\end{subfigure}
    \hfill
    \begin{subfigure}[b]{43mm}
		\includegraphics[width=\textwidth]{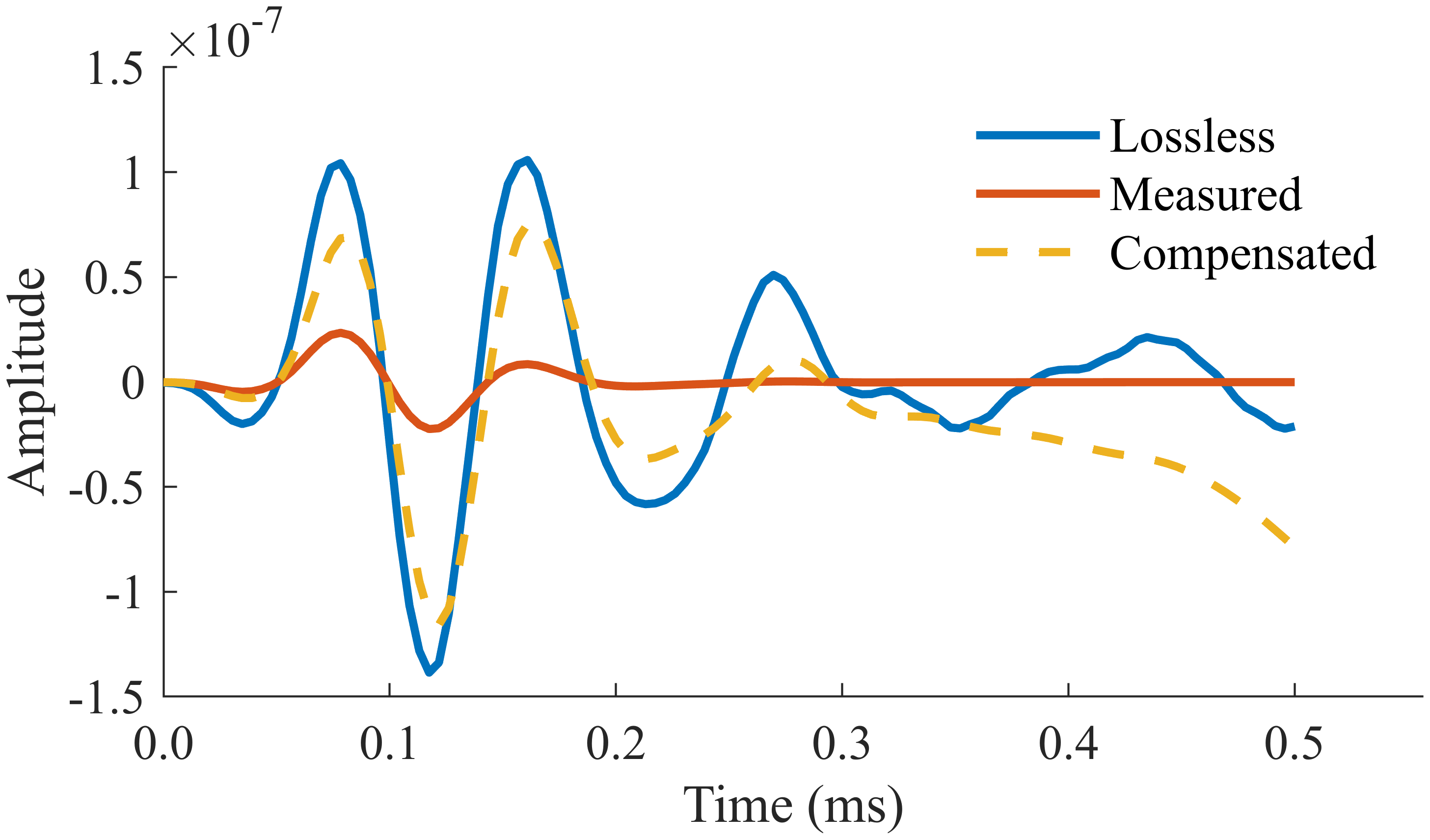}
		\caption{} 
	\end{subfigure}
	\caption{Comparison of lossless, measured, and compensated data for different loss levels. 
  The transmitter is at Sensor 10 and receiver is at Sensor 10 for monostatic data or Sensor 15 for bistatic data.} 
  \label{Data Comparison Boundary Test}
\end{figure}

\begin{figure}[!htb]
  \centering
  \includegraphics[width=0.28\textwidth]{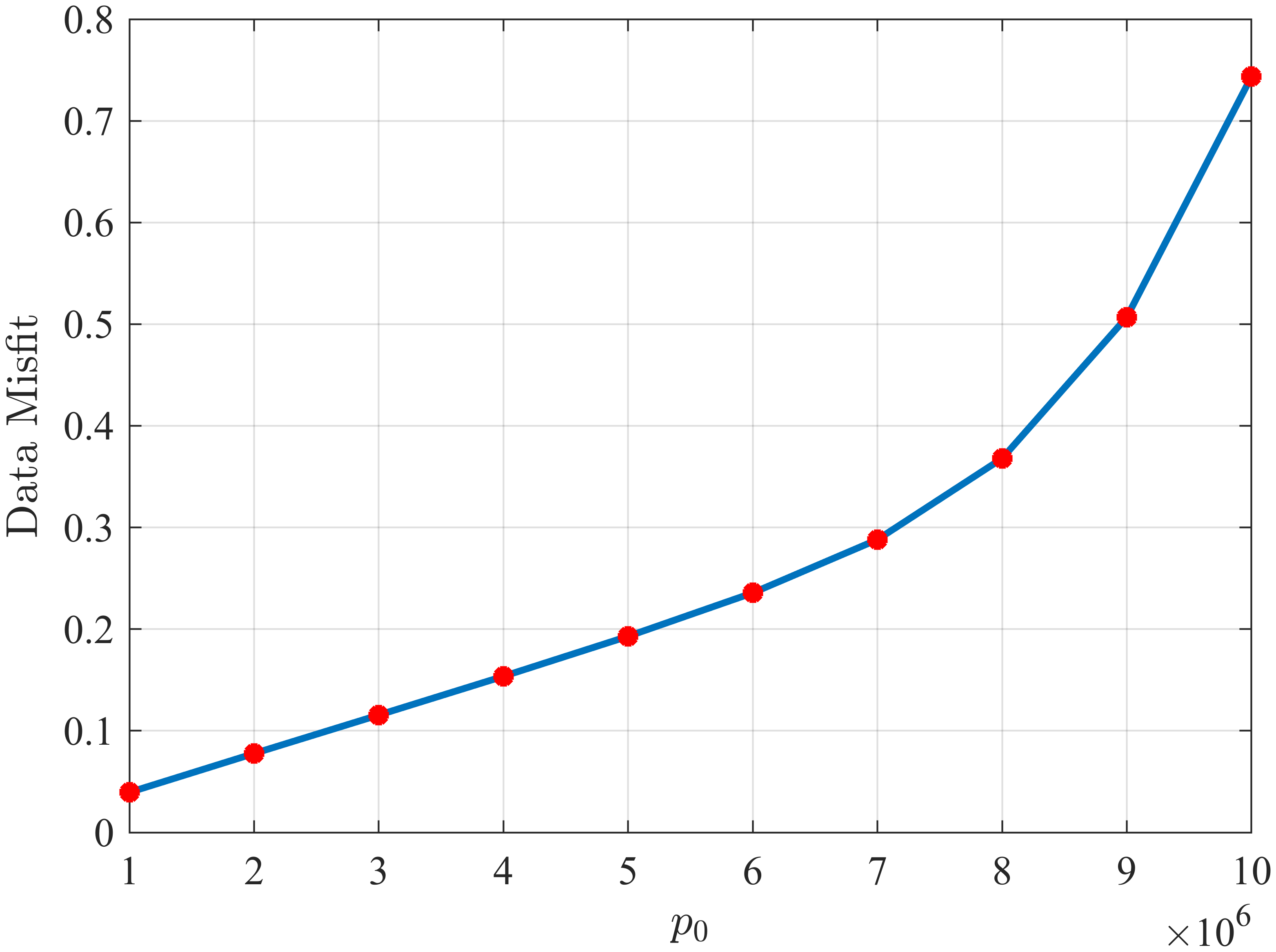}
  \caption{Relative data misfit of compensated data versus the dissipation scaling factor $p_0$.} 
  \label{Compensation Misfit}
\end{figure}

The compensation results for monostatic (Sensor 10) and bistatic (Sensor 15) data as the transmitter is at Sensor 10 is illustrated in Fig.~\ref{Data Comparison Boundary Test}, for three various lossy media with $p_0 = 3\times10^6$ \si{\per\second}, $6\times10^6$ \si{\per\second}, and $10\times10^6$ \si{\per\second}.
The plots illustrate that as medium dissipation intensifies, the waveforms of attenuated measured data become challenging to distinguish. 
Nevertheless, our compensation method can amplify these weak signals. 
For media with moderate loss ($p_0 = 3\times10^6$ \si{\per\second}), the compensated data approximates the lossless data with high fidelity. 
In highly dissipative media ($p_0 = 10\times10^6$ \si{\per\second}), the precision of the compensated data is limited due to the loss of information, while the compensation method can still recover features of the lossless waveform.

The relative data misfit between the compensated data and lossless data for the entire dataset with respect to $p_0$ is demonstrated in Fig.~\ref{Compensation Misfit}.
The error curve indicates that the proposed method achieves high compensation accuracy in moderately dissipative media.
The recovery accuracy decreases as the dissipation becomes more pronounced, attributed to the severe attenuation experienced by the waves in the early stage of the propagation process. 

\section{Conclusion}
% This paper presents a novel inversion approach for lossy media based on wave operator analysis. 
% The proposed method integrates a compensation strategy to estimate lossless data, which is then utilized in the construction of ROMs to extract Born data. 
% By suppressing the nonlinearity of ISPs, the method enhances the performance of the BP algorithm. 
% Numerical experiments validate the effectiveness of the compensation strategy, demonstrating that the compensated data closely approximates the lossless data with appropriately amplified reflected echoes. 
% These results confirm that the proposed method can produce satisfactory outcomes.

In this paper, we addressed the challenge of extending the wave operator model to lossy media and developing an attenuation compensation strategy.
We derive a non-closed-form electric field distribution and decompose it into a closed-form propagation term and a non-closed-form dissipation term.
It is shown that by reversing the dominant exponential decay within the propagation term, the attenuated data can be stably restored to its approximate lossless state.
The accuracy and robustness of this compensation method were validated through numerical experiments. 
This work establishes the necessary theoretical foundation for applying the wave operator framework in lossy environments, paving the way for integration with ROM-based methods to solve inverse scattering problems.
The developed strategy holds promise for enhancing imaging quality in practical applications, such as ground-penetrating radar imaging, seismic exploration and biomedical imaging.

{\appendices

\section{Derivation of the Green's Function}
\label{Appendix Green's Function}

This appendix provides detailed derivations for the Green's function presented in Proposition 1.
The derivation demonstrates that solution \eqref{Green Function Solution 1} is exact under the condition of uniform dissipation and provides an error term for the non-uniform case.

To validate the Green's function \eqref{Green Function Solution 1}, we substitute it into the governing wave equation \eqref{Green Function 1}.
%  of the lossy system \eqref{Green Function 1}, we assume formula \eqref{Green Function Solution 1} holds.
Since the spatial component $\delta_{\boldsymbol{x_s}}(\boldsymbol{x})$ is separable, i.e. $G(\boldsymbol{x}, t) = G(t)\delta_{\boldsymbol{x_s}}(\boldsymbol{x})$, we focus on the temporal component $G(t)$.
% Obviously, $G(t)$ and $\delta_{\boldsymbol{x_s}}(\boldsymbol{x})$ can be separated, i.e. $G(\boldsymbol{x}, t) = G(t)\delta_{\boldsymbol{x_s}}(\boldsymbol{x})$, so we will focus on $G(t)$ in the following.

For brevity, we use $\mathbf{E}$, $\mathbf{C}$, $\mathbf{S}$ and $\mathbf{H}$ to denote $e^{-\frac{p(\boldsymbol{x})t}{2}}$, $\cos(t\sqrt{A'})$, $\sin(t\sqrt{A'})$ and $H(t)$, respectively.
Then we get
\begin{equation}
  G(t) = \mathbf{ECH},
\end{equation}
\begin{equation}
  \partial_t G(t) = -\frac{p}{2}\mathbf{ECH} - \mathbf{E}\sqrt{A'}\mathbf{SH} + \mathbf{EC}\delta(t),
\end{equation}
\begin{equation}
    \begin{aligned}
        \partial_t^2 G(t) =& \frac{p^2}{4}\mathbf{ECH} + \frac{p}{2}\mathbf{E}\sqrt{A'}\mathbf{SH} -\frac{p}{2}\mathbf{EC}\delta(t) \\
        &+ \frac{p}{2}\mathbf{E}\sqrt{A'}\mathbf{SH} - \mathbf{E}A'\mathbf{CH} - \mathbf{E}\sqrt{A'}\mathbf{S}\delta(t) \\
        &- \frac{p}{2}\mathbf{EC}\delta(t) - \mathbf{E}\sqrt{A'}\mathbf{S}\delta(t) + \mathbf{EC}\partial_t\delta(t) \\
        =& \frac{p^2}{4}\mathbf{ECH} - \mathbf{E}A'\mathbf{CH} + p\mathbf{E}\sqrt{A'}\mathbf{SH} - p\mathbf{EC}\delta(t) \\
        &- 2\mathbf{E}\sqrt{A'}\mathbf{S}\delta(t) + \mathbf{EC}\partial_t\delta(t).
    \end{aligned}
\end{equation}
% Note that only if $\sigma(\boldsymbol{x})\propto\epsilon(\boldsymbol{x})$, the calculation of $A'$ and $p$ is commutative.
Note that only if $p(\boldsymbol{x})=\frac{\sigma(\boldsymbol{x})}{\epsilon(\boldsymbol{x})}$ is a constant, the calculation of $A'$ and $p$ is commutative.
    
Then we derive
\begin{equation}
    \begin{aligned}
        (\partial_t^2 + p\partial_t & + A)G(t) \\
        =& \frac{p^2}{4}\mathbf{ECH} - \mathbf{E}A'\mathbf{CH} + p\mathbf{E}\sqrt{A'}\mathbf{SH} \\
        &- p\mathbf{EC}\delta(t)- 2\mathbf{E}\sqrt{A'}\mathbf{S}\delta(t) + \mathbf{EC}\partial_t\delta(t) \\
        &- \frac{p^2}{2}\mathbf{ECH} - p\mathbf{E}\sqrt{A'}\mathbf{SH}+ p\mathbf{EC}\delta(t) + A\mathbf{ECH} \\
        =& A'\mathbf{ECH} - \mathbf{E}A'\mathbf{CH} + \mathbf{EC}\partial_t\delta(t) - 2\mathbf{E}\sqrt{A'}\mathbf{S}\delta(t) \\
        =& \left[A', \mathbf{E}\right]\mathbf{CH} + \mathbf{EC}\partial_t\delta(t) - 2\mathbf{E}\sqrt{A'}\mathbf{S}\delta(t). \\
        % =& \left[A', e^{-\frac{p}{2}t}\right]\cos(t\sqrt{A'})H(t) + e^{-\frac{p}{2}t}\cos(t\sqrt{A'})\partial_t\delta(t) \\
        % &- 2e^{-\frac{p}{2}t}\sin(t\sqrt{A'})\sqrt{A'}\delta(t)
    \end{aligned}
    \label{Green Function Calculate 1}
\end{equation}
The result contains a commutator term $\left[A', \mathbf{E}\right]$, which measures the degree of non-commutativity between $A'$ and $\mathbf{E}$.
The commutator, defined as $[M, N] = MN - NM$, is zero if and only if the operators commute.
$A'$ and $\mathbf{E}$ are commutative only if $p(\boldsymbol{x})=p$ is spatially constant.
% corresponding to a medium where conductivity is proportional to permittivity, $\sigma(\boldsymbol{x})\propto\epsilon(\boldsymbol{x})$.

With uniform dissipation distribution, the commutator term vanishes.
At $t=0$, the Dirac function term $2\mathbf{E}\sqrt{A'}\mathbf{S}\delta(t)$ in \eqref{Green Function Calculate 1} also cancels, confirming that the wave equation is satisfied:
\begin{equation}
      (\partial_t^2 + p\partial_t + A)G(t) = \partial_t\delta(t).
\end{equation}
This validates that \eqref{Green Function Solution 1} is the exact solution of the Green's function in \eqref{Green Function 1} with a uniform distributed dissipation $p(\boldsymbol{x})$.

% For $t<0$,
% \begin{equation}
%     (\partial_t^2 + p\partial_t + A)G(t) = 0
% \end{equation}
% For $t=0$, $\left[A', e^{-\frac{p}{2}t}|_{t=0}\right] = \left[A', I\right] = 0$, then
% \begin{equation}
%     (\partial_t^2 + p\partial_t + A)G(t) = 0
% \end{equation}
% For $t>0$,
% \begin{equation}
%         (\partial_t^2 + p\partial_t + A)G(t) = \left[A', e^{-\frac{p}{2}t}\right]\cos(t\sqrt{A'})H(t)
% \end{equation}

% Therefore, the Green's function of the lossy system is 
% \begin{equation*}
%     G(t) = e^{-\frac{p}{2}t}\cos(t\sqrt{A-\frac{p^2}{4}})H(t)
% \end{equation*}
% with $\sigma(\boldsymbol{x})\propto\epsilon(\boldsymbol{x})$.

Conversely, if the dissipation is non-uniform, the commutator is non-zero and represents an approximation error. 
% We assume $p(\boldsymbol{x})$ is composed of an homogeneous distribution $p_0(\boldsymbol{x}) = \lambda I$ and a variation $\Delta p(\boldsymbol{x})$, i.e., $p(\boldsymbol{x}) = p_0 + \Delta p(\boldsymbol{x})$.
Assume the dissipation parameter $p(\boldsymbol{x})$ can be decomposed as $p(\boldsymbol{x}) = p_0 + \Delta p(\boldsymbol{x})$, where $p_0$ is a constant background and $\Delta p(\boldsymbol{x})$ is a spatial variation. 
Take the first-order Taylor expansion on the exponential term $\mathbf{E}$ to yield
% A first-order Taylor expansion of the exponential term yields:
\begin{equation}
    \mathbf{E} = e^{-\frac{p(\boldsymbol{x})}{2}t} =  e^{-\frac{p_0}{2}t} \cdot  e^{-\frac{\Delta p(\boldsymbol{x})}{2}t} \approx e^{-\frac{p_0}{2}t} \cdot \left(1 - \frac{\Delta p(\boldsymbol{x})}{2}t\right).
\end{equation}
With this approximation, the commutator becomes
\begin{equation}
    \begin{aligned}
        \left[A', e^{-\frac{p(\boldsymbol{x})}{2}t}\right] 
        % \approx& e^{-\frac{p_0}{2}t} \left[A' (1 - \frac{\Delta p(\boldsymbol{x})}{2}t) - (1 - \frac{\Delta p(\boldsymbol{x})}{2}t) A'\right] \\
        \approx& \left[A', e^{-\frac{p_0}{2}t}\left(1 - \frac{\Delta p(\boldsymbol{x})}{2}t\right)\right] \\
        =& e^{-\frac{p_0}{2}t} \left[A', 1 - \frac{\Delta p(\boldsymbol{x})}{2}t \right] \\
        =& \frac{t}{2}e^{-\frac{p_0}{2}t}\left[\Delta p(\boldsymbol{x})A' - A'\Delta p(\boldsymbol{x})\right] \\
        =& \frac{t}{2}e^{-\frac{p_0}{2}t}\left[\Delta p(\boldsymbol{x}), A'\right].
    \end{aligned}
\end{equation}

This result reveals that the approximation error induced by \eqref{Green Function Solution 1} is proportional to the spatial variation of the dissipation $\Delta p(\boldsymbol{x})$, and accumulates linearly with propagation time $t$. 
Therefore, the proposed Green's function \eqref{Green Function Solution 1} provides an accurate approximation when the dissipation is nearly uniform at early time.

% These results demonstrate that for the medium with inhomogeneous $p$ distribution, the Green's function $G(t)$ yields accurate characterization of the lossy system behavior as the variation $\Delta p(\boldsymbol{x})$ decreases.
% Besides, the approximation error accumulates as the wave propagation time $t$ increases.

\section{Approximation of the Dissipation Fraction}
\label{Appendix Lemma}

This appendix provides an intermediate result used in the proof of Proposition 2.
We demonstrate that the dissipation fraction in \eqref{Dissipation Fraction} is dominated by the contribution from the dissipation term $\tilde{E}^{Dis}$.

\newtheorem{lemma}{Lemma}
\begin{lemma}
For small dissipation parameters $0 \le p_1 \le p_2 \ll \|A\|$, the dissipation fraction \eqref{Dissipation Fraction} can be approximated by the ratio of dissipation terms of the electric field:
\begin{equation}
  k(p_1, p_2) \approx \frac{p_1}{p_2} \cdot 
  \frac{\int_\Omega d\boldsymbol{x} \int_{-\infty}^{\infty}\delta^T_{\boldsymbol{x_r}}K(\omega, p_1)\hat{f}(\omega)\delta_{\boldsymbol{x_s}}d\omega}
  {\int_\Omega d\boldsymbol{x} \int_{-\infty}^{\infty}\delta^T_{\boldsymbol{x_r}}K(\omega, p_2)\hat{f}(\omega)\delta_{\boldsymbol{x_s}}d\omega},
  \label{Lemma Dissipation Fraction}
\end{equation}
where
\begin{equation}
  K(\omega, p) = \frac{A + \omega^2}{(\omega^2 - A)^2 + \omega^2p^2}.
  \label{Definition of K}
\end{equation}
\end{lemma}

\begin{proof}

% The objective of this proof is to demonstrate that the dissipation fraction $k(p_1, p_2)$ is determined by the dissipation term $\tilde{E}^{Dis}$ of the electric field.
Without loss of generality, we assume the source $\hat{f}(\omega)$ is normalized.

We define the following two difference terms to convert the dissipation fraction into separate parts:
\begin{equation}
  \Delta \tilde{E}^{Pro}(t, p) = \tilde{E}^{Pro}(t, p) - \tilde{E}^{Pro}(t, 0),
\end{equation}
\begin{equation}
  \begin{aligned}
    \Delta \tilde{E}^{Dis}(t, p) =& \tilde{E}^{Dis}(t, p) - \tilde{E}^{Dis}(t, 0) \\
    =& \frac{p}{2\pi}\int_{-\infty}^{\infty}K(\omega, p)\hat{f}(\omega)\cos(\omega t)\delta_{\boldsymbol{x_s}}(\boldsymbol{x})d\omega,
  \end{aligned}
\end{equation}
thereby the dissipation fraction \eqref{Dissipation Approximation} can be represented as
\begin{equation}
  k(p_1, p_2) = \frac{\int_{\Omega}\delta^T_{\boldsymbol{x_r}}(\boldsymbol{x}) [\Delta \tilde{E}^{Dis}(0, p_1) + \Delta \tilde{E}^{Pro}(0, p_1)] d\boldsymbol{x}} 
  {\int_{\Omega}\delta^T_{\boldsymbol{x_r}}(\boldsymbol{x}) [\Delta \tilde{E}^{Dis}(0, p_2) + \Delta \tilde{E}^{Pro}(0, p_2)] d\boldsymbol{x}}.
  \label{Dissipation Approximation 2}
\end{equation}

First, consider the propagation term difference at $t=0$:
\begin{equation}
  \begin{aligned}
    \Delta E&^{Pro}(0, p) \\
    % =& \left\{\frac{1}{2}\left[\hat{f}(\sqrt{A'}+\frac{jp}{2}) + \hat{f}(\sqrt{A'}-\frac{jp}{2})\right] - \hat{f}(\sqrt{A})\right\}\delta_{\boldsymbol{x_s}}(\boldsymbol{x}) \\
    =& \frac{1}{2}\left[\hat{f}(\sqrt{A-\frac{p^2}{4}}+\frac{jp}{2}) + \hat{f}(\sqrt{A-\frac{p^2}{4}}-\frac{jp}{2})\right]\delta_{\boldsymbol{x_s}}(\boldsymbol{x}) \\
    & -\hat{f}(\sqrt{A})\delta_{\boldsymbol{x_s}}(\boldsymbol{x}).
  \end{aligned}
\end{equation}
Expanding the arguments of $\hat{f}$ around $\sqrt{A}$ by Taylor expansion yields
\begin{equation}
  \begin{aligned}
    \hat{f}(\sqrt{A-\frac{p^2}{4}}+\frac{jp}{2}) \approx& \hat{f}(\sqrt{A}-\frac{p^2}{8\sqrt{A}}+\frac{jp}{2}) \\
    \approx& \hat{f}(\sqrt{A}) - \hat{f}'(\sqrt{A})\cdot(\frac{p^2}{8\sqrt{A}} - j\frac{p}{2}),
  \end{aligned}
\end{equation}
\begin{equation}
  \begin{aligned}
    \hat{f}(\sqrt{A-\frac{p^2}{4}}-\frac{jp}{2}) \approx& \hat{f}(\sqrt{A}-\frac{p^2}{8\sqrt{A}}-\frac{jp}{2}) \\
    \approx& \hat{f}(\sqrt{A}) - \hat{f}'(\sqrt{A})\cdot(\frac{p^2}{8\sqrt{A}} + j\frac{p}{2}).
  \end{aligned}
\end{equation}
Therefore, imaginary parts of the first-order term in $\Delta E^{Pro}(0, p)$ cancel out, and this difference becomes
\begin{equation}
  \begin{aligned}
      \Delta E^{Pro}(0, p) =& \left(-\frac{p^2}{8\sqrt{A}}\hat{f}'(\sqrt{A}) - \frac{p^2}{8}\hat{f}''(\sqrt{A})\right)\delta_{\boldsymbol{x_s}}(\boldsymbol{x}). \\
      =& O(p^2).
  \end{aligned}
  \label{Approx Propagation Term}
\end{equation}

Next, we examine the dissipation term at $t=0$:
\begin{equation}
  \Delta \tilde{E}^{Dis}(0, p) = \frac{p}{2\pi}\int_{-\infty}^{\infty}K(\omega, p)\hat{f}(\omega)\delta_{\boldsymbol{x_s}}(\boldsymbol{x})d\omega = O(p).
  \label{Approx Dissipation Term}
\end{equation}

Comparing the two components, the contribution from the propagation term is of order $O(p^2)$, while the contribution from the dissipation term is of order $O(p)$. 
For a sufficiently small dissipation parameter $p$, the $O(p)$ term is dominant. 
Therefore, the term $D(0, p) - D(0, 0)$ is well approximated by the contribution from the dissipation term:
\begin{equation}
  D(0, p) - D(0, 0) \approx \int_{\Omega}\delta^T_{\boldsymbol{x_r}} \tilde{E}^{Dis}(0, p) d\boldsymbol{x}.
\end{equation}
Applying this approximation to both the numerator and the denominator of the dissipation fraction \eqref{Dissipation Fraction} yields
\begin{equation}
  \begin{aligned}
    k(p_1, p_2) \approx& \frac{\int_{\Omega}\delta^T_{\boldsymbol{x_r}} \tilde{E}^{Dis}(0, p_1) d\boldsymbol{x}}
    {\int_{\Omega}\delta^T_{\boldsymbol{x_r}} \tilde{E}^{Dis}(0, p_2) d\boldsymbol{x}}\\
    =& \frac{p_1}{p_2} \cdot 
    \frac{\int_\Omega d\boldsymbol{x} \int_{-\infty}^{\infty}\delta^T_{\boldsymbol{x_r}}K(\omega, p_1)\hat{f}(\omega)\delta_{\boldsymbol{x_s}}d\omega}
    {\int_\Omega d\boldsymbol{x} \int_{-\infty}^{\infty}\delta^T_{\boldsymbol{x_r}}K(\omega, p_2)\hat{f}(\omega)\delta_{\boldsymbol{x_s}}d\omega}.
  \end{aligned}
\end{equation}
This completes the proof of \eqref{Lemma Dissipation Fraction}.
\end{proof}

\section{Proof of Error Bound of the Dissipation Fraction}
\label{Appendix Dissipation Approximation}

% To prove the error bound of dissipation fraction \eqref{Dissipation Approximation}, we first introduce the following lemma:
This appendix provides the proof of the error bound of the dissipation fraction \eqref{Dissipation Approximation} in Proposition 2.
The proof leverages the result from Lemma 1 \eqref{Lemma Dissipation Fraction}.
Without loss of generality, we assume the source $\hat{f}(\omega)$ is normalized.
Besides, since the source $f(t)$ is a Gaussian-type signal, $\hat{f}(\omega)$ is monotonically decreasing in $[0, \infty)$.

Based on \eqref{Lemma Dissipation Fraction}, proving Proposition 2 is equivalent to show that the rational fraction $r(p_1, p_2)$ defined as
\begin{equation}
  r(p_1, p_2) = \frac{I(p_1)}{I(p_2)},
  \label{Rational Fraction}
\end{equation}
is close to unity, i.e.
\begin{equation}
  r(p_1, p_2) = 1 + O\left(\frac{p_2^2 - p_1^2}{\|A\|}\right),
\end{equation}
where the scalar integral is defined as
\begin{equation}
  I(p) = \int_{-\infty}^{\infty}K(\omega, p)\hat{f}(\omega)d\omega.
\end{equation}

% With Lemma 1 and the rational fraction defined as:
% \begin{equation}
%   r(p_1, p_2) = \frac{\delta_{\boldsymbol{x_r}}^T(\boldsymbol{x}) I(p_1) \delta_{\boldsymbol{x_s}}(\boldsymbol{x})}{\delta_{\boldsymbol{x_r}}^T(\boldsymbol{x}) I(p_2) \delta_{\boldsymbol{x_s}}(\boldsymbol{x})},
% \end{equation}
% we only need to prove that
% \begin{equation}
%   r(p_1, p_2) = 1 + O\left(\frac{p_2^2 - p_1^2}{A}\right),
% \end{equation}
% where
% \begin{equation}
%   I(p_1) = \int_{-\infty}^{\infty}K(\omega, p)\hat{f}(\omega)d\omega.
% \end{equation}

We begin by analyzing a discrepancy term
\begin{equation}
    \begin{aligned}
        \Delta I =& I(p_1) - I(p_2) \\
        =& \int_{-\infty}^{\infty}[K(\omega, p_1) - K(\omega, p_2)]\hat{f}(\omega)d\omega \\
        % =& \int_{-\infty}^{\infty}\frac{(p_2^2-p_1^2)\omega^2(A+\omega^2)}{M(\omega, p_1, p_2)}\hat{f}(\omega)d\omega \\
        =& 2(p_2^2-p_1^2)\int_{0}^{\infty}\frac{\omega^2(A+\omega^2)}{M(\omega, p_1, p_2)}\hat{f}(\omega)d\omega, \\
        % =& \Delta I_N + \Delta I_\infty,
    \end{aligned}
\end{equation}
where
\begin{equation}
  M(\omega, p_1, p_2) = \left[(A-\omega^2)^2+\omega^2p_1\right]\left[(A-\omega^2)^2+\omega^2p_2\right].
\end{equation}
To estimate the order of this integral, we separate the integration domain into two parts:
\begin{equation}
    \begin{aligned}
        \Delta I_N =& 2(p_2^2-p_1^2)\int_{0}^{N}\frac{\omega^2(A+\omega^2)}{M(\omega, p_1, p_2)}\hat{f}(\omega)d\omega, \\
        \Delta I_\infty =& 2(p_2^2-p_1^2)\int_{N}^{\infty}\frac{\omega^2(A+\omega^2)}{M(\omega, p_1, p_2)}\hat{f}(\omega)d\omega. \\
    \end{aligned}
\end{equation}
Here $N$ is a frequency cutoff ensuring $\frac{\omega^2(A+\omega^2)}{M(\omega,p_1,p_2)}$ is monotonically decreasing in $[N, \infty)$.
Then for the high-frequency part, since
\begin{equation}
  \begin{aligned}
    M(\omega, p_1, p_2) \ge& \left[(A-\omega^2)^2+\omega^2p_1\right]^2 \\
    =& \left[\omega^4 - (2A - p_1)\omega^2 +A^2\right]^2 \\
    >& (\omega^4 + A^2)^2,
  \end{aligned}
\end{equation}
we have
\begin{equation}
    \begin{aligned}
        \Delta I_\infty 
        % =& 2(p_2^2-p_1^2)\int_{N}^{\infty}\frac{\omega^2(A+\omega^2)}{M(\omega, p_1, p_2)}\hat{f}(\omega)d\omega \\
        <& 2(p_2^2-p_1^2)\int_{N}^{\infty}\frac{\omega^2(A+\omega^2)}{(\omega^4 + A^2)}d\omega \\
        <& 2(p_2^2-p_1^2)\int_{N}^{\infty}\frac{\omega^2(A+\omega^2)}{\omega^8}d\omega \\
        =& 2(p_2^2-p_1^2)\left[\left.\frac{A}{5\omega^5}\right|_{\infty}^N + \left.\frac{1}{3\omega^3}\right|_{\infty}^N\right] \\
        =& 2(p_2^2-p_1^2)\cdot\left(\frac{A}{5N^5} + \frac{1}{3N^3}\right).
    \end{aligned}
\end{equation}
For the low-frequency part, we can establish the following bound
\begin{equation}
    \begin{aligned}
        \Delta I_N
        % =& 2(p_2^2-p_1^2)\int_{0}^{N}\frac{\omega^2(A+\omega^2)}{M(\omega, p_1, p_2)}\hat{f}(\omega)d\omega \\
        <& 2(p_2^2-p_1^2)\int_{0}^{N}\frac{\omega^2(A + \omega^2)}{(A - \omega^2)^4}d\omega \\
        <& 2(p_2^2-p_1^2)\int_{0}^{N}\frac{\omega^2(A + N^2)}{(A - N^2)^4}d\omega \\
        =& 2(p_2^2-p_1^2)\frac{N^3(A + N^2)}{3(A - N^2)^4}.
        % =& 2(p_2^2-p_1^2)\left[\int_{0}^{N}\frac{\omega^2A + \omega^4}{(\omega^4+?A^2)^2}d\omega + \int_{0}^{N}\frac{\omega^4}{(\omega^4+A^2)^2}d\omega\right] \\
        % <& 2(p_2^2-p_1^2)\left[\int_{0}^{N}\frac{A\omega^2}{(A^4)^2}d\omega + \int_{0}^{N}\frac{1}{\omega^4+\frac{A^4}{\omega^4}+2A^2}d\omega\right] \\
        % <& 2(p_2^2-p_1^2)\cdot\left(\frac{N^3}{3A^3} + \frac{N}{4A^2}\right)
    \end{aligned}
\end{equation}
This bound is controlled by $N$.

Since $0 \le p_1 \le p_2 \ll \|A\|$, we choose $N=2\sqrt{A}$.
This choice simplifies the bounds for the high-frequency part:
\begin{equation}
  \begin{aligned}
      \Delta I_\infty <& 2(p_2^2-p_1^2) \cdot \left(\frac{A}{5(2\sqrt{A})^5} + \frac{1}{3(2\sqrt{A})^3}\right) \\
      =& (p_2^2-p_1^2) \cdot O(\|A\|^{-3/2}),
  \end{aligned}
\end{equation}
and for the low-frequency part:
\begin{equation}
  \begin{aligned}
    \Delta I_N <& 2(p_2^2-p_1^2)\frac{(2\sqrt{A})^3\left[A + (2\sqrt{A})^2\right]}{3\left[A - (2\sqrt{A})^2\right]^4} \\
    =& (p_2^2-p_1^2) \cdot O(\|A\|^{-3/2}).
  \end{aligned}
\end{equation}
% Combining both parts, the integral is of order $O(p_2^2-p_1^2) \cdot O(\|A\|^{-3/2})$.
Combining both parts, the integral is of order $O\left(\frac{p_2^2-p_1^2}{\|A\|^{3/2}}\right)$.
Therefore, the order of the discrepancy term is
\begin{equation}
  \Delta I = O\left(\frac{p_2^2-p_1^2}{\|A\|^{3/2}}\right).
  \label{Upper Bound for Delta I}
\end{equation}

% Here, $M(\omega, p_1, p_2)$ is a polynomial defined as
% \begin{equation}
%     \begin{aligned}
%         M(\omega, p_1, p_2) =& \left[(A-\omega^2)^2+\omega^2p_1\right]\left[(A-\omega^2)^2+\omega^2p_2\right] \\
%         >& \left[(A-\omega^2)^2+\omega^2p_1\right]^2 \\
%         \overset{t=\omega^2}{=}&\left[(t-A)^2+tp_1^2\right] \\
%         =& \left[t^2-(2A-p_1^2)t+A^2\right]^2 \\
%         >& (t^2+A^2)^2
%     \end{aligned}
% \end{equation}
% Besides, $N$ is a select value to ensure $\frac{\omega^2(A+\omega^2)}{M(\omega,p_1,p_2)}$ is monotonically decreasing in $[N, \infty)$.

% Then we will estimate the upper bounds for these two integral.

% Therefore, we get an approximated upper bound
% \begin{equation}
%   \begin{aligned}
%         \Delta I =& \Delta I_N + \Delta I_\infty \\
%         <& 2(p_2^2-p_1^2)\cdot\left(\frac{A}{5N^5} + \frac{1}{3N^3} + \frac{N^3}{3A^3} + \frac{N}{4A^2}\right)
%   \end{aligned}
% \end{equation}
% This bound is controlled by $N$.
% Here, since $p_{1,2}\ll A$, we set $N=2\sqrt{A}$, then this formulation can be calculated as
% \begin{equation}
%     \begin{aligned}
%         \Delta I <& 2(p_2^2-p_1^2)\cdot\left(\frac{1}{160A^{3/2}} + \frac{1}{24A^{3/2}} + \frac{8}{3A^{3/2}} + \frac{1}{2A^{3/2}}\right) \\
%         % =& 2(p_2^2-p_1^2)\cdot O(\frac{1}{A^{3/2}})
%         =& O\left(\frac{p_2^2-p_1^2}{A^{3/2}}\right)
%     \end{aligned}
% \end{equation}

Next, we estimate the order of the integral $I(p_1)$.
The integrando of $I(p_1)$ has a peak around the resonance $\omega^2 \approx A$, then we cam obtain a lower bound by truncating the integration domain to $[0, \sqrt{A}]$:
\begin{equation}
    \begin{aligned}
        I(p_1) =& 2\int_{0}^{\infty}\frac{A + \omega^2}{(A-\omega^2)^2+\omega^2p_1^2}\hat{f}(\omega)d\omega \\
        % =& 2\left[\int_{0}^{\sqrt{A}}\frac{A + \omega^2}{(A-\omega^2)^2+\omega^2p_1^2}\hat{f}(\omega)d\omega \right. \\ 
        % &+ \left. \int_{\sqrt{A}}^{\infty}\frac{A + \omega^2}{(A-\omega^2)^2+\omega^2p_1^2}\hat{f}(\omega)d\omega\right] \\
        >& 2\int_{0}^{\sqrt{A}}\frac{A + \omega^2}{(A+\omega^2)^2+\omega^2p_1^2}\hat{f}(\omega)d\omega \\
        =& 2\int_{0}^{\sqrt{A}}\frac{1}{A+\omega^2+\frac{\omega^2}{A+\omega^2}p_1^2}\hat{f}(\omega)d\omega. \\
        % >& 2\int_{0}^{\sqrt{A}}\frac{1}{A+\omega^2+\frac{Ap_1^2}{A+A}}\hat{f}(\sqrt{A})d\omega \\
    \end{aligned}
\end{equation}
Since $\frac{\omega^2}{A+\omega^2}<1$ and $\hat{f}(\omega)$ is monotonically decreasing in $[0, \sqrt{A}]$, a lower bound is derived
\begin{equation}
  \begin{aligned}
    I(p_1) >& 2\hat{f}(\sqrt{A})\int_{0}^{\sqrt{A}}\frac{1}{\omega^2+A+p_1^2}d\omega \\
        >& 2\hat{f}(\sqrt{A})\int_{0}^{\sqrt{A}}\frac{1}{2A+p_1^2}d\omega \\
        =& 2\hat{f}(\sqrt{A})\frac{\sqrt{A}}{2A+p_1^2} = O\left(\frac{1}{\|A\|^{1/2}}\right).
  \end{aligned}
  \label{Lower Bound for I}
\end{equation}

Finally, we combine the bounds from \eqref{Upper Bound for Delta I} and \eqref{Lower Bound for I} to derive the order of the relative difference:
\begin{equation}
    \frac{\Delta I}{I(p_1)} = \frac{O\left(\frac{p_2^2-p_1^2}{\|A\|^{3/2}}\right)}{O\left(\frac{1}{\|A\|^{1/2}}\right)} = O\left(\frac{p_2^2-p_1^2}{\|A\|}\right).
\end{equation}
This implies the order of the rational fraction \eqref{Rational Fraction} is
\begin{equation}
  r(p_1, p_2) = 1 + \frac{\Delta I}{I(p_2)} = 1 + O\left(\frac{p_2^2-p_1^2}{\|A\|}\right).
\end{equation}
Therefore, the proof of \eqref{Dissipation Approximation} is completed.

% From above we obtained that:
% \begin{equation}
%     \Delta I < O\left(\frac{p_2^2-p_1^2}{A^{3/2}}\right),
% \end{equation}
% \begin{equation}
%     I(p_1) > O\left(\frac{1}{A^{1/2}}\right).
% \end{equation}
% Then we get

% Therefore, the proposition is proved:
% \begin{equation}
%     r(p_1, p_2) = \frac{\delta^T_{\boldsymbol{x_r}}(\boldsymbol{x})}{\delta^T_{\boldsymbol{x_r}}(\boldsymbol{x})}\cdot\frac{I(p_1)}{I(p_2)}\cdot\frac{\delta_{\boldsymbol{x_s}}(\boldsymbol{x})}{\delta_{\boldsymbol{x_s}}(\boldsymbol{x})} = 1 + O\left(\frac{p_2^2-p_1^2}{A}\right)
% \end{equation}

}

\bibliographystyle{ieeetr}
\bibliography{Ref}

\end{document}